\documentclass[12pt]{article}
\DeclareUnicodeCharacter{22C6}{\ensuremath{\star}}
\usepackage[left=1.0in,right=1.0in,top=1.0in,bottom=1.0in]{geometry}
\usepackage{setspace}
\usepackage{fancyhdr}
\usepackage{lmodern}
\usepackage{microtype}
\usepackage[normalem]{ulem}
\usepackage{color}
\usepackage[T1]{fontenc}

\usepackage{titlesec}
\usepackage[title]{appendix}

\usepackage{amssymb, amsmath, amsfonts, mathtools}
\usepackage{breqn}
\usepackage{dsfont}

\usepackage{array, booktabs, makecell, cellspace}
\usepackage{siunitx}
\usepackage{threeparttable, rotating, tabularx}
\usepackage[linesnumbered,ruled,vlined]{algorithm2e}
\usepackage{algorithmic}
\usepackage{graphicx, caption, subcaption}
\usepackage{pdflscape, tikz}
\usetikzlibrary{arrows.meta, positioning, shapes, matrix, calc}
\usepackage{float}

\usepackage{enumitem}
\usepackage{chngcntr}

\usepackage{xurl}
\usepackage[backend=biber,
            style=authoryear,
            maxcitenames=3,
            mincitenames=1,
            maxbibnames=99,
            giveninits=true,
            uniquename=init,
            uniquelist=true,
            dashed=false,
            uniquename=false,
            urldate=long,
            url=true,
            natbib=true]{biblatex}

\usepackage{xcolor}
\definecolor{citationcolor}{RGB}{186, 12, 47}

\DeclareCiteCommand{\cite}
  {\usebibmacro{prenote}}
  {\usebibmacro{citeindex}%
   \printtext[bibhyperref]{\color{citationcolor}\usebibmacro{cite}}}
  {\multicitedelim}
  {\usebibmacro{postnote}}

\DeclareCiteCommand{\parencite}[\mkbibparens]
  {\usebibmacro{prenote}}
  {\usebibmacro{citeindex}%
   \printtext[bibhyperref]{\color{citationcolor}\usebibmacro{cite}}}
  {\multicitedelim}
  {\usebibmacro{postnote}}

\AtBeginDocument{%
  \renewbibmacro*{name:andothers}{%
    \ifboolexpr{
      test {\ifnumequal{\value{listcount}}{\value{liststop}}}
      and
      test \ifmorenames
    }
      {\ifnumgreater{\value{liststop}}{3}
        {\finalandcomma\addspace\bibstring{et\space al\adddot}}
        {\finalandcomma\addspace\bibstring{andothers}}}
      {}}%
}

\usepackage{eurosym, comment, footmisc, chngcntr, endnotes, afterpage, multicol, float}

\DeclareMathOperator*{\argmax}{arg\,max}
\DeclareMathOperator*{\argmin}{arg\,min}

\newtheorem{theorem}{Theorem}
\newtheorem{corollary}{Corollary}
\newtheorem{lemma}{Lemma}

\newtheorem{definition}{Definition}
\newenvironment{proof}[1][Proof]{\noindent\textbf{#1.} }{\ \rule{0.5em}{0.5em}}

\newcolumntype{L}[1]{>{\raggedright\let\newline\\arraybackslash\hspace{0pt}}m{#1}}
\newcolumntype{C}[1]{>{\centering\let\newline\\arraybackslash\hspace{0pt}}m{#1}}
\newcolumntype{R}[1]{>{\raggedleft\let\newline\\arraybackslash\hspace{0pt}}m{#1}}

\usepackage[hidelinks, breaklinks, colorlinks=true, linkcolor=citationcolor, citecolor=citationcolor, urlcolor=citationcolor]{hyperref}

\makeatletter
\AtBeginDocument{%
  \toggletrue{blx@footnote}%
  \renewrobustcmd{\mkbibfootnote}[1]{%
    \iftoggle{blx@footnote}
      {\footnotetext{#1}}
      {\begin{quote}#1\end{quote}}}}
\makeatother

\usepackage{microtype}
\usepackage[breakall]{truncate}

\begin{document}

\begin{titlepage}
\title{Gender Differences in Comparative Advantage Matches: Evidence from Linked Employer-Employee Data\thanks{This paper previous title was ``Assortative Matching and the Gender Wage Gap''. The author is grateful to Ian Schmutte, Gregorio Caetano, and Brantly Callaway for valuable feedback and discussions. This paper benefited from comments received at the Midwest Economic Association Annual Meeting, Econometric Society Summer Meeting, and Western Economic Association International Conference. The views expressed are those of the author and do not necessarily reflect those of any affiliated institutions. Any remaining errors are the author's own.}}
\author{Hugo Sant'Anna \footnote{University of Georgia, hsantanna@uga.edu}}
\date{\today}
\maketitle
\begin{abstract} 
In this paper, I introduce a novel decomposition method based on Gaussian mixtures and k-Means clustering, applied to a large Brazilian administrative dataset, to analyze the gender wage gap through the lens of worker-firm interactions shaped by comparative advantage. These interactions generate wage levels in logs that exceed the simple sum of worker and firm components, making them challenging for traditional linear models to capture effectively. I find that these ``complementarity effects'' account for approximately 17\% of the gender wage gap. Larger firms, high human capital, STEM degrees, and managerial roles are closely related to it. For instance, among managerial occupations, the match effect goes as high as one-third of the total gap. I also find women are less likely to be employed by firms offering higher returns to both human capital and firm-specific premiums, resulting in a significantly larger firm contribution to the gender wage gap than previously estimated. Combined, these factors explain nearly half of the overall gender wage gap, suggesting the importance of understanding firm-worker matches in addressing gender-based pay disparities.\\

\vspace{0in}
\vspace{0in}

\bigskip
\end{abstract}
\setcounter{page}{0}
\thispagestyle{empty}
\end{titlepage}
\pagebreak \newpage

\doublespacing

\section{Introduction}

Human capital has been a crucial component in narrowing wage disparities between men and women. In the United States and similar developed and emerging markets, women are more likely than men to hold college degrees \parencite{oecd_oecd_2024}. However, there is still substantial gap, and several studies attempted to tackle the issue under different perspectives. Differences in productivity has been argued by \textcite{mulligan_selection_2008}. There is also evidence that non-monetary factors, such as preferences for flexible hours, play a major role in generating the gap \parencite{goldin_grand_2014}.

A recent strand of the literature focus primarily on the contribution of firm-specific pay policies that would be important in creating differences across genders. Building on a simple rent-sharing model, these papers break down the wage into two main components: a worker component, solely generated by human capital levels and other worker characteristics, and a firm component, arising from firm heterogeneity such as economic activity, market power, and size. Leveraging from administrative data, \textcite{card_bargaining_2016} (CCK) introduced a Kitagawa-Oaxaca-Blinder (KOB) decomposition \parencite{kitagawa_components_1955, oaxaca_male-female_1973, blinder_wage_1973} to measure the contribution of firm effects on the gender wage gap, finding that around a fifth of the gap arises from firm premiums. Even though analyses based on rent-sharing models are effective in providing a comprehensive overview of the impact of firm-specific pay premiums on the gender wage gap, these models typically assume that the value of worker characteristics remains constant across firms and vice-versa. Therefore, this ``additive separability'' assumption is restrictive as it constrains the ability to capture particular worker-firm interactions due to the rank condition. This precludes scenarios where comparative advantages\footnote{In this paper I use ``comparative advantage'' effects and ``complementarity effects'' interchangeably.} arise from specific worker-firm matches. These models may fail to account for an important source of wage variation when different classes of firms perceive similar workers differently, or when particular matches in the labor market are advantageous to certain workers, which could significantly contribute to explaining gender wage disparities. 

In this paper, I provide the first comprehensive analysis of worker-firm interaction effects on the gender wage gap, explicitly accounting for assortative matching in the labor market. Extending the two-step distributional framework of \textcite{bonhomme_distributional_2019}, I apply k-means clustering and a Gaussian mixture model to the log-hourly wage distribution from massive linked employer-employee data in Brazil. The data provides the universe of formal workers and firms, with a rich set of variables, such as extremely detailed economic activity and occupation codes, gender, race, education level, firm location and more.

My innovative approach groups workers and firms into ``types'' and ``classes'' respectively, reducing dimensionality to satisfy the rank condition necessary to explore worker-firm match effects in the wage structure. The model assumes that each group represents categories of workers and firms that are comparable and, when interacted, generate wages by drawing from a Gaussian distribution where parameters are specific to that match. This methodology allows for a wage generation process that deviate from the restrictive additive separable framework, enabling the identification of wage effects that arise solely from specific worker-firm interactions.

For firms, I employ k-means clustering to group them under similar payment distributions. To determine the optimal number of clusters, I utilize a gap statistics analysis \parencite{tibshirani_estimating_2001}, which identifies the point at which within-cluster variance is minimized. For workers, I model the probability density function (PDF) of wages within each firm class as a mixture of log-normal distributions. I demonstrate the robustness of my results across alternative specifications of these combinations. Furthermore, I provide evidence that estimated clusters can be mapped to observable characteristics, validating their economic significance. Since the identification of my model relies on job movers, I test the exogenous mobility assumption by showing job movement is not related to unobservables.

My model is flexible enough to allow for the identification of differential firm valuations of workers with similar unobserved characteristics. Through Monte Carlo simulations, I identify three distinct channels that contribute to gender wage disparities.

My key contribution is the identification of a ``match effect''. This component captures the wage effect of specific worker-firm interactions, revealing complementarities that arise when particular worker types are matched with certain firm classes. I simulate a labor market with no complementarity, therefore absent of match effects, and compare it with baseline estimates. While typical separable models struggle to identify this component, I find that women are less likely to benefit from positive complementarity effects in wages. Quantitatively, in a counterfactual world without complementarities, women's average log hourly wages increase by one log point (from 2.10 to 2.11), while men's decrease by three log points (from 2.33 to 2.30). Hence, transitioning from the observed labor market with complementarities to a simulated market without reduces the gender wage gap in log hourly wages from 0.24 to 0.20, a decrease of approximately 17 percent. This result suggests that female workers are more likely to be found in disadvantageous worker-firm interactions that yield negative complementarities in wages and even when they are present in interactions that yield positive complementarities, women tend to benefit less than men.

My results also indicate that the complementarity effect grows with higher levels of human capital and the complexity of occupations. These contributions increase with both education and age. Notably, workers in occupations typically associated with the hospitality industry, such as cleaners and waiters, show no evidence of complementarity effects. However, for individuals in occupations requiring STEM degrees, such as engineers and economists, complementarity contribution becomes positive. For managers, it accounts for as much as one-third of the gender wage gap.

In the spirit of \textcite{card_bargaining_2016}, I also explore the overall contribution of firms to the gender wage gap by assuming the labor market is under assortative matching\footnote{For this paper, I follow \parencite{becker_theory_1973} to consider assortative matching as the propensity of high quality firms to match with high quality workers.} and firms evaluate worker characteristics on top of offering premiums. To perform this analysis, I hold constant the distribution of worker clusters across men and women, measuring the counterfactual wage difference when the distributions of firms and their expected payments vary. This approach reveals two components, in addition to the match component, that mirror the established literature: sorting (women's under-representation at higher-paying firms) and bargaining (equally productive women receiving a smaller share of payments). Additive separable models potentially underestimate these effects since they assume the value of worker characteristics is constant across the labor market, thereby imposing a downward bias to the impact of assortative matching in generating the gender wage gap.

The sorting component, representing the contribution of differences in firm allocations in the labor market, accounts for approximately 37.5 percent of the 24 log point gender wage gap in Brazil. This sorting effect is substantially larger than estimates obtained from additive separable models, about 9 percent. The enhanced magnitude stems from my model's ability to capture heterogeneous firm-specific returns to worker characteristics. Specifically, it reveals that women are disproportionately concentrated in labor market segments where firms offer lower returns to worker and firm characteristics for all workers, regardless of gender. This component is less relevant for young individuals, but increases considerably for older and highly educated individuals, reaching about 40 percent to college graduates.

The final component, representing the contribution of differential payments to similarly productive men and women, accounts for approximately 8.3 percent of the gender wage gap. This ``bargaining'' effect suggests that even when women overcome sorting barriers, they still face wage disadvantages within firms. Collectively, the ``match,'' ``sorting,'' and ``bargaining'' components explain more than sixty percent of the observed gap, suggesting that understanding assortative matching in the labor market is essential to mitigate wage disparities.

This paper belongs to the applied literature investigating the channels generating the gender wage gap. While the gap narrowed in recent decades \parencite{blau_womens_2008} due to increases in female human capital \parencite{goldin_homecoming_2006, black_gender_2008, ceci_women_2014}, \textcite{goldin_grand_2014} finds that women's preference for flexible hours over monetary compensation is a relevant factor to narrow the remaining gap.  \textcite{bertrand_breaking_2019} shows that women are disproportionately underrepresented in jobs with high returns on human capital investment.

The more recent strand in the literature investigates the contribution of firms to the gender wage gap. These papers belong to the applied literature that employs the AKM model \parencite{abowd_high_1999, card_introduction_2023}, specifically focusing on firm effects on the gap. Generally, these studies use linked employer-employee data in a two-sided separable model with worker and firm identifiers as ``plugin'' estimators in a log wage linear regression. \textcite{card_workplace_2013} showed wage dispersion could be largely attributed to these firm components using West Germany data and \textcite{card_bargaining_2016} proposed a KOB decomposition using Portuguese data, finding that firms contribute around 20 percent to the gender wage gap \footnote{Other papers using the AKM model to explore firm effects KOB decomposition on the gender wage gap: \textcite{gallen_labor_2019} in Denmark, \textcite{bruns_changes_2019} in West Germany, \textcite{jewell_who_2020} in the UK, \textcite{masso_role_2022} in Estonia, and more recently \textcite{casarico_what_2024} in Italy.}.

Nevertheless, the AKM-KOB analysis assumes human capital returns for specific workers are constant across the labor market. It also requires special data manipulation to extract a dual connected set of firms through male and female workers changing jobs \parencite{abowd_computing_2002, card_bargaining_2016}. I show in my supplementary material that this data restriction may not be innocuous, since the trimming procedure may disproportionately preserve larger firms, which typically exhibit higher wage dispersion. Moreover, wage variance analysis in AKM can be biased by underestimating the role of worker-firm interactions\footnote{Some earlier examples of biased effects are \textcite{barth_assortative_2003} using Norse data and \textcite{gruetter_importance_2009} using French data, where they found negative covariance estimates in joint worker-firm effects.}. It was first assumed to be an economic phenomenon; however, \textcite{andrews_high_2008} showed that this was, in fact, an econometric issue related to small sample bias. The proposed straightforward correction to this ``limited mobility bias'' can be computationally prohibitive \parencite{gaure_correlation_2014, azkarate-askasua_correcting_2023}. \textcite{kline_leaveout_2020} introduced a Leave-One-Out methodology to fix it.

My paper goes in line with alternative approaches that moves away from AKM's additive separable assumption to avoid biased results and capture match effects. \textcite{woodcock_wage_2008} proposed a random effects approach to satisfy the rank condition. More recently, \textcite{bonhomme_distributional_2019} (BLM) proposed a novel approach that involves clustering both firms and workers into broader categories. This method offers two key advantages. First, reducing the worker firm dimensions is computationally tractable and allows for further exploration of worker-firm match effects. Second, its non-separable nature gives a unique opportunity to observe how firms valuate workers differently under similar circumstances but differing only by gender. \textcite{bonhomme_how_2023} demonstrated that random effect models such as the BLM are particularly effective in circumventing the AKM limitations, even in short time panels.

I contribute to the literature in many ways. First, to my knowledge, this study is the first to implement the methodology of \textcite{bonhomme_distributional_2019} in the context of analyzing wage disparities between two groups. Secondly, I empirically show that separable models underestimate the role of firms in generating the gap. More importantly, I find that some interactions in the labor market exhibit comparative advantage effects, generating wage levels that substantially exceed predictions from the traditional models, but with less intensity when these interactions occur with female workers. These matches contribute significantly to wage disparities but are often smoothed out under the separability assumption. Additionally, I find these interactions exist particularly in high-paying, larger firms and they are particularly strong in highly educated workers\footnote{Studies that explore high paying firms and top earners are, for example, \textcite{bertrand_dynamics_2010, bertrand_breaking_2019}, demonstrating the existence of a ``glass ceiling'' effect.}. 

My findings have meaningful policy implications. They show that closing the gender wage gap requires improving pay practices in key roles where highly skilled women are employed, particularly in leadership positions. Equally important are efforts to break down barriers in the labor market that push women into lower-paying firms. These firms not only offer smaller wage premiums but also limit the returns on women's skills and education, deepening income inequality over time.

The remainder of the paper is organized as follows: Section \ref{sec:background} provides an explanation of what is a complementarity effect and why additive separable models cannot capture it under typical settings. Section \ref{sec:data} provides an overview of the Brazilian data used in this study. Section \ref{sec:empiricalstrategy} explains the BLM clustering method in two-steps. In Section \ref{sec:mainresults}, I provide the clustering results. In Section \ref{sec:discussion}, I construct the Monte Carlo simulation counterfactuals. I conclude the paper in Section \ref{sec:conclusion}.


\section{Additive Separable Models and Complementarity} \label{sec:background}

Researchers are often interested in identifying the returns to unobserved heterogeneity of both workers and firms in labor markets, particularly when administrative data with social identifiers are available. In many empirical studies, these social identifiers are utilized as ``plug-in estimators,'' commonly modeled as fixed effects in a linear equation where the outcome is the wage in logs.

For instance, consider a labor market consisting of $N$ workers and $J$ firms, where workers and firms interact over $T$ periods. Under the assumption of additive separability, the log wage $w$ of worker \(i\) at time \(t\), net time varying effects, can be modeled as:
\begin{align}
\log w_{it} = \alpha_i + \phi_j + &\varepsilon_{it} \\
                                \text{s.t.                } j = J(i,t)&
\end{align}
where \(\alpha_i\) represents the fixed effect associated with worker \(i\) (capturing unobserved worker-specific characteristics such as skills or human capital), and \(\phi_j\) denotes the firm-specific premium associated with firm \(j\) (the wage component determined by firm characteristics, independent of worker-specific attributes). The error term \(\varepsilon_{it}\) captures idiosyncratic shocks. Firm assignment is indicated by the function \(J(i,t)\), which tracks the firm employing worker \(i\) at time \(t\).

In this framework, the additive separable model assumes constant returns for both workers and firms. That is, the firm-specific premium \(\phi_j\) is unaffected by the characteristics of the worker employed by the firm, and vice versa. This implies that reshuffling workers across firms does not alter the firm component of wages. Such an assumption is particularly strong and potentially unrealistic in labor markets where comparative advantage is thought to play a role in worker-firm interactions \parencite{shimer_assortative_2000, eeckhout_assortative_2018}.

To allow for complementarity between workers and firms, one could extend the model to include interaction effects. Specifically, the wage equation can be rewritten as:

\begin{align}
    \log w_{it} = \alpha_i + \phi_j + M_{ij} + \varepsilon_{it} \\
                                \text{s.t.                } j = J(i,t)&
\end{align}

where \(M_{ij}\) represents the interaction effect between worker \(i\) and firm \(j\). This term captures the potential complementarity effect that only arises when worker $i$ is employed at firm $j$. It reflects the idea that certain worker-firm pairings generate higher (or lower) wages than what would be predicted based solely on the worker's fixed effect \(\alpha_i\) and the firm's premium \(\phi_j\). However, estimating this interaction term in practice is infeasible due to the large dimensionality of the model. The matrix \(M_{ij}\) would have \(N \times J\) terms, which quickly becomes computationally intractable given that linked employer-employee data typically contain millions of workers and thousands of firms.

Moreover, these models are prone to bias in settings with short panel datasets, where estimating the worker and firm fixed effects becomes difficult. Building on the approach of \cite{bonhomme_distributional_2019}, I employ a methodology that reduces the dimensionality of workers and firms by clustering them into latent groups. This allows for the estimation of complementarity effects while avoiding the rank deficiency problem inherent in models with such high-dimensional interactions. Specifically, workers and firms are grouped based on their interactions in the labor market. I assume each worker-firm interaction draws wages from log-normal distributions, meaning I can relax the linear assumption and the additive separability. 

The wage generation function now can be expressed as:
\begin{equation}
\log w_{i} = f(\alpha_{L(i)} \mid \phi_{K(i,j)})
\end{equation}

where $L(i)$ denotes the assignment function of worker $i$'s type, and $K(i,j)$ denotes the assignment function of firm $j$'s class. \(f(\alpha_{L(i)} \mid \phi_{K(i,j)})\) denotes a probabilistic function that draws wages from a log-normal distribution that is specific related to the match of worker type $l$ and firm class $k$. Since wages are assumed to be derived specifically from worker-firm interactions, some matches are allowed to yield comparative (dis)advantage effects on wages, given that returns to firm and worker characteristics, under this model, are not necessarily constant across the labor market.

Therefore, recovering these latent types can be done by exploring the surface of observed wages. The idea behind the model is to recover Gaussian distributions that are ``combined'' in the full distribution, derived from all the matches occurring from the labor market. Thus the Gaussian mixture model application.

A key assumption for this model is the exogenous mobility assumption, meaning job mobility depend on the type of the worker and the classes of the firms, but not directly on earnings. I discuss the exogenous mobility assumption in my model in Section \ref{sec:exomobil}.

In the context of wage differences due to gender, my approach is to estimate the Gaussian mixture parameters in a pooled dataset, meaning the model does not observe gender at first. The reasoning is to facilitate the comparability of individuals under the same umbrella of unobserved heterogeneity. I discuss further the model in Section \ref{sec:empiricalstrategy}.

\section{Data} \label{sec:data}

In this section, I provide an overview of the Brazilian administrative data used for the study and the preparatory cleaning for the analysis, followed by a descriptive statistics of the cleaned sample.

\subsection{Data Overview and Institutional Background}
I use the \textit{Relação Anual de Informações Sociais} (RAIS), an extensive linked employer-employee dataset (LEED) from Brazil spanning from 2010 to 2017. RAIS is mandated and maintained by the Brazilian Ministry of Labor and Employment, serving as a source for the administration of tax and social programs. The dataset offers an universal representation of the formal labor market in Brazil and is characterized by its richness in variables.

A key advantage of using the \textit{RAIS} dataset for this analysis is the relative homogeneity of job-related amenities across firms due to Brazil's robust labor regulations. The Brazilian labor laws, known as the Labor Laws Consolidation (\textit{Consolidação das Leis do Trabalho} (CLT), in Portuguese), mandates a broad range of standardized benefits and protections for all formal workers, regardless of industry or firm size. This regulatory framework significantly reduces variation in non-wage compensation, allowing the analysis to focus more cleanly on wage differentials without the confounding effects of divergent job-related amenities. 

For example, Brazilian law requires all formal employees to receive the 13th salary, which is essentially a mandatory annual bonus equivalent to one month's wage, usually paid during Christmas time. Additionally, firms are obligated to provide meal vouchers or food stipends, as well as transportation subsidies for commuting. These benefits are non-negotiable and standardized across the formal labor market. Moreover, formal workers are entitled to thirty days of paid vacation, overtime pay, and severance protections via the \textit{Fundo de Garantia por Tempo de Serviço}\footnote{Roughly translated as Severance Indemnity Fund for Length of Service} (FGTS), which further ensures that variations in non-wage job characteristics can be minimized.

In Brazil, maternity leave is a legally guaranteed right under the CLT. Female employees are entitled to 120 days of paid maternity leave, funded by the Brazilian Social Security system. In some cases, companies can extend this leave to 180 days through the Empresa Cidadã program, which offers tax incentives to employers. During maternity leave, the employee's job is protected, and she is guaranteed to return to her position or a similar one without loss of salary or benefits. Additionally, Brazilian law prohibits the dismissal of pregnant workers from the moment pregnancy is confirmed until five months after childbirth, with some exceptions under fair cause.

This regulatory uniformity is particularly beneficial for my analysis, as it mitigates concerns that differences in firm payment patterns are due to job amenities that could ultimately explain wage differentials between male and female workers. In contrast, in countries where non-wage compensation varies significantly across firms or sectors, disentangling wage differences from benefit-driven compensation can complicate the analysis of wage gaps.

In my study I focus on São Paulo state, which represents the most economically dynamic region in Brazil, making sure my results are not driven by geographical heterogeneity. For example, a male worker in manufacturing and a female worker in retail, though in distinct sectors, would both receive a standardized package of legal protections and benefits coming not only from federal law, but also from local state law, ensuring that wage comparisons are not distorted by differences in state policies.

Regarding gender dynamics in São Paulo's labor market, it is important to note that, similar to other countries analyzed in the literature, approximately more than 50 percent of the Brazilian women there participate in the labor force, with 71 percent of these women employed full-time. This proportion rises to 90 percent when considering only those employed in the private sector. Furthermore, the gender wage gap in Brazil mirrors those observed in more developed economies, offering additional comparative insights. In 2016, the median earnings gap between male and female full-time workers was approximately 14.3 percent in Brazil, closely aligned with the average of 13.4 percent observed across OECD countries, and slightly better than the 18.1 percent reported for the United States \parencite{oecd_oecd_2024}.

\subsection{Data Preparation}

The RAIS database records each formal employment contract as a separate entry, meaning that for any given year, a worker with multiple contracts, whether with the same employer or different firms, will appear multiple times. To address this, and following the methodologies of \textcite{gerard_assortative_2021} and \textcite{lavetti_gender_2023}, I refine the dataset by retaining only the longest-duration and highest-paid contract for each individual per year. This adjustment shifts the data from a contract-year structure to an individual-year framework, ensuring that the analysis focuses on each worker’s primary employment.

To align with a long-run perspective, the sample is further restricted to a \textit{quasi}-full-time workers, defined as those working a minimum of 30 hours per week, and limited exclusively to the private sector. I allow this flexibility to capture a certain degree of non-monetary preference particularly found in female cohorts \parencite{goldin_grand_2014}. This exclusion criteria eliminates part-time employees, public sector workers, and the self-employed from the analysis, thereby focusing on a more homogeneous labor market.

\subsubsection{Biennial Grouping and Panel Balancing}

The organization of the data for my analysis involves grouping the dataset into jumping biennials. Specifically, the years 2010 and 2012 are paired, 2011 and 2013, and so forth. This method skips intermediate years to avoid transitional anomalies that may occur in short periods, such as firm mergers or changes in identifiers. This ``jumping'' approach closely mirrors the sample selection method employed by \textcite{bonhomme_distributional_2019}.

In my analysis, it comprises of six sets of balanced panel data spanning from 2010 and 2012 to 2015 and 2017. Each biennial set is balanced and analyzed to estimate worker and firm clusters, with final estimates related to wages presented as a weighted average of these samples. This ``rolling'' approach has been used to some extent in \textcite{card_bargaining_2016} and \textcite{lachowska_firm_2023}.

Each biennial panel is balanced, ensuring that the same set of workers and firms are observed consistently within each two-periods. In addition, firms with pronounced gender preferences are excluded from the analysis. Only firms exhibiting a gender ratio of 1 to 4 are included, which helps mitigating any potential bias that could arise from firm gender imbalance.

\subsection{Summary Statistics}

\begin{table}[htbp!] \label{tab:sumstat}
\centering
\caption{Descriptive Statistics by Gender}
\label{tab:desc_stats}
\small
\begin{tabular}{lcc}
\toprule
 & {\textbf{Female Workers}} & {\textbf{Male Workers}} \\
\textit{\textbf{Features}} & (1) & (2) \\
\midrule
\multicolumn{3}{l}{\textit{Firm Characteristics}} \\
\quad Number of Firms & {\num{346617}} & {\num{346617}} \\
\quad Firms with $\geq$ 10 Workers & {\num{204994}} & {\num{204994}} \\
\quad Firms with $\geq$ 50 Workers & {\num{58866}} & {\num{58866}} \\
\quad Mean Firm Size & 57 & 57 \\
\quad Median Firm Size & 13 & 13 \\
\addlinespace[0.5em]
\multicolumn{3}{l}{\textit{Worker Characteristics}} \\
\quad Education (\%) \\
\qquad Dropout & 22 & 28 \\
\qquad High School Graduates & 48 & 45 \\
\qquad Some College & 30 & 27 \\
\quad Age (\%) \\
\qquad $<$ 30 & 40 & 37 \\
\qquad 31--50 & 50 & 49 \\
\qquad $\geq$ 51 & 10 & 14 \\
\addlinespace[0.5em]
\multicolumn{3}{l}{\textit{Sector of Employment (\%)}} \\
\quad Primary & 2 & 2 \\
\quad Manufacturing & 19 & 26 \\
\quad Construction & 1 & 2 \\
\quad Trade & 24 & 25 \\
\quad Services & 54 & 45 \\
\addlinespace[0.5em]
\multicolumn{3}{l}{\textit{Occupation (\%)}} \\
\quad Scientific and Liberal Arts & 11 & 11 \\
\quad Technicians & 11 & 11 \\
\quad Administrative & 34 & 18 \\
\quad Managers & 5 & 7 \\
\quad Traders & 25 & 22 \\
\quad Rural & 1 & 2 \\
\quad Factory & 13 & 29 \\
\addlinespace[0.5em]
\multicolumn{3}{l}{\textit{Labor Market Outcomes}} \\
\quad Mean Tenure (years) & 4.04 & 4.63 \\
\quad Mean Log-Wage & 2.06 & 2.29 \\
\quad Variance of Log-Wage & 0.52 & 0.65 \\
\quad Worker-Year Observations & {\num{9503233}} & {\num{10283471}} \\
\quad Unique Number of Workers & {\num{3497651}} & {\num{3725990}} \\
\quad Gender Fraction (\%) & {\num{48}} & {\num{52}} \\
\bottomrule
\end{tabular}
\caption*{\small\textit{Note:} \textsuperscript{1}Descriptive statistics calculated from the first year of each biennial sample (2010-2015). \textsuperscript{2} Percentages may not sum to 100\% due to rounding. \textsuperscript{3}The number of firms is the same for both genders since every firm in the cleaned sample employs both male and female workers.}
\end{table}

Table \ref{tab:desc_stats} reports descriptive statistics by gender cohorts for the aggregated cleaned data, representing the first year of each biennial sample. Columns (1) and (2) represent the statistics for female and male workers, respectively.

The dataset encompasses a total of 346,617 unique firms. Of these, a substantial portion is relatively large; approximately 204,994 firms employ 10 or more workers, and 58,866 firms have at least 50 workers. The average firm size across the sample is 57 employees, but the median firm size is considerably smaller, at 13 employees, indicating a skewed distribution.

Gender related educational attainment confirms that women are generally more educated than their male counterparts. The data show a higher prevalence of men without high school diplomas, while women are more likely to have completed high school or pursued some college education. As stated previously, this educational dynamic is consistent with recent trends observed in both developed and developing nations, such as the United States and other OECD countries.

Approximately 40 percent of the female sample is under 30 years old, with another 50 percent aged between 31 and 50. In contrast, 37 percent of the male sample is under 30, with 49 percent in the 31 to 50 age bracket. Moreover, men are slightly more represented in the over-50 cohort, constituting 12 percent compared to 8 percent of women. Hence, the average experience in the labor market is 4.6 years for males and 4.0 years for females.

Industry distribution varies significantly between genders. Men dominate in sectors such as manufacturing, agriculture, and trade, whereas women are predominantly engaged in services, an umbrella term that includes sectors such as healthcare, education, hospitality, and financial services.

The occupational distribution also highlights a notable gender sorting: women are almost twice as likely as men to hold administrative positions, representing 34 percent of women compared to 18 almost of men. Men are more frequently employed in manual labor-intensive roles such as in agricultural settings and factories.

Despite these occupational disparities and the educational advantages observed for women, the unweighted gender wage gap remains substantial at approximately 23 log-points. This gap persists even though women are, for instance, equally likely as men to occupy scientific roles, which typically require higher educational qualifications.

\subsection{Extended Mincer Equation} \label{sec:mincer}

As a first step to analyze the gender wage gap, I provide a classical Kitagawa-Oaxaca-Blinder \parencite{kitagawa_components_1955, oaxaca_male-female_1973, blinder_wage_1973} decomposition of an extended Mincer equation and an AKM equation, assuming the gap is a mean difference of female and male wages. A ``Mincer wage function'' can be specified as:

\begin{equation}
w_{it} = \beta_0 + \beta_1 \text{Age}_{it} + \beta_2 \text{Age}_{it}^2 + \beta_3 \text{Education}_{it} + \beta_4 \text{Occupation}_{it} + \beta_5 \text{Activity}_{it} + \varepsilon_{it}
\end{equation}
where $w_{it}$ is the natural logarithm of hourly wages for individual $i$ in time period $t$, regressed on the worker's age and their squared age, their education level, the firm's industry, the worker's occupation, and a idiosyncratic error term. For the Oaxaca decomposition, I run this regression for the male and female observations separately, for each biennial sample.

Assume the matrix of explanatory observables can be expressed as $X^g$, where $g$ represents the gender sample used in the regression. Also assume $\beta$ is the vector of estimates. The KOB decomposition can be expressed as:

\begin{equation}
\bar{w}^m - \bar{w}^f = 
\underbrace{\left( \bar{X}^m - \bar{X}^f \right) \hat{\beta}^f}_{\text{Explained}} + 
\underbrace{\bar{X}^f \left( \hat{\beta}^m - \hat{\beta}^f \right)}_{\text{Unexplained}}
\end{equation}
where $\left( \bar{X}_m - \bar{X}_f \right) \hat{\beta}_f$ represents the ``explained'' component of the decomposition. In simpler terms, this term represents a counterfactual scenario where men and women possess the same returns to covariates, however, they differ in these covariates' distribution. The unexplained component, on the other hand, captures differences in the returns to these characteristics. This is expressed as $\bar{X}_f \left( \hat{\beta}_m - \hat{\beta}_f \right)$, where the difference in coefficients ($\hat{\beta}_m - \hat{\beta}_f$) measures a scenario where men and women have the same observable characteristics, however, the market values differently each gender. The unexplained portion is often interpreted as the part of the wage gap that cannot be accounted for by observable factors alone, potentially indicating discrimination or other structural labor market imbalances.

Table \ref{tab:mincerequation} presents the overall log hourly wage gap in means, the explained, and the unexplained portion of the gender wage gap across the six biennial samples, along with the number of observations for each sample. The overall wage gap remains consistent at 24 log-points for the first three samples. However, the gap slightly decreases in the subsequent samples, with the smallest gap observed in 2015-2017 at 22 log-points.

\begin{table}[htbp!]
\centering
\begin{threeparttable}
\caption{Extended Mincer Equation KOB Decomposition For Each Biennial Sample}
\label{tab:mincerequation}
\small
\begin{tabular}{lcccc}
\toprule
\textbf{Sample} & {\textbf{Overall Gap}} & {\textbf{Explained Gap}} & {\textbf{Unexplained Gap}} & {\textbf{N}} \\
\midrule
2010--2012 & -0.244 & -0.0651 & -0.179 & \num{5946240} \\
2011--2013 & -0.244 & -0.0637 & -0.180 & \num{6145676} \\
2012--2014 & -0.244 & -0.0642 & -0.180 & \num{6534444} \\
2013--2015 & -0.241 & -0.0621 & -0.179 & \num{6787446} \\
2014--2016 & -0.230 & -0.0571 & -0.173 & \num{7086062} \\
2015--2017 & -0.221 & -0.0542 & -0.167 & \num{7073540} \\
\midrule
\textbf{Weighted Avg}\tnote{a} & \bfseries -0.237 & \bfseries -0.0611 & \bfseries -0.176 & \bfseries \num{39573408}\tnote{b} \\
\bottomrule
\end{tabular}
\begin{tablenotes}[flushleft]
    \setlength{\itemindent}{-\leftmargin}
    \small
    \item \textit{Note:}  \textsuperscript{a}Weighted average calculated using sample sizes as weights and the gap as $female - male$. \textsuperscript{b}Total number of observations across all samples. \textsuperscript{1}Extended Mincer equation defined as $\log(y_i) = \beta_0 + \beta_1 \text{Age}_i + \beta_2 \text{Age}_i^2 + \beta_3 \text{Education}_i + \beta_4 \text{Occupation}_i + \beta_5 \text{Activity}_i + \varepsilon_i$. \textsuperscript{2}Explained gap represents differences in distribution of characteristics. Unexplained gap represents differences in estimated returns to characteristics.
\end{tablenotes}
\end{threeparttable}
\end{table}

In this setting, the explained portion of the Oaxaca decomposition accounts for approximately 6.11 log-points, or roughly one-quarter of the total gender wage gap. This indicates that observable factors, such as the allocation of workers across different occupations or sectors, explain about 25 percent of the wage differential in an additively separable labor market.

In Section \ref{sec:discussion}, I extend the analysis by introducing firm identifiers as fixed effects under an AKM framework following \textcite{card_bargaining_2016}. Under this specification, firm effects explain about 9 percent of the gender wage gap.

The issue with separable models is the assumption that these components should not vary depending on the association happening. Under AKM, these firm effects will occur in any worker reshuffling instance of the labor market.

In the next sections, I propose the distributional framework of \textcite{bonhomme_distributional_2019} to capture particular interactions in the labor market that does not necessarily follow an additive separable assumption.

\section{Empirical Framework: The BLM Model} \label{sec:empiricalstrategy}

Estimating the Gaussian mixture requires two main parts. Following \textcite{bonhomme_distributional_2019}, I assume cluster membership of firms is exogenous to the model, allowing their estimation by employing straightforward clustering methods from features observed from the data. Still following BLM, I choose to cluster firms based on their wage cumulative distribution function using k-means clustering \parencite{macqueen_methods_1967}. 

In the second part, I take the estimated firm clusters, called ``firm classes'', to assume that they are Gaussian mixtures of latent worker types in log wages. In the spirit of AKM settings, I leverage individuals changing jobs to identify the Gaussian parameters.

Finally, I use a \textit{maximum a posteriori} estimation to find the most likely worker type for each worker observation. After the classification, I split the data into male and female cohorts.

\subsection{Recovering Firm Classes}

The first objective is to recover firm clusters, or ``firm classes'', which are initially unobserved in the data. The approach relies on two key assumptions. First, the mapping of firms to clusters is exogenous to labor market dynamics. 

Formally, let $k(j)$ denote cluster assignment of firm $j$. The exogeneity assumption can be expressed as:

\begin{equation}
    P(k(j) | X) = P(k(j))
\end{equation}
where $X$ represents labor market conditions and worker characteristics. In plain language, this condition ensures that the probability of a firm belonging to a firm class is unconditional to these labor market features, which allows a direct estimation of firm classes using the clustering method.

Secondly, the wage distribution in the data follows a log-normal shape for workers, conditional on these firm clusters. Consequently, each firm class represents a Gaussian mixture of log-wages. Within these mixtures, each component corresponds to a log-normal distribution arising from the unobserved heterogeneity of worker groups, which is termed ``worker types'' in this study, following BLM's terminology.

Formally, the assumption states that for a firm $j$ in class $k$, the log-wage distribution, for a given time period, can be expressed as:

\begin{equation}
f_k(w_i) = \mathds{1} \{\hat{k}(i) = k \} \sum_{l=1}^L q_{k}(L(i)) \mathcal{N}(\theta_{kl})
\end{equation}
where, $f_k(\log(w_i))$ is the log hourly wage mixture of firm class $k$, when observing worker $i$. With some abuse of notation, $L$ denotes the number of worker types, $q_k(L(i))$ represents the proportion of workers of type $L(i)$ in class $k$, and $\mathcal{N}(\theta_{kl})$ is the Gaussian probability density function for type $l$ workers in class $k$, with $\theta_{kl}$ representing the parameters of this distribution. The indicator function $\{\hat{k}(i) = k \}$ ensures that we consider only the wage distributions of workers assigned to the specific firm class $k$.

My approach leverages firm clustering to address the dimensionality challenge inherent in firm heterogeneity analyses. By aggregating individual firms into a more manageable set of ``firm classes'', I circumvent the need to restrict the dataset to a set of connected firms through workers. However, the identification strategy of this methodology still relies on job movements. It shifts, however, the focus from tracking movements between individual firms to observing transitions across firm classes. Therefore, while this mixture model still fundamentally relies on job mobility, it does so at a more aggregated level. In the supplementary material, I perform a clustered AKM regression to show that on average, the residual change in wages for these movers is close to zero, suggesting the movement pattern is not related to the labor market structure itself.

A crucial assumption of this approach is that each worker type, to be estimated in the second step, exhibits a unique pattern in their ``cycling'' through firm classes as they navigate job changes. These transitional pathways must be sufficiently distinct to allow for clear identification of worker type parameters \parencite{bonhomme_distributional_2019}. The robustness of this assumption in my context is based on the substantial number of observations in the dataset, which provides the statistical power necessary to discern these distinct mobility patterns between worker types and firm classes.

The k-means algorithm aims to group firms with similar payment schedules. Formally:

\begin{equation} \label{eq:firmkmeans}
    \argmin_{k(1),\ldots k(J), H_1,\ldots , H_K} \sum^J_{j = 1} n_j \int (\hat{F}_j(w) - H_{k(j)}(w))^2 d\mu(w)
\end{equation}
where $\hat{F}_j$ represents the empirical CDF of the log-weekly wages $w$ of firm $j$, $\mu$ is a discrete measurement, supported by a finite grid of ventiles from the population. $K$, the number of firm classes, is known, while the array $k(1),\ldots, k(J)$ represents the partitioning for each firm. $H_{k}$ represents cluster $k$'s CDF. Finally, $n_j$ is the firm's corresponding workforce size. I perform 1000 repetitions to ensure a global minimum distance estimation. 

In simple terms, this procedure minimizes the distance between firms and unobserved classes using as measurements each firm's empirical CDF generated from the ventiles of the observed population log hourly wage distribution. It imposes a weighting parameter to ensure different minimization process for larger firms. For each biennial sample, I assume that the firm class classification is time-invariant.

I choose $K = 10$ as the baseline number of groups since it minimizes the wage variance within each group. I follow \textcite{bonhomme_grouped_2015, bonhomme_distributional_2019}, where the estimation of firm classes does not affect parameter estimation in the Gaussian mixture step. Nevertheless, in the Appendix \label{sec:clusterchoiceanalysis}, I provide a comprehensive cluster choice analysis using gap statistics to find optimal K-Means clustering estimation. I also provide alternative cluster settings as robustness checks in the discussion section.

\subsection{Gaussian Mixture Estimation}

I assume that observed wages follow a mixture of log-normal distributions, where every ``latent'' probability distribution represents an interaction of a latent worker ``type'' with the respective firm class. This approach enables me to reduce the high-dimensional unobserved heterogeneity among individual workers into a manageable set of Gaussian distributions.

I estimate the parameters with the pooled dataset, not observing gender at first. By not accounting for gender at the outset, I ensure that male and female workers assigned to the same distribution are as similar as possible in terms of unobserved characteristics. The idea is that the algorithm will approximate individuals with sufficiently similar unobserved characteristics that spawn the same distribution of wages, regardless of gender. It allows for a more precise comparison of how these latent worker types interact with firm classes without biasing the results by preemptively imposing gender differences.

This approach allows for a more flexible examination of the wage structure assumption in the labor market. By constructing and comparing expected payment levels for each worker type and firm class interaction, I can empirically assess at which extent the additive separability assumption hold, and capture interactions in the market that deviates from this condition. Finally, I can disaggregate these payment levels by gender to measure the differential complementarity effects on wages, providing insights into how worker-firm interactions contribute to gender wage disparities, especially at matches where the separable form is not observed.

\subsubsection{Recovering Worker Types}
To identify latent worker types, I posit that the wage distribution for each type depends on their associated firm class and follows a log-normal distribution. This approach incorporates potential complementarities characteristic of specific worker-firm matches. I first, estimate the densities for job movers, and subsequently, I estimate the proportions of stayers using the job mover distributions from the initial period.

I formulate this as a maximum likelihood problem, closely following \textcite{bonhomme_distributional_2019}:
\begin{equation} \label{eq:gaussian_jobmovers}
\argmax_{\theta_p, \theta_{1}, \theta_{2}} \sum_{i=1}^{N_m} \sum_{k=1}^K \sum_{k'=1}^K \mathds{1} \{ \hat{k}_{i1} = k \} \mathds{1} \{\hat{k}_{i2} = k' \} \log \left( \sum_{l=1}^L p_{kk'}(l;\theta_p) f_{k l}^1 (w_{i1}; \theta_{1}) f_{k' l}^2 (w_{i2}; \theta_{2}) \right)
\end{equation}
where $N_m$ denotes the number of job movers, $K$ the number of firm classes, and $L$ the number of worker types (set to 10 for interpretability). The indicator functions $\mathds{1} \{ \hat{k}_{i1} = k\}$ and $\mathds{1} \{\hat{k}_{i2} = k'\}$ capture the transition of worker $i$ from firm class $k$ to $k'$ between periods 1 and 2. $p_{kk'}(l;\theta_p)$ represents the proportion of type $l$ workers moving from class $k$ to class $k'$, while $f_{kl}^1$ and $f_{k'l}^2$ are log-normal wage distributions for type $l$ workers in classes $k$ and $k'$ in periods 1 and 2, respectively.

Therefore, Equation \ref{eq:gaussian_jobmovers} captures the parameters of the conditional distributions of the worker types leveraging the job movers. 

For job stayers, I estimate:
\begin{equation} \label{eq:gaussian_jobstayers}
\argmax_{\theta_q} \sum_{i=1}^{N_s} \sum_{k=1}^K \mathds{1} \{ \hat{k}{i1} = k \} \log \left( \sum_{l=1}^L q_{k}(l;\theta_q) f_{k l}^1 (w_{i1}; \hat{\theta}_{1}) \right)
\end{equation}
where $N_s$ is the number of stayers, and $q_{k}(l;\theta_q)$ is the proportion of type $l$ stayers in class $k$. I leverage the first year parameters for job movers. I employ the Expectation-Maximization (EM) algorithm with 50 repetitions to estimate these parameters.

To recover the most likely worker type for each observation, I utilize the Maximum A Posteriori (MAP) estimation. Formally, for a worker $i$ in firm class $k$ with wage $w_i$, the probability of belonging to type $l$ is given by:
\begin{equation} \label{eq:worker_type_probability}
P(l|w_i, k) = \frac{q_k(l;\hat{\theta}q) f_{kl}(w_i; \hat{\theta})}{\sum_{l'=1}^L q_k(l';\hat{\theta}_q) f_{kl'}(w_i; \hat{\theta})}
\end{equation}

The worker type is then assigned as:
\begin{equation} \label{eq:worker_type_assignment}
\hat{l}_i = \argmax_l P(l|w_i, k)
\end{equation}

\section{Estimated Parameters} \label{sec:mainresults}

In this section, I present the estimated parameters for the mixtures, beginning with firm class estimates, followed by the mixture proportions, and concluding with a detailed analysis of the estimated moments disaggregated by gender.

\subsection{Cluster eCDFs}

The effectiveness of the algorithm in segregating firms into distinct clusters is evaluated by visualizing the empirical cumulative distribution function of the generated clusters. They are illustrated in Figure \ref{fig:ecdf_class}.

As depicted, the algorithm managed to delineate mostly clear firm classes, grouping firms with similar pay policy, evidenced by the ``clear cuts'' of each cluster's CDFs, with the exception being firm class 4.

\begin{figure}[htb!]
    \centering
    \begin{subfigure}{\textwidth}
        \centering
        \includegraphics[width=\textwidth]{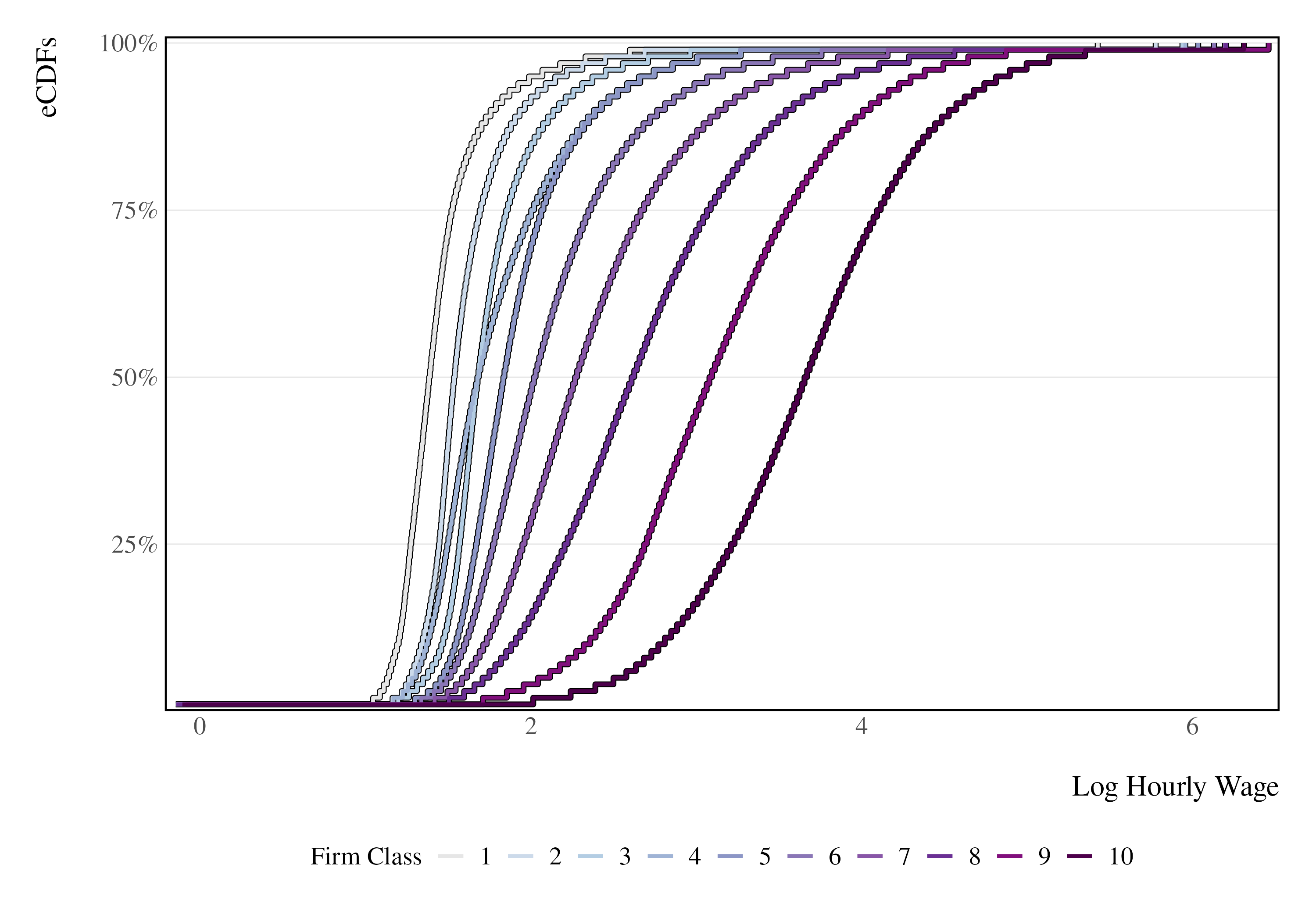}
        \caption{Empirical CDF of Firm Classes}
        \label{fig:ecdf_class}
    \end{subfigure}
    
    \vspace{1em}
    
    \begin{subfigure}{0.48\textwidth}
        \centering
        \includegraphics[width=\textwidth]{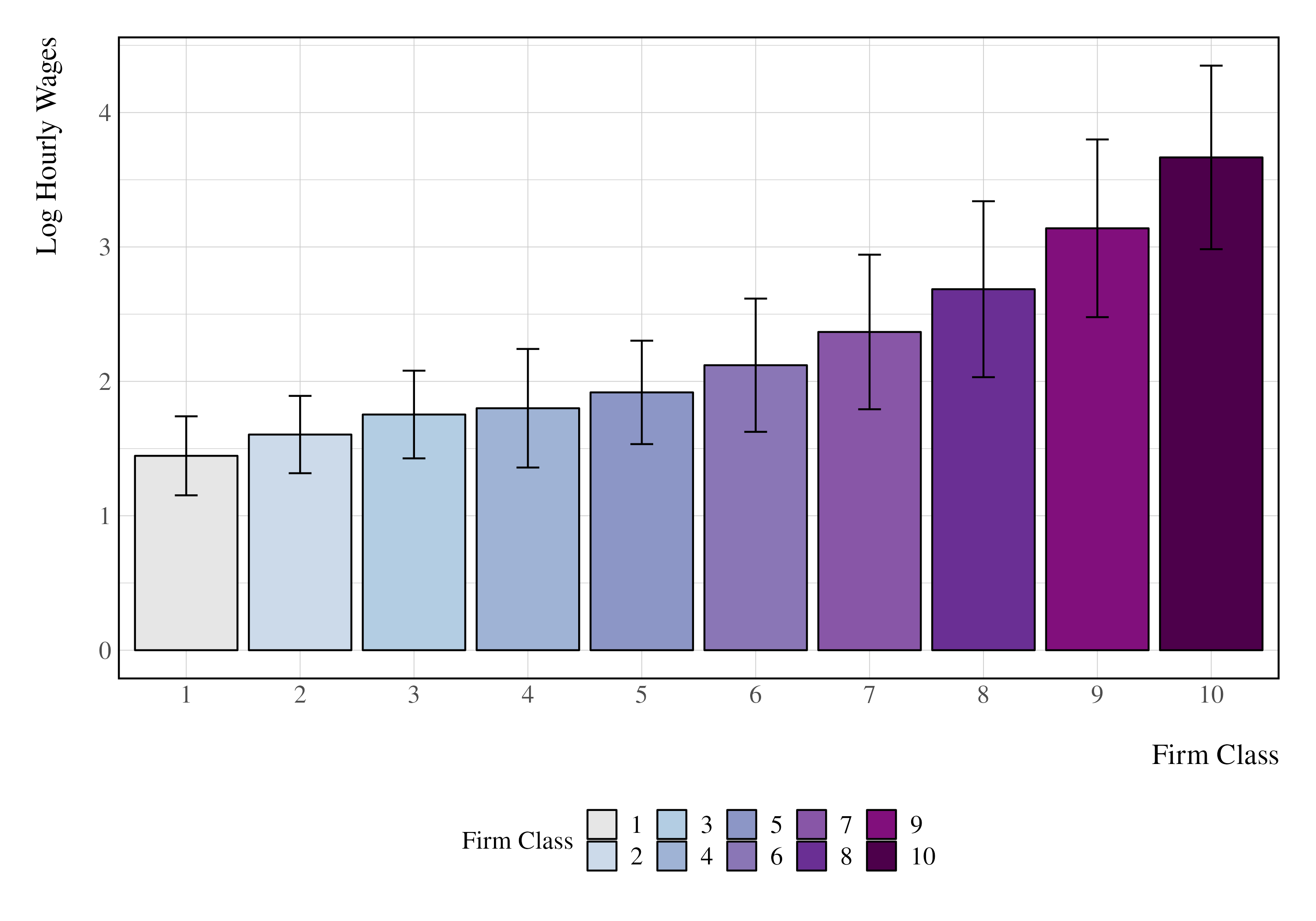}
        \caption{Wage (Mean and Variance) Statistics}
        \label{fig:wagestat}
    \end{subfigure}
    \hfill
    \begin{subfigure}{0.48\textwidth}
        \centering
        \includegraphics[width=\textwidth]{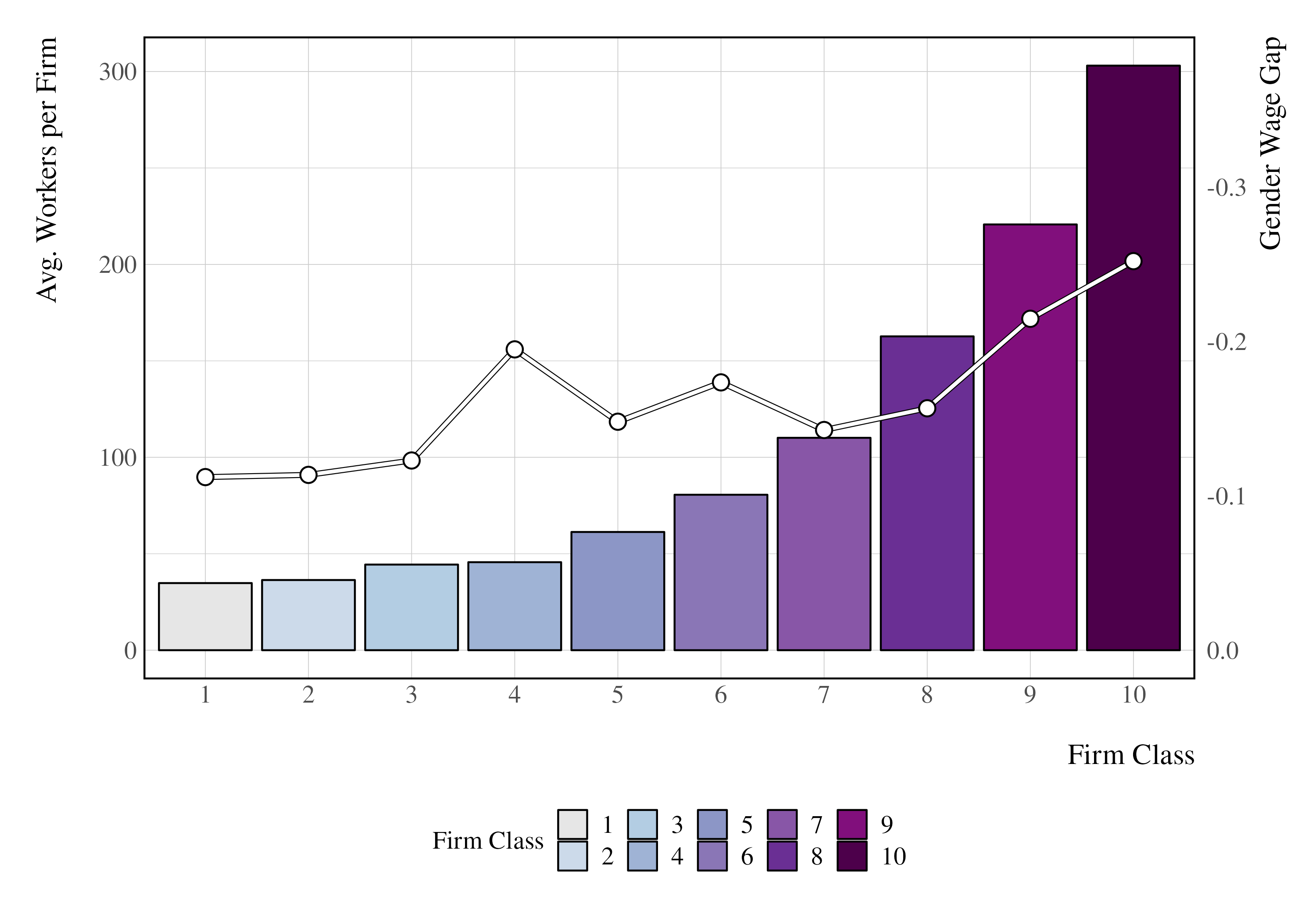}
        \caption{Gap and Size Statistics}
        \label{fig:sizestat}
    \end{subfigure}
    
    \caption{(a) Firm Class ECDFs, (b) Firm Class Mean and Variance, and (c) Firm Class Size and Gender Wage Gap Statistics}
    \caption*{\footnotesize \textit{Note:} \textsuperscript{1}Firm classes estimated by a k-means clustering algorithm using as measurement their empirical cumulative distribution function supported by the ventiles of the population. \textsuperscript{2}The Gender wage gap in means (line in Panel C) is calculated as the female minus male: $\mathbb{E}[w_{it}^f | k] - \mathbb{E}[w_{it}^m | k]$.}
    \label{fig:combined}
\end{figure}

Figure \ref{fig:wagestat} provides the moments of their log hourly wage distribution, with the means as the bars and the first standard deviation as the error-bars. For each estimated cluster, not only expected payment increase but also their dispersion when going upward in the firm class ranking. For example, the lowest firm class pays, on average, 1.45 in log hourly wages, with a variance of 0.09, while the highest pays 3.67 with 0.47 in variance.

Figure \ref{fig:sizestat} reveals the gender wage gap in means (expressed as $\mathbb{E}[w_{it}^f | k] - \mathbb{E}[w_{it}^m | k]$) as the line plot (the right $y$ axis), and the average size per firm as the bar plot (the left $y$ axis). The expected gender wage gap in means has a tendency to increase when going up in firm class ranking. The lowest paying firms are the most equitable firms in the labor market, with the lowest difference between genders at 11 log-points. The plot also reveals highest-paying firms, which tend to be larger firms\footnote{For a full descriptive statistics of firm classes, see Table \ref{tab:lowerclasses} and \ref{tab:upperclasses}}, exhibiting the largest gender wage disparities, reaching 25 log points. This finding is not entirely unexpected given the substantial variance in wages within firm class 9 or 10. This pattern suggests potential overestimation of the magnitude of firm effects contribution to the gender gap under additive separable models. This overestimation likely stems from the necessary practice of focusing on large firms to ensure sufficient worker mobility within a connected set, while addressing the ``double-coincidence'' problem of observing both male and female job transitions. However, this approach inadvertently oversamples precisely those firms where gender wage disparities are most pronounced, potentially skewing overall estimates of firm effects on the wage gap.

\subsection{Assortative Matching of Estimated Parameters}

Figure \ref{fig:priors} displays the unconditional distribution of workers across firm classes (top row) and worker types (bottom row) for each gender. Both male and female workers exhibit a concentration of employment in firm class 6, but the proportion is slightly higher for men, with 15 percent of the male workforce in this class compared to 13 percent for women. Additionally, the distribution for men shows a more noticeable skew towards higher-productivity firms. Specifically, 17 percent of men are employed in the top two firm classes (9 and 10), whereas about 14 percent of women are employed in these high-productivity firms. This suggests that men are more likely to be employed in firms that offer higher wage premiums, which may contribute to the observed gender wage gap through the sorting channel.

The differences in distribution become more pronounced when examining worker types. The female distribution is heavily skewed to the left, with nearly 24 percent of women concentrated in worker type 3 versus 17 percent among male workers. In contrast, the male distribution is more evenly spread across worker types, exhibiting a more balanced, albeit still slightly left-skewed, pattern.

In this paper, worker types represent comparable unobserved heterogeneity. Meaning female and male type 3 are individuals where their wages are likely drawn from the same set of Gaussian distributions. The firm class distribution has a more straightforward interpretation, as the proportion of firms with similar payment policies, mirroring patterns of productivity and industry.

When I discuss the gender wage gap decomposition, I hold the distribution of worker types constant since channels of worker type heterogeneity may arise from a multitude of mechanisms in the labor market, such as non-monetary preferences or human capital levels.

Figure \ref{fig:typeclasssort} displays firm classes along the horizontal axis against the stacked conditional proportions of corresponding worker types, separately for female and male workers. These proportions are recovered by grouping types for each male and female sample conditional on each firm class after the \textit{maximum a priori} classification.

Worker types and firm classes are numbered according to expected payment. Therefore, type 10 represents on average the highest paid worker in the data, a proxy for individuals that overall possess high human capital value. The visual representation clearly illustrates an assortative matching pattern, revealing that higher-paying firms predominantly employ higher types of workers for both genders. However, there are notable differences between male and female sorting patterns.

For female workers, there is a strong concentration of lower-type workers in lower-class firms. For instance, in firm class 1, 29 and 23 percent of the workforce comprise of type 1 and type 2 workers, with another 36 percent being type 3 and 4 together. Moving to higher firm classes, this composition shifts dramatically: in class 10, less than 5 percent are type 1 workers, while 15 and 26 percent belongs to type 10 and 9 workers.

On the other hand, male workers shows a slightly different trend. Type 1 and 2 workers comprise together 44 percent of firm class 1 workforce, slightly less concentrated than for females. In the highest firm class, while also presenting negligible proportions of the lowest type, 51 percent of the workforce is comprised of type 10 and 9 workers.

Therefore, while assortative matching is evident for both genders, the patterns reveal some disparity in how men and women with sufficiently similar unobserved heterogeneity are sorted across firm classes. Women appear to face some friction in ascending the firm classification hierarchy, resulting in a more pronounced concentration in lower-tier firms even when their latent productivity (as captured by worker types) is comparable to that of their male counterparts. 

\begin{figure}[htbp!]
    \centering
    \includegraphics[width=\textwidth]{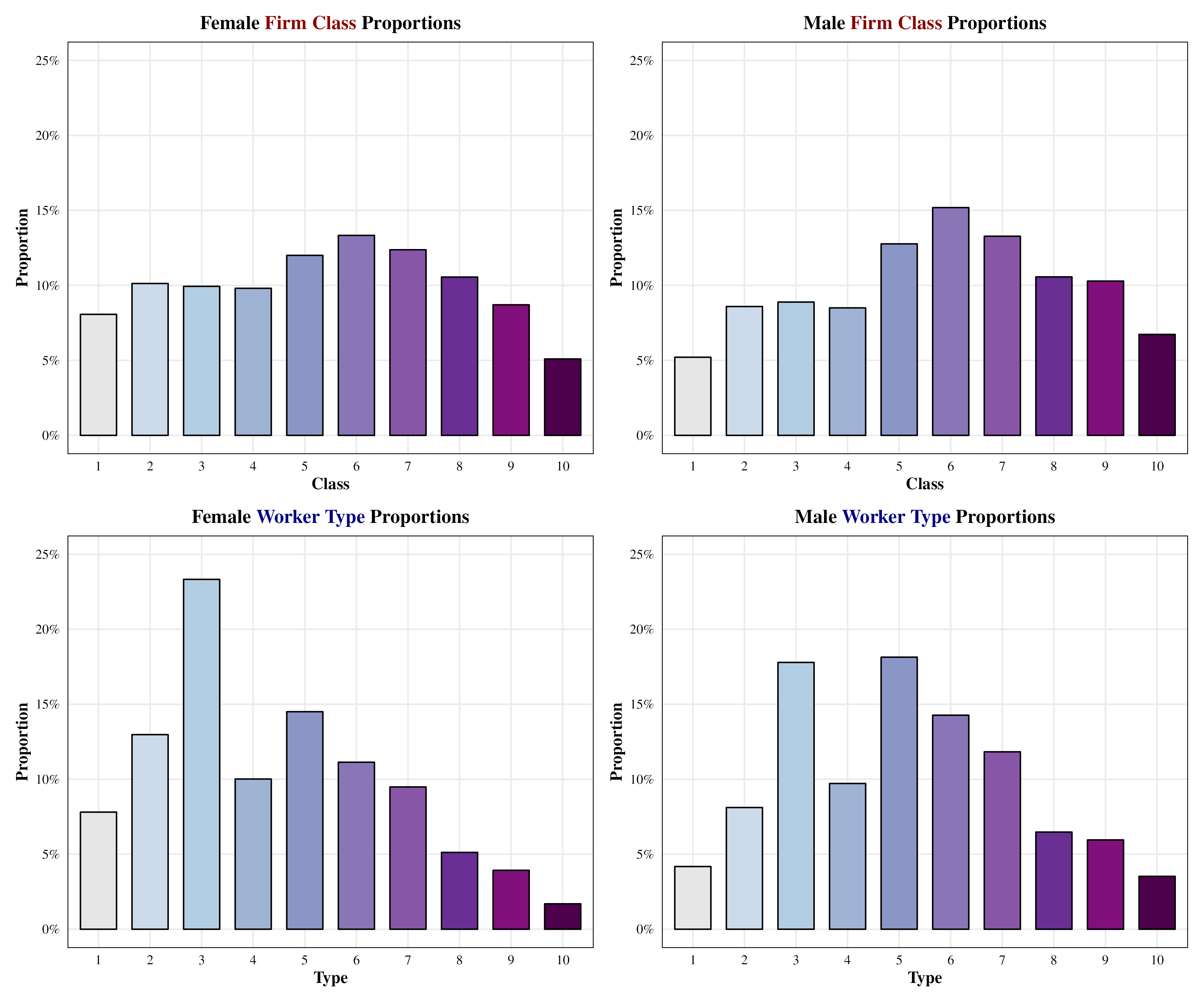}
    \caption{Worker Type and Firm Class Unconditional Probabilities per Gender}
    \caption*{\small \textit{Note:} \textsuperscript{1} Firm class estimated using k-means clustering of the cumulative distribution function of payments. Worker types estimated using a Gaussian mixture model where I assume each latent worker type interact with firm classes by drawing wages from a log-normal distribution.}
    \label{fig:priors}
\end{figure}

 \begin{figure}[htbp!]
    \centering
    \includegraphics[width=1\textwidth]{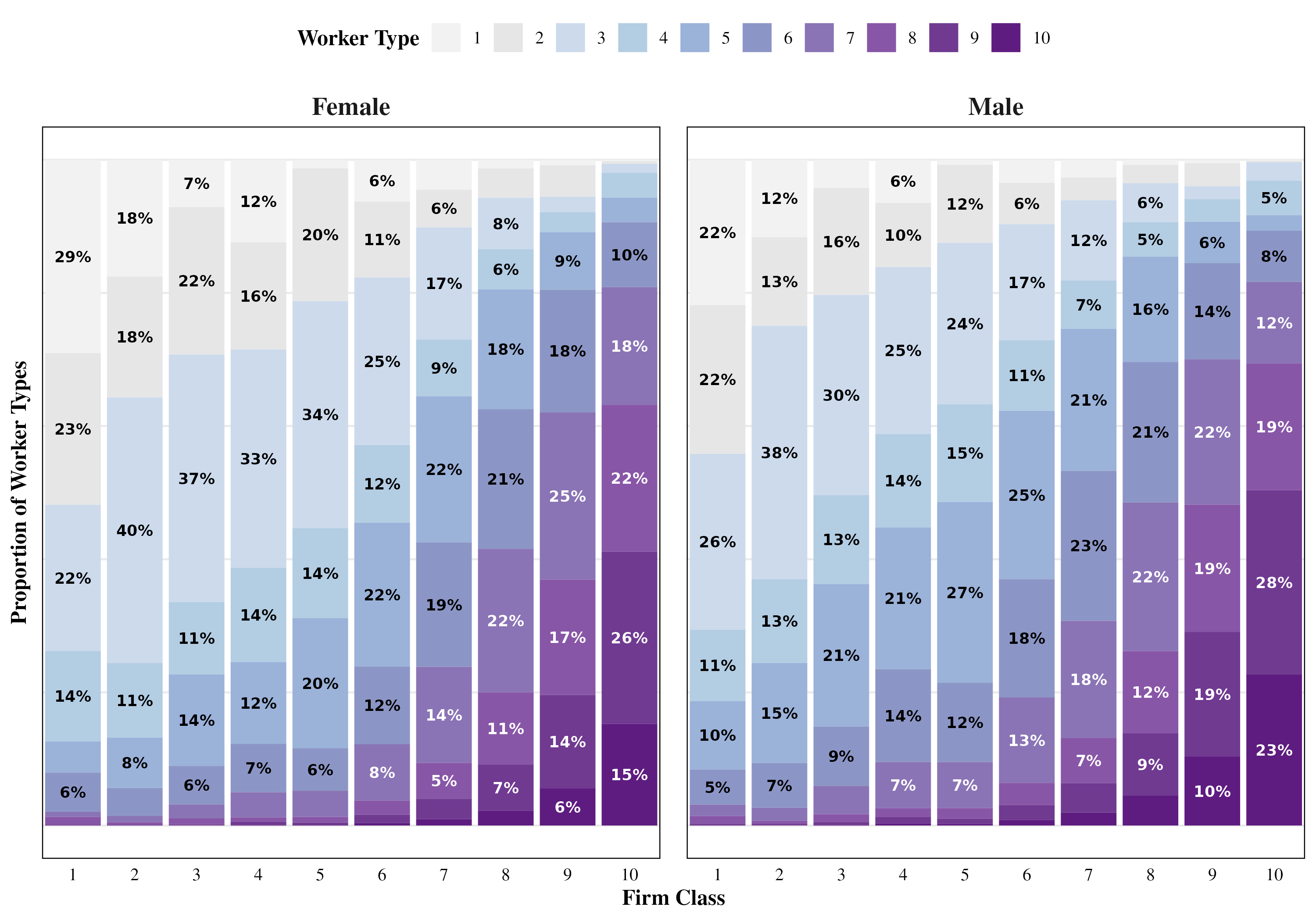}
    \caption{Proportion of Estimated Worker Types and Firm Classes}
    \caption*{\small \textit{Note:} \textsuperscript{1}Proportions of worker types recovered using a finite Gaussian mixture of log hourly wages conditional on observed firm classes. Firm classes recovered using k-means clustering algorithm on firm's log hourly wage's CDFs. \textsuperscript{2}Worker type membership assigned using a \textit{maximum a priori} estimation. \textsuperscript{3}Types and classes ordering based on expected log hourly hours.}
    \label{fig:typeclasssort}
\end{figure}

\subsubsection{Theil Index}

To quantitatively assess whether the male distribution in the labor market is slightly more symmetrical compared to the female, I perform a Theil Index calculation to measure the inequality, where the metric is the number of workers per match. I separate the firm classes into low and high classes, where low comprises of firm class 1 to firm class 5, while high comprises of firm class 6 to firm class 10.

The simple Theil Index formula is:
\begin{equation}
    T = \frac{1}{M}\sum_{m=1}^{M} \frac{N_m}{\bar{N_m}} \log \left ( \frac{N_m}{\bar{N_m}} \right )
\end{equation}
where $m$ is each match in the labor market, $M$ the total number of matches, in this case, 100. $N_m$ is the total number of workers for each match, while $\bar{N_m}$ is the average number of workers per match in the labor market.

The Theil index for the male distribution is 0.43. For the female distribution is 0.48, slightly larger, suggesting that the female distribution of workers in the labor market is more sorted towards the left, concentrated in overall less paying matches.

\subsection{Payment Schedules}

The BLM method not only demonstrates the flexibility to capture assortative matching but also enables researchers to discern the underlying wage structure arising from firm-worker interactions by assessing whether certain compensation patterns result in wage levels that surpass predictions from additive separable models.

Figure \ref{fig:payschedules} presents the payment schedules by firm class and worker type. Panel (a) represents estimated average payments directly under the Gaussian mixture model. Panel (b), on the other hand, is a counterfactual scenario where worker and firms do not yield complementary wage effects in their interactions. I construct this counterfactual by performing an ``AKM'' two-way fixed effect model such as:

\begin{equation}
    w_{it} = \alpha^g_{L(i)} + \psi^g_{K(i,t)} + \varepsilon_{it}
\end{equation}
where $w_{it}$ is the log hourly wage of worker $i$ in time period $t$, $\alpha_{L(i)}$ is the fixed effect of worker $i$'s type $l$, represented under the assignment function $L(i) = l$, $\psi^g_{K(i,t)}$ is the fixed effect of firm class $k$, also represented under an assignment function $K(i,t)$. $\varepsilon_{it}$ is the idiosyncratic error term. To make sure I preserve gender disparities, I regress twice for each gender sample.

I introduce a weighting parameter to mitigate the influence of extreme values on my estimation. It leverages the fact that common interactions in the labor market tend to possess small complementarity effects\footnote{Figure \ref{fig:typeclasssort} reveals that ``extreme complementarity'' matches are approximately 5 percent of the total.}. Consequently, the objective function for this minimization problem can be expressed as:

\begin{equation} \label{eq:weightedols}
    \min \sum_{i,t} n_{M(k, l)} (w_{it} -  \alpha^g_{L(i)} + \psi^g_{K(i,t)})^2
\end{equation}
where $n_{M(k, l)}$ represents firm class $k$ and worker type $l$ match's proportion of the number of workers.

Each panel in Figure \ref{fig:payschedules} shows each line representing an expected payment ``path'' of each worker type when hired by a particular firm class.

The Gaussian mixture model is able to capture different wage levels that do not necessarily follow a linear trend, as shown in Figure \ref{fig:gaussianpred}.Top firm classes tend to offer substantially higher wage levels to individuals, with particularly pronounced effects for workers in the lower to middle range of the skill distribution. High-earning individuals exhibit remarkable wage stability across firm classes, maintaining their elevated earnings even when matched with low firm classes, with a slight decrease. There is also severe wage compression at the left tail of the distribution. In particular, ``worker type 8'' experiences severe wage compression if matched with extreme low firm classes such as 1 or 2. 

When worker-firm interactions are assumed to be ``additive separable'', particular interactions are smoothed out as shown by Figure \ref{fig:linearpred}, the lines become parallel, which is the quintessential feature of the additive separability assumption: workers and firms contribute to the wage generation function by adding their respective ``values''. That means worker type 10, for example, if transferred from firm class 10 to 1, should not lose the part of their wage that belongs purely to their components. 

\begin{figure}[H]
    \centering
    \begin{subfigure}{0.7\textwidth}  
        \centering
        \includegraphics[width=\linewidth]{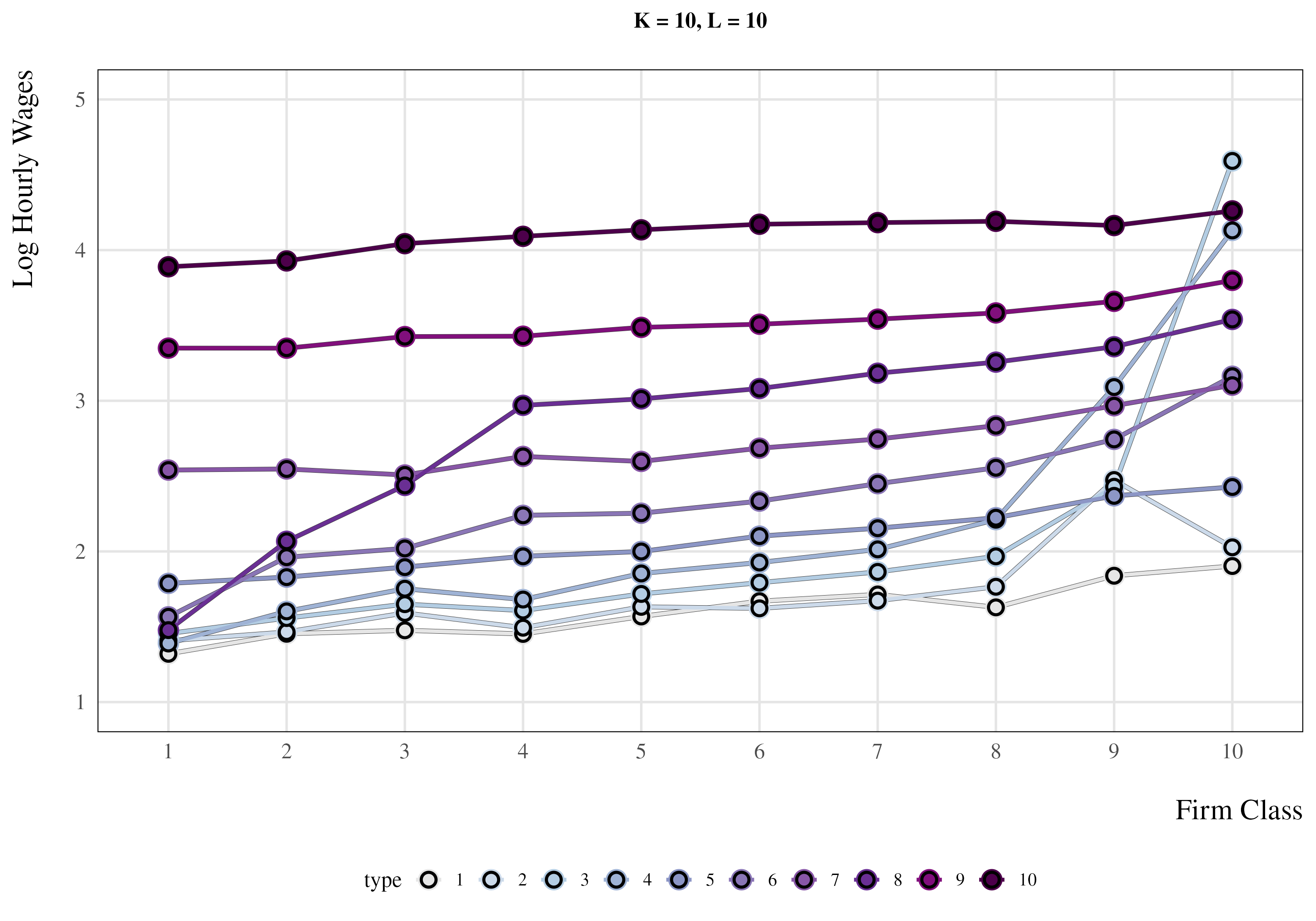}
        \caption{Gaussian Mixture Estimated Wages per Match}
        \label{fig:gaussianpred}
    \end{subfigure}

    \vspace{1em} 

    \begin{subfigure}{0.7\textwidth}  
        \centering
        \includegraphics[width=\linewidth]{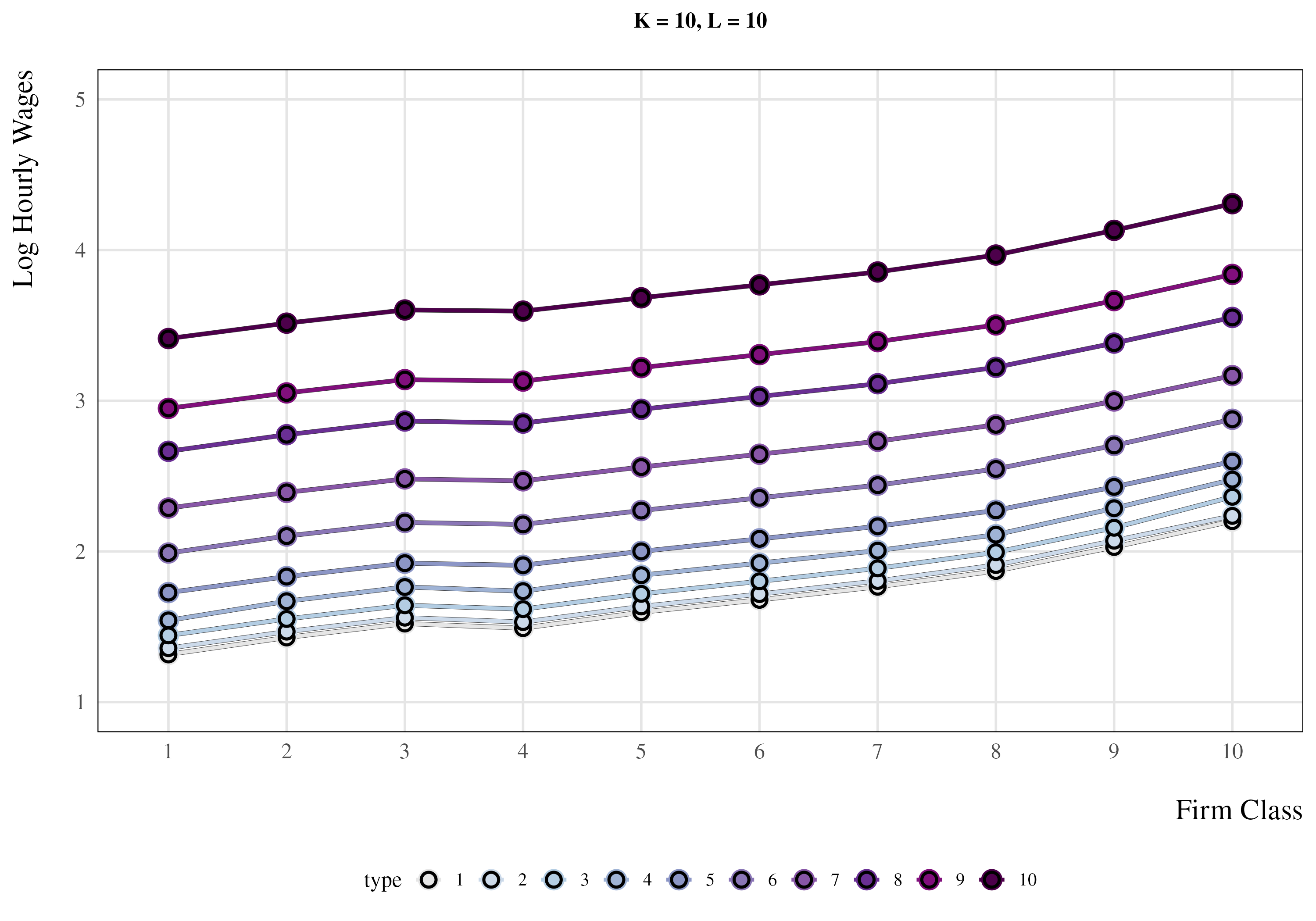}
        \caption{Additive Separable Prediction of Wages per Match}
        \label{fig:linearpred}
    \end{subfigure}
    
    \caption{Payment Schedules of worker-firm interactions under Gaussian mixture estimates and predicted linear model.}
    \caption*{\small \textit{Note:} Panel (a) generated by using estimated means and variances of each Gaussian component of the mixture. Panel (b) generated by running a two-way fixed effect estimation with firm classes and worker types as fixed effects, weighted by the number of workers per each match.}
    \label{fig:payschedules}
\end{figure}

\begin{figure}[H]
    \centering
    \includegraphics[width=\textwidth]{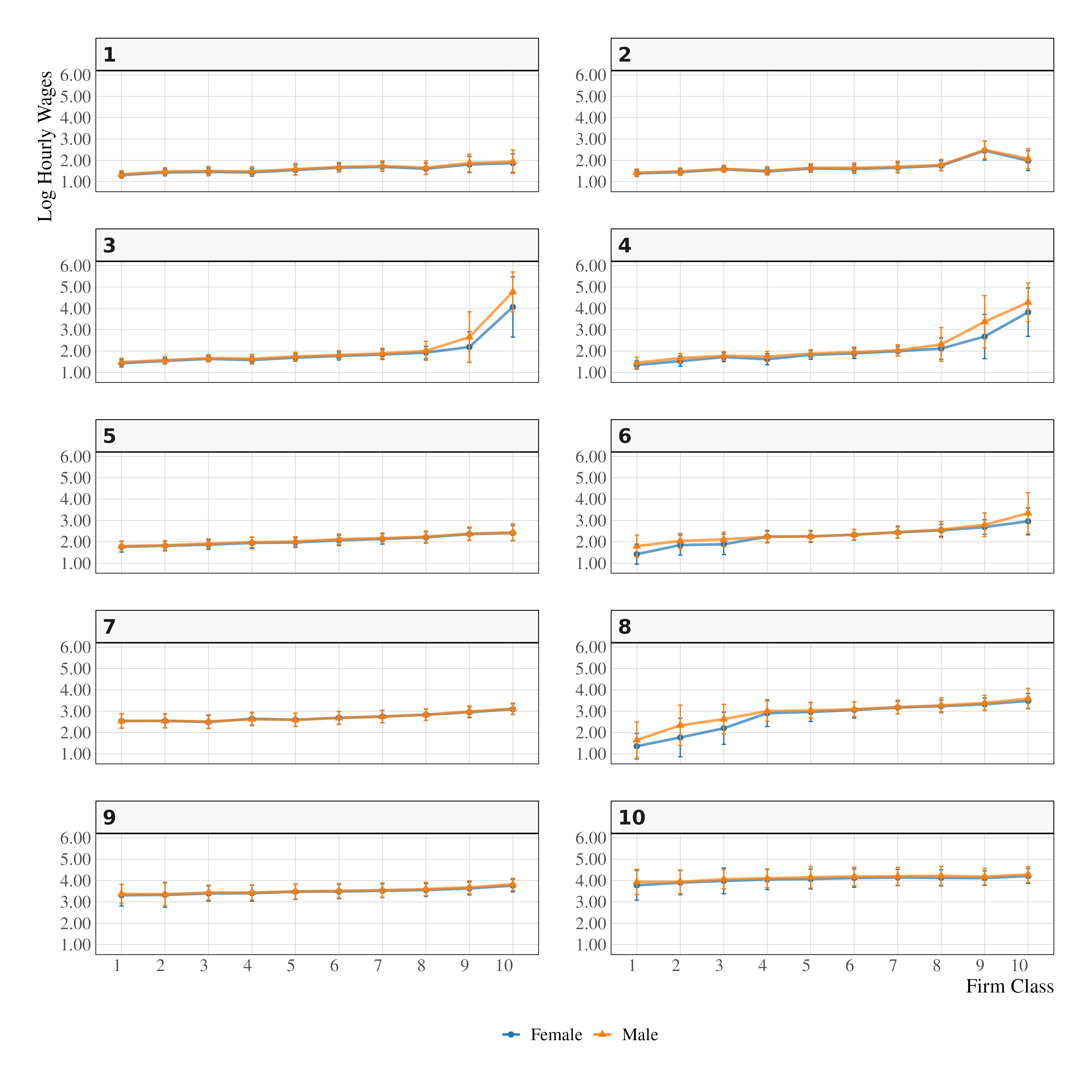}
    \caption{Pay Schedules by Gender, Firm Class, and Worker Type}
    \caption*{\small \textit{Note:} \textsuperscript{1}Each panel represents a worker type payment schedule for each firm class, in log hourly wages, grouped by gender. \textsuperscript{2}Points indicate the mean log hourly wage, and error bars represent one standard deviation of the estimated wage distribution.}
    \label{fig:genderpayschedules}
\end{figure}

\subsubsection{Gender-wise Payment Schedules}

A natural question that arises is to what extent these particularities matter for the creation of pay differentials between male and female workers. I construct Figure \ref{fig:genderpayschedules} by expanding Figure \ref{fig:gaussianpred}, separating each worker type into a male and a female components.

Complementarity effect arising from those matches implies significantly higher wages for men as compared to women, in particular for the worker types 3, 4, and 8, for which the deviation from a ``separable setting'' is most extreme. For example, type 3 matched with firm class 10 yields almost 1 log hourly wage gap. Type 4 under the same match yields about 50 log-points in gap.

Type 8, and to a lesser extent type 6, exhibit negative deviations as they approach the lower extreme of firm productivity. When these worker types, characterized by moderate-to-high human capital accumulation, are found in low-productivity firms, a compression effect on wages emerges. In these cases, the expected wage falls below the sum of the expected firm and worker effects. Female workers are more susceptible to these unfavorable matches compared to their male counterparts.

\subsubsection{Workers Under Comparative Advantage}

To understand better the differences being under a comparative advantage match and otherwise in the labor market, I compare individuals with sufficiently close payments but differing in matches. Specifically, I compare the type 3 and type 10 male and female workers when hired by the class of firm 10. Table \ref{tab:complementarity_workers} provide a descriptive statistics of these workers.

\begin{table}[htp!]
\centering
\begin{threeparttable}
\caption{Workers Under Complementarity and Non-complementarity Matches in Firm Class 10}
\label{tab:complementarity_workers}
\begin{tabular}{lcccc}
\toprule
    &\multicolumn{2}{c}{Type 3 Workers} &\multicolumn{2}{c}{Type 10 Workers} \\
    \cmidrule(lr){2-3} \cmidrule(lr){4-5}
  & Female & Male & Female & Male\\
  & (1) & (2) & (3) & (4) \\
\midrule
\multicolumn{5}{l}{\textit{Education And Age}} \\
Dropout & \num{0.04} & \num{0.01} & \num{0.00} & \num{0.01}\\
High School Graduates & \num{0.12} & \num{0.06} & \num{0.03} & \num{0.04}\\
College & \num{0.84} & \num{0.92} & \num{0.97} & \num{0.95}\\
\addlinespace
Age ($<$30) & \num{0.15} & \num{0.06} & \num{0.10} & \num{0.09}\\
Age 31-50 & \num{0.69} & \num{0.66} & \num{0.74} & \num{0.70}\\
Age ($\geq$51) & \num{0.13} & \num{0.24} & \num{0.14} & \num{0.18}\\
\addlinespace
\midrule
\multicolumn{5}{l}{\textit{Occupation Statistics}} \\
Scientific and Liberal Arts & \num{0.20} & \num{0.25} & \num{0.36} & \num{0.37}\\
Technicians & \num{0.04} & \num{0.03} & \num{0.08} & \num{0.11}\\
Administrative & \num{0.18} & \num{0.07} & \num{0.19} & \num{0.13}\\
Managers & \num{0.52} & \num{0.62} & \num{0.35} & \num{0.36}\\
Traders & \num{0.05} & \num{0.01} & \num{0.02} & \num{0.02}\\
Rural & \num{0.00} & \num{0.00} & \num{0.00} & \num{0.00}\\
Factory & \num{0.01} & \num{0.01} & \num{0.01} & \num{0.01}\\
\midrule
Mean experience (years) & \num{6.375} & \num{7.302} & \num{8.417} & \num{8.288}\\
\addlinespace
Mean Log-Wage & \num{4.308} & \num{4.866} & \num{4.204} & \num{4.270}\\
Variance of Log-Wage & \num{1.836} & \num{0.715} & \num{0.094} & \num{0.116}\\
\midrule
Worker-years observations & \num{5718} & \num{18401} & \num{73869} & \num{157282}\\
Number of Workers & \num{4895} & \num{15542} & \num{44309} & \num{93073}\\
Fraction of Women & \num{0.24} & \num{0.76} & \num{0.32} & \num{0.68}\\
\bottomrule
\end{tabular}
\begin{tablenotes}
\small
\item \textit{Notes:} \textsuperscript{1}Under complementarity matches are type 3 and type 4. Without complementarity is type 10 match, which in wage levels is comparable to matches under complementarity. \textsuperscript{2} Education, age, and occupation statistics are fractions that may not necessarily add to one due to rounding.
\end{tablenotes}
\end{threeparttable}
\end{table}

The table highlights distinct differences in education, age distribution, and occupations between both types of workers when under firm class 10, male and female. Both groups, regardless of gender, possess a high concentration of college degree individuals, with ``type 10 workers'' having slightly more. Age is also similar, whereas less than 30 years old female workers are more likely to be found under type 3, while 31-50 are marginally more likely to be found under type 10.

For occupation, while both types display a higher concentration of scientific and liberal arts \footnote{These categories are generated from the Brazilian Code for Occupation, and tend to have similarities with other codes internationally. ``Scientific and Liberal Arts' is a generic code that summarizes economists, engineers, lawyers, professors, among others, whose jobs under normal circumstances require at least a degree from a superior institution of learning or education.}, there is a much higher concentration of managers.

Even with similarities, there is some evidence that individuals in complementarity effect matches might have leading positions and are particularly valuable for firms to employ. Such individuals would be suffering more extreme wage compressions if hired elsewhere, while worker type 10 experiences a more predictable wage path along firm classes.

\section{Discussions: Monte Carlo Simulation and Variance Decomposition} \label{sec:discussion}

In this section, I introduce a novel decomposition of the gender wage gap that accounts for complementarity effects in the labor market. I decompose the gender wage gap into three distinct components. The first component captures the contribution of complementarity effects, which I isolate by constructing a counterfactual labor market without comparative (dis)advantage matches.

The remaining two components are inspired by \textcite{card_bargaining_2016}. The second component, referred to as the ``sorting'' component, reflects the impact of firm allocation on the wage gap. I calculate this by simulating male and female labor markets where all factors are held constant except for the distribution of firms.

The final component, the ``bargaining'' component, represents the wage gap contribution arising when equally productive individuals are employed by firms of the same class, but a gender-based differential persists. I isolate this effect through a simulation in which male and female labor markets share identical distributions of workers and firms, while the means and variances of each gender-wise Gaussian distribution remains as observed in the original data.

The ``bargaining'' and ``complementarity'' components share certain similarities in nature. The complementarity effect can be viewed as a subset of the bargaining effect in the context of a CCK framework, as both reflect differences in returns for similar individuals within the same firm. However, the Gaussian mixture model allows me to distinguish between these two components, as it identifies labor market matches where wages deviate from the assumption of additive separability. As a result, complementarity effects emerge only in these specific labor market settings, whereas bargaining is more applicable in contexts where the additive separability condition holds.

\subsection{Monte Carlo Simulations}

To setup the Monte Carlo Simulations, I first calculate the realized moments of every worker type and firm class match in the labor market for the male and the female sample. Then I calculate the unconditional probabilities of worker types and firm classes for male and female\footnote{Means and the standard deviations of each match are shown in Figure \ref{fig:genderpayschedules}. The unconditional probabilities are shown in Figure \ref{fig:priors}.}

To create a separable market, I match workers following a ``diagonal pattern in matches'' in Figure \ref{fig:gaussianpred}. That means type worker 10 is guaranteed to work in firm class 10 as long as there is a spot available. When firm class 10 job slots are filled, firm class 9 starts hiring the best available, until all jobs are filled with workers. Figure \ref{fig:cond_separable_market} shows the resulting conditional probabilities of worker types given firm classes and gender under a separable market.

 In a labor market characterized by strictly additive separability, reshuffling matches is expected to have negligible effects on overall wage levels \parencite{graham_complementarity_2014}. Therefore, I leverage this fact, and the fact that ``diagonal'' matches do not yield large complementarity effects, to construct a labor market that behaves under the additive separability assumption. As a robustness check, I also perform my analysis using the weighted linear regression predicted fixed effects shown in Figure \ref{fig:linearpred}.

\begin{figure}[H]
    \centering
    \includegraphics[width = \textwidth]{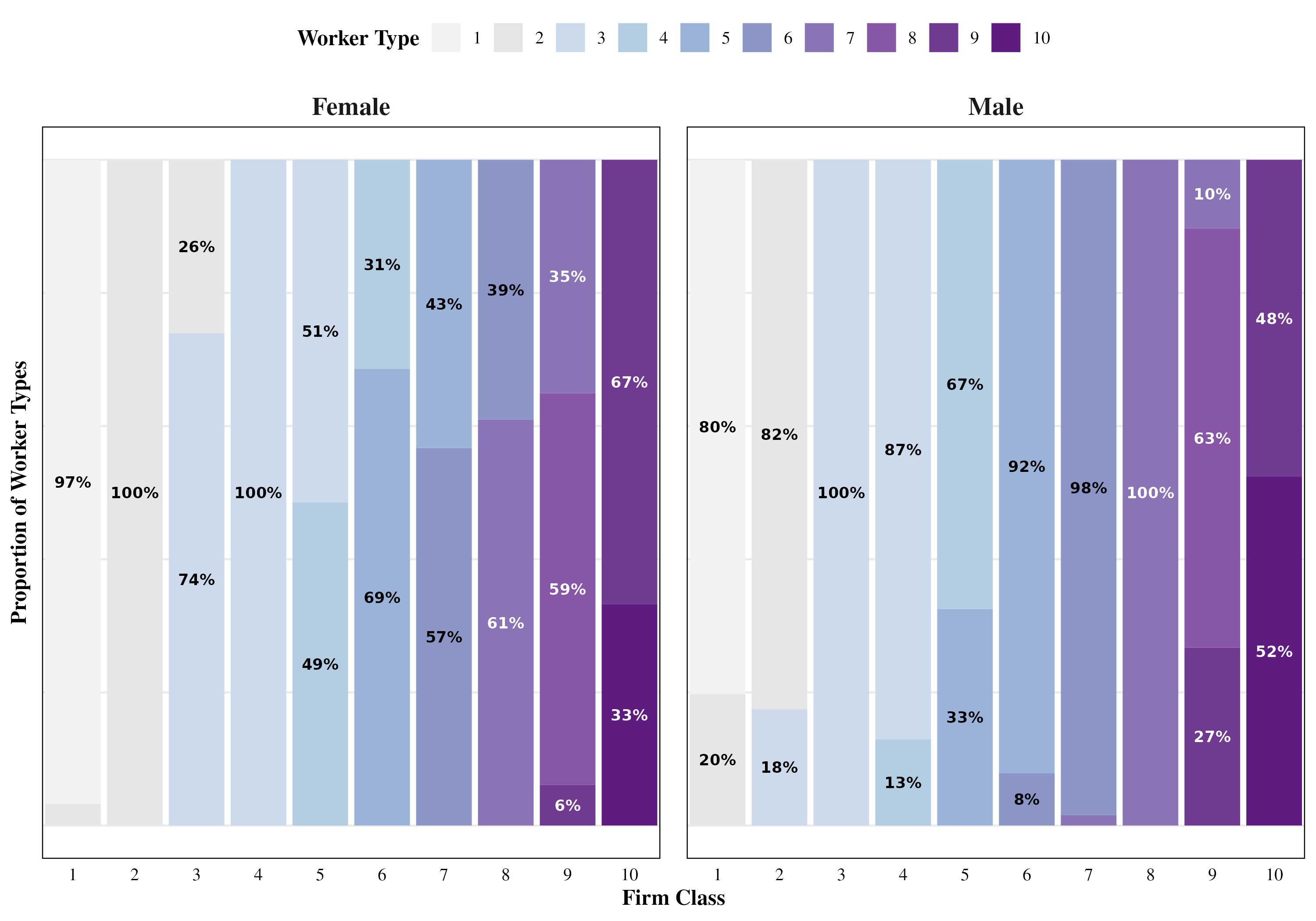}
    \caption{Conditional Probabilities of Worker Types Given Firm Classes and Gender, Under a Separable Market}
    \caption*{\small \textit{Note:} Probabilities calculated by generating a labor market where ``top workers'' are guaranteed to match with ``top firms'', until all positions are filled, following \textcite{becker_theory_1973}'s principle of assortative matching.}
    \label{fig:cond_separable_market}
\end{figure}

The difference in the gender wage gap between the separable labor market and the original setting reflects the contribution of complementarity effects that are not captured in the separable model. For CCK's ``bargaining'' and ``sorting'' components, I conduct simulations within the separable labor market to ensure that these components remain distinct and do not overlap.

\subsubsection{Simulation Results}

Table \ref{tab:wage_gap_decomposition} presents the results of the simulations. The first row displays the observed mean log hourly wages for female and male workers, along with the difference in log points, for the overall dataset, as well as broken down by education cohort and age group. The first column reports the difference in means calculated using the Gaussian distribution, which closely mirrors the observed gender wage gap. The second column reflects the gender wage gap in the counterfactual world where the labor market is additively separable.

The next three columns represent the contributions of the three proposed components. The first column shows the difference between the baseline and the separable market, indicating the contribution of complementarity effects to the gender wage gap. The ``sorting'' component reflects the contribution of firm allocation to the wage gap, holding all else constant in the separable market except the distribution of firms. Lastly, the ``bargaining'' component captures the contribution to the wage gap when all else is held constant in the separable market, except for the observed moments (means and variances) of each Gaussian distribution.

\begin{table}[htp!]
\centering
\begin{threeparttable}
\caption{Gaussian Mixture Decomposition of Gender Wage Gaps}
\label{tab:wage_gap_decomposition}
\begin{tabular}{lccccc}
\toprule
& & & \multicolumn{3}{c}{Contribution to Gender Wage Gap} \\
\cmidrule(lr){4-6}
        & Baseline & Separable & Complementarity & Sorting & Bargaining \\
 Group  & Market Gap      & Market Gap & Contribution   & Contribution & Contribution \\
        & (1) & (2) & (3) & (4) & (5)  \\
\hline
All & -0.24 & -0.20 & -0.04 & -0.09 & -0.02 \\  &  &  & (16.7\%) & (37.5\%) & (8.3\%) \\
\hline
\multicolumn{6}{l}{\textit{Education}} \\
No high-school & -0.30 & -0.27 & -0.03 & -0.13 & -0.03 \\
&  &  & (10.0\%) & (43.3\%) & (10.0\%) \\
\addlinespace
High-school & -0.23 & -0.22 & -0.01 & -0.08 & -0.02 \\
&  &  & (4.3\%) & (34.7\%) & (8.7\%) \\
\addlinespace
College & -0.35 & -0.30 & -0.05 & -0.16 & -0.04 \\
&  &  & (14.3\%) & (45.7\%) & (11.4\%) \\
\hline
\multicolumn{6}{l}{\textit{Age}} \\
<30 & -0.09 & -0.08 & -0.01 & -0.03 & -0.02 \\
&  &  & (11.1\%) & (33.3\%) & (22.2\%) \\
\addlinespace
31-50 & -0.30 & -0.26 & -0.04 & -0.13 & -0.03 \\ 
&  &  & (13.3\%) & (43.3\%) & (10.0\%) \\
\addlinespace
51> & -0.33 & -0.28 & -0.05 & -0.13 & -0.03 \\ 
&  &  & (15.2\%) & (39.4\%) & (9.1\%) \\
\bottomrule
\end{tabular}
\begin{tablenotes}
\small
\item \textit{Notes:} \textsuperscript{1}All values represent log wage gaps (female - male). Baseline Gap is the observed gap. \textsuperscript{2}Separable Market Gap assumes interactions do not yield complementarity effects. \textsuperscript{3}Complementarity Contribution is the difference between Baseline and Separable Market gaps. \textsuperscript{4}Sorting Contribution is the reduction in gap after equalizing means and variances of worker-firm interactions. \textsuperscript{5}Bargaining Contribution is the reduction in the gap after equalizing firm probabilities. \textsuperscript{7}Both sorting and bargaining are calculated under a separable market. \textsuperscript{6}Numbers in parentheses show the percentage of the Baseline Gap explained by each component.
\end{tablenotes}
\end{threeparttable}
\end{table}

Overall, the complementarity effect contributes approximately 16 to 17 percent to the gender wage gap, indicating that disparities arising from comparative advantage matches play a significant role in generating wage differences between male and female workers. Labor market allocation accounts for about 37.5 percent, while differences in bargaining without considering complementarities contribute roughly 8.3 percent. Together, these components explain nearly two-thirds of the gender wage gap. These findings demonstrate that non-separable, two-sided heterogeneity models, such as the Gaussian mixture approach of BLM, more effectively capture the substantial role that firms play in contributing to the gender wage gap, both horizontally as vertically.

The gender wage gap is smaller among individuals with a high school education but reaches its peak among those with college degrees. Furthermore, changes in the gender wage gap are not primarily driven by firm allocations; instead, they are largely explained by complementarity effects, especially for individuals with college degrees. This suggests that women with high levels of human capital, who are in positions of comparative advantage, are particularly susceptible to wage disparities arising from these effects.

Another evidence of the human capital accumulation and complementarity effect positive correlation is the age analysis. While sorting increases its role in wage differentials for individuals older than 30 years, it stabilizes for older than 51, while complementarity effect contribution keeps increasing slightly.

\begin{table}[htp!]
\centering
\begin{threeparttable}
\caption{Gaussian Mixture Decomposition of Gender Wage Gaps - Firm sizes and occupations}
\label{tab:wage_gap_decomposition_firm}
\begin{tabular}{lccccc}
\toprule
& & & \multicolumn{3}{c}{Contribution to Gender Wage Gap} \\
\cmidrule(lr){4-6}
        & Baseline & Separable & Complementarity & Sorting & Bargaining \\
 Group  & Market Gap      & Market Gap & Contribution   & Contribution & Contribution \\
        & (1) & (2) & (3) & (4) & (5)  \\
\hline
All & -0.24 & -0.20 & -0.04 & -0.09 & -0.02 \\  &  &  & (16.7\%) & (37.5\%) & (8.3\%) \\
\hline
\multicolumn{6}{l}{\textit{Firm Size}} \\
Firms <10 & -0.12 & -0.12 & 0.00 & -0.02 & -0.03 \\  &  &  & (0.0\%) & (16.7\%) & (25.0\%) \\
\addlinespace
Firms 10-50 & -0.14 & -0.13 & -0.01 & -0.03 & -0.01 \\  &  &  & (7.1\%) & (21.4\%) & (7.1\%) \\
\addlinespace
Firms 51> & -0.21 & -0.17 & -0.04 & -0.08 & -0.02 \\  &  &  & (19.0\%) & (38.1\%) & (9.5\%) \\
\hline
\multicolumn{6}{l}{\textit{Occupations}} \\
Hotel and & -0.12 & -0.12 & 0.00 & -0.04 & -0.03 \\  
Restaurants &  &  & (0.0\%) & (33.3\%) & (25.0\%) \\
\addlinespace
Economists  & -0.39 & -0.35 & -0.04 & -0.16 & -0.05 \\ 
 and Engineers &  &  & (10.3\%) & (41.0\%) & (12.8\%) \\
\addlinespace
Managers & -0.33 & -0.22 & -0.11 & -0.08 & -0.03 \\  &  &  & (33.3\%) & (24.2\%) & (9.1\%) \\
\bottomrule
\end{tabular}
\begin{tablenotes}
\small
\item \textit{Notes:} \textsuperscript{1}All values represent log wage gaps (female - male). Baseline Gap is the observed gap. \textsuperscript{2}Separable Market Gap assumes interactions do not yield complementarity effects. \textsuperscript{3}Complementarity Contribution is the difference between Baseline and Separable Market gaps. \textsuperscript{4}Sorting Contribution is the reduction in gap after equalizing means and variances of worker-firm interactions. \textsuperscript{5}Bargaining Contribution is the reduction in the gap after equalizing firm probabilities. \textsuperscript{7}Both sorting and bargaining are calculated under a separable market. \textsuperscript{6}Numbers in parentheses show the percentage of the Baseline Gap explained by each component.
\end{tablenotes}
\end{threeparttable}
\end{table}

To further expand my study, I present Table \ref{tab:wage_gap_decomposition_firm} which shows the results of simulations based on samples of different firm sizes and occupations particularly relevant to this study. I focus on three categories of occupations. The first, "Hotels and Restaurants," includes workers directly involved in the hospitality sector, such as waiters, kitchen staff, and cleaners. 

The second category, ``economists and engineers'', is self-explanatory. The rationale for selecting these professions lies in the fact that, in Brazil, these fields are highly regulated, meaning that only individuals with the appropriate college degree are legally permitted to practice. This choice offers two key advantages: it controls for college diplomas that are more uniform in their practice than, for instance, medical doctors, but also focusing on degrees that typically lead to higher compensation in the labor market.

Finally, the "Managers" category includes workers in leadership positions. These workers are likely to possess high levels of firm-specific human capital, giving them a strong comparative advantage in the labor market.

If comparative advantage drives complementarity effects on the gender wage gap, then strategic positions in the labor market, those involving valuable human capital accumulation and leadership roles, and the size of the firm, associated with bargaining power, should reveal particularly high levels of these effects.

The firm size panel of Table \ref{tab:wage_gap_decomposition_firm} reveals that the gender wage gap increases with firm size, particularly due to the complementarity and the sorting effect, meaning women not only are more likely to be found in low-paying firms when controlling for larger firms, but also in positions where men are receiving much higher complementarity compensations. 

The last three rows of Table \ref{tab:wage_gap_decomposition_firm} represent the results of the simulation of different occupations. Hotel and restaurants are typically occupations assumed in the literature to possess high turnover rate and zero firm premium in wages \textcite{card_bargaining_2016, casarico_what_2024}. Therefore, these occupations are expected to have negligible comparative advantage effects. Accordingly, my results suggest that the hotel and restaurants labor market is mostly governed by the additive separable assumption, given that the simulated separable market yielded the exact same wage differentials as the baseline market gap\footnote{For the male and female wage levels for all simulations, refer to Table \ref{tab:wagelevels}}, confirming that there is no complementarity effect. However, as the expected human capital accumulation is increased, the gap increases. While a considerable portion of the gap is due to firm allocations for economists and engineers, the complementarity contribution represents 10 percent of the gender wage gap.

For managers, the distance in wage differentials between the baseline market and the separable market is the largest, with the complementarity contribution accounting for about 33 percent of the gap. Moreover, the sorting contribution drastically reduces, from 16 to 8 log-points, falling from 41 percent in contribution to 24 percent.

Under an additive separable model, such as linear regression, the results would suggest that labor market allocations are the primary drivers of the gender wage gap among managers\footnote{In another example, \textcite{card_bargaining_2016} found that firm-related factors contributed approximately 4 percent to the gender wage gap after controlling for managers.}. However, in a non-separable model, I can identify that a substantial portion of the previously unexplained differential is due to specific labor market matches that generate complementarity effects. Because additive separable models assume constant returns to unobserved heterogeneity of workers and firms, these contributions are difficult to capture accurately.

\subsection{Robustness Checks}

I perform a series of exercises to show my results are not sensitive to particular choice of parameters. While keeping the optimal number of firm classes according to the gap statistics ($K = 10$), I vary the number of worker types, which are the number of Gaussians observed in each firm class. I test with $L = 6$, $L = 10$, and $L = 12$. I also provide an alternative simulation of the separable market where I use a weighted ordinary least squares with firm classes and worker types as fixed effects. Following Equation \ref{eq:weightedols}, the weighted parameter is the fraction of workers of each worker type-firm class match. For the weighted OLS, I maintain $L = 10$.

For the alternative number of worker types, results were consistent across all specifications, with the exception of $L = 12$, that seemed to underestimate complementarity effects, putting slightly more contributions to the sorting and the bargaining contribution. The alternative separable market simulation yielded virtually the same estimates as the original, with a more conservative estimation of the complementarity effects contribution to the gender wage gap.

Despite some differences in estimation, overall, the results indicate that my measurements are not driven by errors arising from the Gaussian mixture estimation, local maxima, or a particular setting.

z

\subsection{Variance Decomposition}

A variance decomposition of log wages in works related to \textcite{abowd_high_1999} (AKM models). It decomposes the variance of log wages into five distinct components: (1) the contribution of worker fixed effects, (2) the contribution of firm fixed effects, (3) twice the covariance between worker and firm effects, (4) the variance of time-varying covariates and their associated covariances, typically captured by period dummies interacted with time-varying human capital indicators, and (5) the residual variance.

The AKM model, however, tends to negatively correlate worker and firm effects \parencite{andrews_high_2008}, implying a downward bias estimate for assortative matching. \textcite{bonhomme_distributional_2019} proposed using the dimension reduction technique relying in the Gaussian mixture model to mitigate the bias.

\textcite{card_bargaining_2016} found that approximately 10 percent of wage variance can be attributed to assortative matching for both male and female workers. In this section, I apply the framework of \textcite{bonhomme_distributional_2019} to examine the extent to which assortative matching may be underestimated in the wage variance decomposition used in AKM gender analysis.

I compare three models. The first specification uses clustered firm and individual worker identifiers, related to \textcite{bonhomme_grouped_2015}, this clustering approach allows the researcher to maintain the linear assumption but reduces the negative bias in assortative matching. I name this model ``clustered AKM'', or C-AKM, in which fixed effects for firms are now the firm classes\footnote{See Appendix Section \ref{sec:akm_discussion} and \ref{sec:CAKM} for discussions on the bias (also often dubbed ``limited mobility bias'' in the literature) and using the clustered AKM method to perform a KOB decomposition on the gender wage gap.}. In the second setting I employ the full BLM approach by leveraging both worker types and firm classes.

Finally, I test the variance decomposition analysis under a classical CCK approach which uses individual firm and worker identifiers as fixed effects. It requires the largest dual connected set of firms and bias correction. In this study, I use the bootstrapping approximation of \textcite{azkarate-askasua_correcting_2023} to correct the assortative matching bias.

Formally, the regression setting is:
\begin{equation} \label{eq:mainspec}
    w_{it} = \Phi^g_{K(i,t)} + \Lambda^g_{L(i)} + x'_{it} \beta + \varepsilon_{it}
\end{equation}

And the variance decomposition can be formally stated as:
\begin{align}
\underbrace{\text{Var}(w_{it})}_{\text{Log Hourly Wage Variance}} = & \underbrace{\text{ Var}(\Phi^g_{K(i,t)})}_{\text{Firm Class Variance}}+ \underbrace{\text{ Var}(\Lambda^g_{L(i)})}_{\text{Worker Type Variance}} + \underbrace{\text{ Var}(\varepsilon_{it})}_{\text{Residual Variance}} \\
&  + \underbrace{\text{Var}(x'_{it} \beta) + 2 \cdot \text{Cov}(\Lambda^g_{L(i)}, x'_{it} \beta) + 2 \cdot \text{Cov}(\Phi^g_{K(i,t)}, x'_{it} \beta)}_{\text{Time-Varying Covariates Variance and Associated Covariances}} \notag \\
& + \underbrace{2 \cdot \text{Cov}(\Phi^g_{K(i,t)}, \Lambda^g_{L(i)})}_{\text{Worker Type and Firm Class Covariance}} \notag
\end{align}
where $w_{it}$ represents the logarithmic hourly wage of worker $i$ in period $t$, decomposed as follows: $\Phi^g_{K(i,t)}$ represents the firm effects, where $K(i,t)$ is the assignment function. Here, $K$ can denote either firm classes or individual firm identifiers. The term $\Lambda^g_{L(i)}$ represents individual worker or worker type effects, where $L(i)$ can refer to either a worker type or an individual identifier. Time-varying covariates are represented by $x'_{it} \beta$, while gender heterogeneity is accounted for by the superscript $g$. Finally, $\varepsilon_{it}$ denotes the idiosyncratic error term.

The covariance between worker and firm effects is of particular interest in understanding the dynamics of assortative matching and its impact on the gender wage gap. This component can be potentially underestimated due to the limited mobility bias. I test three variance decompositions from three different settings. The C-AKM model, by clustering firms, potentially reduces the noise in firm effect estimates, allowing for a more stable estimation of the worker-firm covariance, however, it still relies on individual fixed effects. The CCK model under the bootstrapping correction provides a ``lower bound'' of these estimates, given that the bootstrapping correction is an approximation, not a total mitigation. The BLM model, employing both worker and firm clusters, provides a framework that effectively circumvents the limited mobility bias by coarsening job movements in the dataset at the cost of noisier results.

Table \ref{tab:variance_decomposition} presents the variance decomposition results. The first two columns show the results for the BLM decomposition, columns (3) and (4), the clustered firm AKM methodology. Finally, columns (5) and (6) represents the results for the classical AKM approach from CCK.

\begin{table}[htp!]
\centering
\begin{threeparttable}
\caption{Variance Decomposition of log hourly Wages}
\label{tab:variance_decomposition}
\small
\begin{tabular}{lcccccc}
\toprule
 & \multicolumn{2}{c}{BLM} & \multicolumn{2}{c}{Clustered AKM} & \multicolumn{2}{c}{CCK}\\
\cmidrule(lr){2-3} \cmidrule(lr){4-5} \cmidrule(lr){6-7}
  & \multicolumn{1}{c}{Female} & \multicolumn{1}{c}{Male} & \multicolumn{1}{c}{Female} & \multicolumn{1}{c}{Male} & \multicolumn{1}{c}{Female} & \multicolumn{1}{c}{Male}\\
& (1) & (2) & (3) & (4) & (5) & (6) \\
\midrule
\addlinespace
Var(log hourly wage) & 0.529 & 0.650 & 0.529 & 0.650 & 0.642 & 0.797 \\
\addlinespace
\midrule
\multicolumn{7}{l}{\textit{Panel A: \textbf{Variance Estimates}}} \\
\midrule
Firm effects & 0.046 & 0.066 & 0.020 & 0.023 & 0.042 & 0.044 \\
Worker effects & 0.367 & 0.357 & 0.377 & 0.467 & 0.479 & 0.626 \\
Time-varying covariates & 0.008 & 0.011 & 0.009 & 0.011  & 0.006 & 0.008 \\
Cov(Worker, Firm) & 0.137 & 0.122 & 0.114 & 0.136 & 0.106 & 0.108 \\
Residual & 0.062 & 0.094 & 0.010 & 0.012 & 0.009 & 0.011 \\
\addlinespace
\midrule
\multicolumn{7}{l}{\textit{Panel B: \textbf{Share of Total Variance (\%)}}} \\
\midrule
Firm effects & 8.8 & 10.2 & 3.7 & 3.6 & 6.6 & 5.5 \\
Worker effects & 55.2 & 54.9 & 71.2 & 71.9 & 74.6 & 78.5 \\
Time-varying covariates & 1.6 & 1.6 & 1.6 & 1.7 & 1.0 & 1.1\\
Cov(Worker, Firm) & 26.0 & 18.8 & 21.5 & 21.0  & 16.5 & 13.6\\
Residual & 11.6 & 14.5 & 1.9 & 1.8 & 1.4 & 1.3\\
\bottomrule
\end{tabular}
\begin{tablenotes}
\small
\item \textit{Notes:} \textsuperscript{1}AKM, BLM, and CCK stand for \textcite{abowd_high_1999, bonhomme_distributional_2019, card_bargaining_2016}, respectively. \textsuperscript{2}Clustered AKM represents firm clustered using a kmeans algorithm and individual worker identifiers as parameters. \textsuperscript{3}Panel A showcases the magnitude of estimated variance components, while Panel B presents these components as percentages of the total log hourly wage variance. \textsuperscript{4}Results are a weighted average based on the six biennial samples' number of observations.
\end{tablenotes}
\end{threeparttable}
\end{table}

The first row presents the total variance of log hourly wages by gender. Both the C-AKM and BLM models yield similar magnitudes, as they utilize the full set of worker observations. In contrast, the CCK model relies on the connected set of firms through job movers, which tends to overrepresent larger firms, resulting in higher wage variance estimates.

Male firm effect contribution to the wage variance ranges from 5 percent in the male sample under CCK, to 10 percent under BLM. On the other hand, female firm effect contribution ranges from 3.6 in the clustered AKM model, to 8.8 percent under BLM.

Although AKM models attribute the largest portion of wage variance to worker effects (over 70 percent) for both genders, in the BLM approach worker effects account for a smaller, though still significant, share of wage variance: 41.3 percent for women and 32.1 percent for men. 

The reduction in worker effect contribution can be explained by the larger worker-firm covariance contribution, at 26.0 and 18.8 percent of the total variance for women and men. Compared to the ``AKM'' model, the BLM and the C-AKM model were more effective in capturing assortative matching effects.

Worker-firm covariance is consistently higher for women across all specifications, suggesting that assortative matching is indeed a meaningful contributor to wage dispersion in the labor market, especially for female workers. My results are particularly relevant in the context of recent discussions on the rise of assortativity in labor markets, as highlighted by \textcite{song_firming_2019}, who documented the increasing trend of assortative matching in the United States. If men are systematically more likely to find high-paying matches in the labor market, while women are concentrated in lower-paying positions, this could exacerbate the gender wage gap.

\section{Conclusion} \label{sec:conclusion}

Additive separable models are unable to capture labor market interactions that generate wages based on comparative advantage, where the match between firms and workers results in compensation that exceeds (or are less than) the simple sum of individual worker and firm contributions.

In this paper, I deviate from linear additive models. I use a linked employer-employee dataset covering all firms and workers from São Paulo, Brazil (2010-2017), to apply the two-sided unobserved heterogeneity framework introduced by \textcite{bonhomme_distributional_2019}. This approach allows me to investigate the contribution of specific worker-firm interactions to the gender wage gap, assuming each interaction generates wages drawn from log-normal distributions. This method allows me to capture the complementarity effects arising from particular worker-firm assortative matching.

Employing Monte Carlo Simulations, I propose a novel decomposition of the gender wage gap into three components. Following \textcite{card_bargaining_2016}, the sorting component, representing labor market allocations, and bargaining component, representing differences in negotiation of equally productive workers under the same firm. The third component, the complementarity component, is a special case of ``bargaining'', however, under these matches wage levels do not correspond to the predicted from additive separable model.

I find a positive relationship between human capital and these complementarity effects. They are more pronounced for male workers compared to female workers, accounting for approximately 17 percent of the overall gender wage gap. This contribution go as far as a third of the gender wage gap for individuals in leadership positions. Controlling for occupations that generally require lower levels of human capital, such as occupations related to the hospitality sector, yielded negligible results.

I also find that these interactions are more present at the tails of the wage distribution where larger firms operate and wage dispersion is higher. For firms larger than 50 employees, about a fourth of the gender wage gap is explained by these complementarity effects.

My study demonstrates that firms and the broader labor market structure play a more significant role in shaping the gender wage gap than previously recognized. I demonstrate that firms not only provide varying wage premiums but also evaluate human capital and other worker characteristics in heterogeneous ways. This differential valuation of worker attributes across firms contributes substantially to gender-based wage disparities.

The pronounced complementarity effects observed in managerial positions suggest that policies aimed at increasing transparency in the labor market and promoting key leadership roles among female workers are essential for reducing gender wage disparities.

Future research could leverage on the increased availability of linked employer-employee data and computational power. They could extend the analysis by providing a dynamic framework and exploring how worker-firm interactions evolve over time in response to earnings shocks and their influence on mobility decisions. Incorporating collective bargaining data would further enhance our understanding of the non-monetary factors that shape gender-specific sorting patterns. Expanding this methodological approach to different countries could provide valuable cross-national insights into the extent to which gender wage gaps are driven by universal factors or are shaped by specific institutional and cultural contexts.
\clearpage
\singlespacing
\printbibliography

@article{mulligan_selection_2008,
	title = {Selection, {Investment}, and {Women}'s {Relative} {Wages} {Over} {Time}*},
	volume = {123},
	issn = {0033-5533},
	url = {https://doi.org/10.1162/qjec.2008.123.3.1061},
	doi = {10.1162/qjec.2008.123.3.1061},
	number = {3},
	urldate = {2024-09-30},
	journal = {The Quarterly Journal of Economics},
	author = {Mulligan, Casey B. and Rubinstein, Yona},
	month = aug,
	year = {2008},
	pages = {1061--1110},
}

@article{goldin_homecoming_2006,
	title = {The {Homecoming} of {American} {College} {Women}: {The} {Reversal} of the {College} {Gender} {Gap}},
	volume = {20},
	issn = {0895-3309},
	shorttitle = {The {Homecoming} of {American} {College} {Women}},
	url = {https://www.aeaweb.org/articles?id=10.1257%2Fjep.20.4.133&fbclid=IwAR2YgNWGj5pMpZpEpG5a97ahef4Dwi8Wk7rUpLqqRGPMqwOGOnJYmGvcLys},
	doi = {10.1257/jep.20.4.133},
	language = {en},
	number = {4},
	urldate = {2024-09-12},
	journal = {Journal of Economic Perspectives},
	author = {Goldin, Claudia and Katz, Lawrence F. and Kuziemko, Ilyana},
	month = dec,
	year = {2006},
	pages = {133--156},
}

@article{ceci_women_2014,
	title = {Women in {Academic} {Science}: {A} {Changing} {Landscape}},
	volume = {15},
	issn = {1529-1006},
	shorttitle = {Women in {Academic} {Science}},
	url = {https://doi.org/10.1177/1529100614541236},
	doi = {10.1177/1529100614541236},
	language = {en},
	number = {3},
	urldate = {2024-09-12},
	journal = {Psychological Science in the Public Interest},
	author = {Ceci, Stephen J. and Ginther, Donna K. and Kahn, Shulamit and Williams, Wendy M.},
	month = dec,
	year = {2014},
	note = {Publisher: SAGE Publications Inc},
	pages = {75--141},
}

@article{black_gender_2008,
	title = {Gender {Wage} {Disparities} among the {Highly} {Educated}},
	volume = {43},
	copyright = {© 2008 by the Board of Regents of the University of Wisconsin System},
	issn = {0022-166X, 1548-8004},
	url = {https://jhr.uwpress.org/content/43/3/630},
	doi = {10.3368/jhr.43.3.630},
	language = {en},
	number = {3},
	urldate = {2024-09-12},
	journal = {Journal of Human Resources},
	author = {Black, Dan A. and Haviland, Amelia M. and Sanders, Seth G. and Taylor, Lowell J.},
	month = jul,
	year = {2008},
	note = {Publisher: University of Wisconsin Press
Section: Articles},
	pages = {630--659},
}

@incollection{blau_womens_2008,
	address = {London},
	title = {Women’s {Work} and {Wages}},
	isbn = {978-1-349-95121-5},
	url = {https://doi.org/10.1057/978-1-349-95121-5_2207-1},
	language = {en},
	urldate = {2024-09-12},
	booktitle = {The {New} {Palgrave} {Dictionary} of {Economics}},
	publisher = {Palgrave Macmillan UK},
	author = {Blau, Francine D. and Kahn, Lawrence M.},
	year = {2008},
	doi = {10.1057/978-1-349-95121-5_2207-1},
	pages = {1--14},
}

@article{lavetti_gender_2023,
	title = {Gender differences in sorting on wages and risk},
	volume = {233},
	issn = {0304-4076},
	url = {https://www.sciencedirect.com/science/article/pii/S030440762200183X},
	doi = {10.1016/j.jeconom.2022.06.012},
	number = {2},
	urldate = {2023-10-07},
	journal = {Journal of Econometrics},
	author = {Lavetti, Kurt and Schmutte, Ian},
	month = apr,
	year = {2023},
	pages = {507--523},
}

@article{card_introduction_2023,
	title = {Introduction to the {Special} {Issue}: {Models} of linked employer–employee data: {Twenty} years after “{High} {Wage} {Workers} and {High} {Wage} {Firms}”},
	volume = {233},
	issn = {03044076},
	shorttitle = {Introduction to the {Special} {Issue}},
	url = {https://linkinghub.elsevier.com/retrieve/pii/S0304407623000337},
	doi = {10.1016/j.jeconom.2023.01.012},
	language = {en},
	number = {2},
	urldate = {2024-08-05},
	journal = {Journal of Econometrics},
	author = {Card, David and Schmutte, Ian and Vilhuber, Lars},
	month = apr,
	year = {2023},
	pages = {333--339},
}

@article{lachowska_firm_2023,
	title = {Do firm effects drift? {Evidence} from {Washington} administrative data},
	volume = {233},
	issn = {03044076},
	shorttitle = {Do firm effects drift?},
	url = {https://linkinghub.elsevier.com/retrieve/pii/S0304407622000604},
	doi = {10.1016/j.jeconom.2021.12.014},
	language = {en},
	number = {2},
	urldate = {2024-08-05},
	journal = {Journal of Econometrics},
	author = {Lachowska, Marta and Mas, Alexandre and Saggio, Raffaele and Woodbury, Stephen A.},
	month = apr,
	year = {2023},
	pages = {375--395},
}

@article{bertrand_dynamics_2010,
	title = {Dynamics of the {Gender} {Gap} for {Young} {Professionals} in the {Financial} and {Corporate} {Sectors}},
	volume = {2},
	issn = {1945-7782},
	url = {https://www.aeaweb.org/articles?id=10.1257/app.2.3.228},
	doi = {10.1257/app.2.3.228},
	language = {en},
	number = {3},
	urldate = {2024-07-30},
	journal = {American Economic Journal: Applied Economics},
	author = {Bertrand, Marianne and Goldin, Claudia and Katz, Lawrence F.},
	month = jul,
	year = {2010},
	pages = {228--255},
}

@article{kitagawa_components_1955,
	title = {Components of a {Difference} {Between} {Two} {Rates}*},
	volume = {50},
	issn = {0162-1459},
	url = {https://doi.org/10.1080/01621459.1955.10501299},
	doi = {10.1080/01621459.1955.10501299},
	number = {272},
	urldate = {2024-07-30},
	journal = {Journal of the American Statistical Association},
	author = {Kitagawa, Evelyn M.},
	month = dec,
	year = {1955},
	note = {Publisher: Taylor \& Francis
\_eprint: https://doi.org/10.1080/01621459.1955.10501299},
	pages = {1168--1194},
}

@article{bertrand_breaking_2019,
	title = {Breaking the {Glass} {Ceiling}? {The} {Effect} of {Board} {Quotas} on {Female} {Labour} {Market} {Outcomes} in {Norway}},
	volume = {86},
	issn = {0034-6527},
	shorttitle = {Breaking the {Glass} {Ceiling}?},
	url = {https://doi.org/10.1093/restud/rdy032},
	doi = {10.1093/restud/rdy032},
	number = {1},
	urldate = {2024-07-30},
	journal = {The Review of Economic Studies},
	author = {Bertrand, Marianne and Black, Sandra E and Jensen, Sissel and Lleras-Muney, Adriana},
	month = jan,
	year = {2019},
	pages = {191--239},
}

@article{jewell_who_2020,
	title = {Who {Works} for {Whom} and the {UK} {Gender} {Pay} {Gap}},
	volume = {58},
	copyright = {© 2019 John Wiley \& Sons Ltd.},
	issn = {1467-8543},
	url = {https://onlinelibrary.wiley.com/doi/abs/10.1111/bjir.12497},
	doi = {10.1111/bjir.12497},
	language = {en},
	number = {1},
	urldate = {2024-07-21},
	journal = {British Journal of Industrial Relations},
	author = {Jewell, Sarah Louise and Razzu, Giovanni and Singleton, Carl},
	year = {2020},
	pages = {50--81},
}

@article{song_firming_2019,
	title = {Firming {Up} {Inequality}*},
	volume = {134},
	issn = {0033-5533},
	url = {https://doi.org/10.1093/qje/qjy025},
	doi = {10.1093/qje/qjy025},
	number = {1},
	urldate = {2024-07-13},
	journal = {The Quarterly Journal of Economics},
	author = {Song, Jae and Price, David J and Guvenen, Fatih and Bloom, Nicholas and von Wachter, Till},
	month = feb,
	year = {2019},
	pages = {1--50},
}

@article{kline_leaveout_2020,
	title = {Leave‐{Out} {Estimation} of {Variance} {Components}},
	volume = {88},
	issn = {0012-9682},
	url = {https://www.econometricsociety.org/doi/10.3982/ECTA16410},
	doi = {10.3982/ECTA16410},
	language = {en},
	number = {5},
	urldate = {2024-07-11},
	journal = {Econometrica},
	author = {Kline, Patrick and Saggio, Raffaele and Sølvsten, Mikkel},
	year = {2020},
	pages = {1859--1898},
}

@techreport{coudin_family_2018,
	type = {Working {Paper}},
	title = {Family, firms and the gender wage gap in {France}},
	copyright = {http://www.econstor.eu/dspace/Nutzungsbedingungen},
	url = {https://www.econstor.eu/handle/10419/200290},
	language = {eng},
	number = {W18/01},
	urldate = {2024-07-07},
	institution = {IFS Working Papers},
	author = {Coudin, Elise and Maillard, Sophie and Tô, Maxime},
	year = {2018},
	doi = {10.1920/wp.ifs.2018.W1801},
}

@article{cruz_effects_2022,
	title = {The effects of equal pay laws on firm pay premiums: {Evidence} from {Chile}},
	volume = {75},
	issn = {0927-5371},
	shorttitle = {The effects of equal pay laws on firm pay premiums},
	url = {https://www.sciencedirect.com/science/article/pii/S0927537122000288},
	doi = {10.1016/j.labeco.2022.102135},
	urldate = {2024-07-07},
	journal = {Labour Economics},
	author = {Cruz, Gabriel and Rau, Tomás},
	month = apr,
	year = {2022},
	pages = {102135},
}

@article{abowd_high_1999,
	title = {High {Wage} {Workers} and {High} {Wage} {Firms}},
	volume = {67},
	copyright = {Econometric Society 1999},
	issn = {1468-0262},
	url = {https://onlinelibrary.wiley.com/doi/abs/10.1111/1468-0262.00020},
	doi = {10.1111/1468-0262.00020},
	language = {en},
	number = {2},
	urldate = {2023-09-20},
	journal = {Econometrica},
	author = {Abowd, John M. and Kramarz, Francis and Margolis, David N.},
	year = {1999},
	pages = {251--333},
}

@article{bonhomme_distributional_2019,
	title = {A {Distributional} {Framework} for {Matched} {Employer} {Employee} {Data}},
	volume = {87},
	copyright = {© 2019 The Econometric Society},
	issn = {1468-0262},
	url = {https://onlinelibrary.wiley.com/doi/abs/10.3982/ECTA15722},
	doi = {10.3982/ECTA15722},
	language = {en},
	number = {3},
	urldate = {2023-08-21},
	journal = {Econometrica},
	author = {Bonhomme, Stéphane and Lamadon, Thibaut and Manresa, Elena},
	year = {2019},
	pages = {699--739},
}

@article{tibshirani_estimating_2001,
	title = {Estimating the {Number} of {Clusters} in a {Data} {Set} via the {Gap} {Statistic}},
	volume = {63},
	issn = {1369-7412},
	url = {https://www.jstor.org/stable/2680607},
	number = {2},
	urldate = {2024-07-05},
	journal = {Journal of the Royal Statistical Society. Series B (Statistical Methodology)},
	author = {Tibshirani, Robert and Walther, Guenther and Hastie, Trevor},
	year = {2001},
	note = {Publisher: [Royal Statistical Society, Wiley]},
	pages = {411--423},
}

@article{macqueen_methods_1967,
	title = {Some {Methods} for classification and analysis of multivariate observations},
	url = {https://projecteuclid.org/ebooks/berkeley-symposium-on-mathematical-statistics-and-probability/Proceedings-of-the-Fifth-Berkeley-Symposium-on-Mathematical-Statistics-and/chapter/Some-methods-for-classification-and-analysis-of-multivariate-observations/bsmsp/1200512992},
	language = {en},
	journal = {Proceedings of the Fifth Berkeley Symposium on Mathematical Statistics and Probability, Volume 1: Statistics},
	author = {MacQueen, James},
	year = {1967},
	note = {Google-Books-ID: IC4Ku\_7dBFUC},
}

@article{bruns_changes_2019,
	title = {Changes in {Workplace} {Heterogeneity} and {How} {They} {Widen} the {Gender} {Wage} {Gap}},
	volume = {11},
	issn = {1945-7782},
	url = {https://www.aeaweb.org/articles?id=10.1257/app.20160664},
	doi = {10.1257/app.20160664},
	language = {en},
	number = {2},
	urldate = {2024-05-29},
	journal = {American Economic Journal: Applied Economics},
	author = {Bruns, Benjamin},
	month = apr,
	year = {2019},
	pages = {74--113},
}

@article{card_workplace_2013,
	title = {Workplace {Heterogeneity} and the {Rise} of {West} {German} {Wage} {Inequality}*},
	volume = {128},
	issn = {0033-5533},
	url = {https://doi.org/10.1093/qje/qjt006},
	doi = {10.1093/qje/qjt006},
	number = {3},
	urldate = {2023-09-25},
	journal = {The Quarterly Journal of Economics},
	author = {Card, David and Heining, Joerg and Kline, Patrick},
	month = aug,
	year = {2013},
	pages = {967--1015},
}

@article{woodcock_wage_2008,
	series = {European {Association} of {Labour} {Economists} 19th annual conference / {Firms} and {Employees}},
	title = {Wage differentials in the presence of unobserved worker, firm, and match heterogeneity},
	volume = {15},
	issn = {0927-5371},
	url = {https://www.sciencedirect.com/science/article/pii/S0927537107000395},
	doi = {10.1016/j.labeco.2007.06.003},
	number = {4},
	urldate = {2024-05-17},
	journal = {Labour Economics},
	author = {Woodcock, Simon D.},
	month = aug,
	year = {2008},
	pages = {771--793},
}

@article{andrews_high_2008,
	title = {High {Wage} {Workers} and {Low} {Wage} {Firms}: {Negative} {Assortative} {Matching} or {Limited} {Mobility} {Bias}?},
	volume = {171},
	issn = {0964-1998},
	shorttitle = {High {Wage} {Workers} and {Low} {Wage} {Firms}},
	url = {https://doi.org/10.1111/j.1467-985X.2007.00533.x},
	doi = {10.1111/j.1467-985X.2007.00533.x},
	number = {3},
	urldate = {2024-05-17},
	journal = {Journal of the Royal Statistical Society Series A: Statistics in Society},
	author = {Andrews, M. J. and Gill, L. and Schank, T. and Upward, R.},
	month = jun,
	year = {2008},
	pages = {673--697},
}

@article{barth_assortative_2003,
	title = {Assortative matching in the labour market? {Stylised} facts about workers and plants},
	language = {en},
	journal = {Institute for Social Research, Oslo, Norway},
	author = {Barth, Erling and Dale-Olsen, Harald},
	year = {2003},
}

@article{gruetter_importance_2009,
	title = {The importance of firms in wage determination},
	volume = {16},
	issn = {0927-5371},
	url = {https://www.sciencedirect.com/science/article/pii/S0927537108001012},
	doi = {10.1016/j.labeco.2008.09.001},
	number = {2},
	urldate = {2024-05-17},
	journal = {Labour Economics},
	author = {Gruetter, Max and Lalive, Rafael},
	month = apr,
	year = {2009},
	pages = {149--160},
}

@article{graham_complementarity_2014,
	title = {Complementarity and aggregate implications of assortative matching: {A} nonparametric analysis},
	volume = {5},
	copyright = {Copyright © 2014 Bryan S. Graham, Guido W. Imbens, and Geert Ridder},
	issn = {1759-7331},
	shorttitle = {Complementarity and aggregate implications of assortative matching},
	url = {https://onlinelibrary.wiley.com/doi/abs/10.3982/QE45},
	doi = {10.3982/QE45},
	language = {en},
	number = {1},
	urldate = {2024-05-14},
	journal = {Quantitative Economics},
	author = {Graham, Bryan S. and Imbens, Guido W. and Ridder, Geert},
	year = {2014},
	note = {\_eprint: https://onlinelibrary.wiley.com/doi/pdf/10.3982/QE45},
	pages = {29--66},
}

@article{casarico_what_2024,
	title = {What {Firms} {Do}: {Gender} {Inequality} in {Linked} {Employer}-{Employee} {Data}},
	volume = {42},
	issn = {0734-306X},
	shorttitle = {What {Firms} {Do}},
	url = {https://www.journals.uchicago.edu/doi/abs/10.1086/723177},
	doi = {10.1086/723177},
	number = {2},
	urldate = {2024-04-08},
	journal = {Journal of Labor Economics},
	author = {Casarico, Alessandra and Lattanzio, Salvatore},
	month = apr,
	year = {2024},
	note = {Publisher: The University of Chicago Press},
	pages = {325--355},
}

@article{gerard_assortative_2021,
	title = {Assortative {Matching} or {Exclusionary} {Hiring}? {The} {Impact} of {Employment} and {Pay} {Policies} on {Racial} {Wage} {Differences} in {Brazil}},
	volume = {111},
	issn = {0002-8282},
	shorttitle = {Assortative {Matching} or {Exclusionary} {Hiring}?},
	url = {https://www.aeaweb.org/articles?id=10.1257/aer.20181596},
	doi = {10.1257/aer.20181596},
	language = {en},
	number = {10},
	urldate = {2024-03-07},
	journal = {American Economic Review},
	author = {Gerard, François and Lagos, Lorenzo and Severnini, Edson and Card, David},
	month = oct,
	year = {2021},
	pages = {3418--3457},
}

@article{bonhomme_how_2023,
	title = {How {Much} {Should} {We} {Trust} {Estimates} of {Firm} {Effects} and {Worker} {Sorting}?},
	volume = {41},
	issn = {0734-306X},
	url = {https://www.journals.uchicago.edu/doi/full/10.1086/720009},
	doi = {10.1086/720009},
	number = {2},
	urldate = {2024-03-07},
	journal = {Journal of Labor Economics},
	author = {Bonhomme, Stéphane and Holzheu, Kerstin and Lamadon, Thibaut and Manresa, Elena and Mogstad, Magne and Setzler, Bradley},
	month = apr,
	year = {2023},
	note = {Publisher: The University of Chicago Press},
	pages = {291--322},
}

@article{masso_role_2022,
	title = {The role of firms in the gender wage gap},
	volume = {50},
	issn = {0147-5967},
	url = {https://www.sciencedirect.com/science/article/pii/S0147596721000639},
	doi = {10.1016/j.jce.2021.10.001},
	number = {2},
	urldate = {2024-03-01},
	journal = {Journal of Comparative Economics},
	author = {Masso, Jaan and Meriküll, Jaanika and Vahter, Priit},
	month = jun,
	year = {2022},
	pages = {454--473},
}

@article{gallen_labor_2019,
	title = {The labor market gender gap in {Denmark}: {Sorting} out the past 30 years},
	volume = {56},
	issn = {0927-5371},
	shorttitle = {The labor market gender gap in {Denmark}},
	url = {https://www.sciencedirect.com/science/article/pii/S0927537118301234},
	doi = {10.1016/j.labeco.2018.11.003},
	urldate = {2024-03-01},
	journal = {Labour Economics},
	author = {Gallen, Yana and Lesner, Rune V. and Vejlin, Rune},
	month = jan,
	year = {2019},
	pages = {58--67},
}

@misc{oecd_oecd_2024,
	title = {{OECD} {Family} {Database} - {OECD}},
	url = {https://www.oecd.org/social/family/database.htm},
	urldate = {2024-03-01},
	author = {OECD},
	year = {2024},
}

@article{shimer_assortative_2000,
	title = {Assortative {Matching} and {Search}},
	volume = {68},
	copyright = {Econometric Society 2000},
	issn = {1468-0262},
	url = {https://onlinelibrary.wiley.com/doi/abs/10.1111/1468-0262.00112},
	doi = {10.1111/1468-0262.00112},
	language = {en},
	number = {2},
	urldate = {2024-02-16},
	journal = {Econometrica},
	author = {Shimer, Robert and Smith, Lones},
	year = {2000},
	note = {\_eprint: https://onlinelibrary.wiley.com/doi/pdf/10.1111/1468-0262.00112},
	pages = {343--369},
}

@article{becker_theory_1973,
	title = {A {Theory} of {Marriage}: {Part} {I}},
	volume = {81},
	issn = {0022-3808},
	shorttitle = {A {Theory} of {Marriage}},
	url = {https://www.journals.uchicago.edu/doi/abs/10.1086/260084},
	doi = {10.1086/260084},
	number = {4},
	urldate = {2024-02-11},
	journal = {Journal of Political Economy},
	author = {Becker, Gary S.},
	month = jul,
	year = {1973},
	note = {Publisher: The University of Chicago Press},
	pages = {813--846},
}

@article{eeckhout_assortative_2018,
	title = {Assortative {Matching} {With} {Large} {Firms}},
	volume = {86},
	copyright = {© 2018 The Econometric Society},
	issn = {1468-0262},
	url = {https://onlinelibrary.wiley.com/doi/abs/10.3982/ECTA14450},
	doi = {10.3982/ECTA14450},
	language = {en},
	number = {1},
	urldate = {2024-02-11},
	journal = {Econometrica},
	author = {Eeckhout, Jan and Kircher, Philipp},
	year = {2018},
	note = {\_eprint: https://onlinelibrary.wiley.com/doi/pdf/10.3982/ECTA14450},
	pages = {85--132},
}

@misc{azkarate-askasua_correcting_2023,
	address = {Rochester, NY},
	type = {{SSRN} {Scholarly} {Paper}},
	title = {Correcting {Small} {Sample} {Bias} in {Linear} {Models} with {Many} {Covariates}},
	url = {https://papers.ssrn.com/abstract=4322300},
	doi = {10.2139/ssrn.4322300},
	language = {en},
	urldate = {2024-02-05},
	author = {Azkarate-Askasua, Miren and Zerecero, Miguel},
	month = jan,
	year = {2023},
}

@article{abowd_computing_2002,
	title = {Computing {Person} and {Firm} {Effects} {Using} {Linked} {Longitudinal} {Employer}-{Employee} {Data}},
	url = {https://ideas.repec.org//p/cen/tpaper/2002-06.html},
	language = {en},
	urldate = {2023-12-08},
	journal = {Longitudinal Employer-Household Dynamics Technical Papers},
	author = {Abowd, John M. and Creecy, Robert H. and Kramarz, Francis},
	month = mar,
	year = {2002},
	note = {Number: 2002-06
Publisher: Center for Economic Studies, U.S. Census Bureau},
}

@article{gaure_correlation_2014,
	title = {Correlation bias correction in two-way fixed-effects linear regression},
	volume = {3},
	copyright = {Copyright © 2014 John Wiley \& Sons, Ltd.},
	issn = {2049-1573},
	url = {https://onlinelibrary.wiley.com/doi/abs/10.1002/sta4.68},
	doi = {10.1002/sta4.68},
	language = {en},
	number = {1},
	urldate = {2023-10-23},
	journal = {Stat},
	author = {Gaure, Simen},
	year = {2014},
	note = {\_eprint: https://onlinelibrary.wiley.com/doi/pdf/10.1002/sta4.68},
	pages = {379--390},
}

@article{bonhomme_grouped_2015,
	title = {Grouped {Patterns} of {Heterogeneity} in {Panel} {Data}},
	volume = {83},
	copyright = {© 2015 The Econometric Society},
	issn = {1468-0262},
	url = {https://onlinelibrary.wiley.com/doi/abs/10.3982/ECTA11319},
	doi = {10.3982/ECTA11319},
	language = {en},
	number = {3},
	urldate = {2023-09-25},
	journal = {Econometrica},
	author = {Bonhomme, Stéphane and Manresa, Elena},
	year = {2015},
	note = {\_eprint: https://onlinelibrary.wiley.com/doi/pdf/10.3982/ECTA11319},
	pages = {1147--1184},
}

@article{goldin_grand_2014,
	title = {A {Grand} {Gender} {Convergence}: {Its} {Last} {Chapter}},
	volume = {104},
	issn = {0002-8282},
	shorttitle = {A {Grand} {Gender} {Convergence}},
	url = {https://www.aeaweb.org/articles?id=10.1257/aer.104.4.1091},
	doi = {10.1257/aer.104.4.1091},
	language = {en},
	number = {4},
	urldate = {2023-09-20},
	journal = {American Economic Review},
	author = {Goldin, Claudia},
	month = apr,
	year = {2014},
	pages = {1091--1119},
}

@article{oaxaca_male-female_1973,
	title = {Male-{Female} {Wage} {Differentials} in {Urban} {Labor} {Markets}},
	volume = {14},
	issn = {0020-6598},
	url = {https://www.jstor.org/stable/2525981},
	doi = {10.2307/2525981},
	number = {3},
	urldate = {2023-08-22},
	journal = {International Economic Review},
	author = {Oaxaca, Ronald},
	year = {1973},
	note = {Publisher: [Economics Department of the University of Pennsylvania, Wiley, Institute of Social and Economic Research, Osaka University]},
	pages = {693--709},
}

@article{blinder_wage_1973,
	title = {Wage {Discrimination}: {Reduced} {Form} and {Structural} {Estimates}},
	volume = {8},
	issn = {0022-166X},
	shorttitle = {Wage {Discrimination}},
	url = {https://www.jstor.org/stable/144855},
	doi = {10.2307/144855},
	number = {4},
	urldate = {2023-08-22},
	journal = {The Journal of Human Resources},
	author = {Blinder, Alan S.},
	year = {1973},
	note = {Publisher: [University of Wisconsin Press, Board of Regents of the University of Wisconsin System]},
	pages = {436--455},
}

@article{card_bargaining_2016,
	title = {Bargaining, {Sorting}, and the {Gender} {Wage} {Gap}: {Quantifying} the {Impact} of {Firms} on the {Relative} {Pay} of {Women} *},
	volume = {131},
	issn = {0033-5533},
	shorttitle = {Bargaining, {Sorting}, and the {Gender} {Wage} {Gap}},
	url = {https://doi.org/10.1093/qje/qjv038},
	doi = {10.1093/qje/qjv038},
	number = {2},
	urldate = {2023-08-21},
	journal = {The Quarterly Journal of Economics},
	author = {Card, David and Cardoso, Ana Rute and Kline, Patrick},
	month = may,
	year = {2016},
	pages = {633--686},
}
\doublespacing

\clearpage

\begin{appendices}
\counterwithin{figure}{section}
\counterwithin{table}{section}

\section{Cluster Choice Analysis} \label{sec:clusterchoiceanalysis}
A common drawback of clustering methods is the optimal number of clusters. To increase the robustness of my analysis, I employ the gap statistic method, a widely used technique in cluster analysis \parencite{tibshirani_estimating_2001}, to explore the within variance of clusters and therefore, choose the right number. 

The gap statistic compares the total within-cluster variation for different values of $k$ with their expected values under a null reference distribution of the data. I calculate it through the following steps:

\begin{enumerate}
    \item For each number of clusters $K$, I compute the within-cluster sum of squares $W_k = \sum_{k=1}^K \sum_{i \in k} (w_i - \bar{w}_k)^2$ where $k$ is given cluster, $w_i$ is the observed worker $i$ log-weekly wage, $\bar{w}_k$ the empirical mean of cluster $k$ (centroid).
    
    \item I generate $B$ reference datasets by sampling uniformly from the range of my observed data. For each reference dataset, I compute $W_{kb}$, the within-cluster sum of squares when clustering the reference data into $k$ clusters.

    \item I then compute the gap statistic as:
    \[Gap(k) = \frac{1}{B} \sum_{b=1}^B \log(W_{kb}) - \log(W_k)\]
    
    \item I use the bootstrapped standard deviation as the standard error across the $B$ reference datasets. I implement this procedure with $B=500$ to ensure stable estimates.
    
    \item Finally, I choose the optimal number of clusters as the smallest $k$ such that:
    \[Gap(k) \geq Gap(k+1) - s_{k+1}\]
    where $s_{k+1}$ is the standard deviation for $k+1$ clusters.
\end{enumerate}

Figure \ref{fig:gap_stat} presents the gap statistic values for different numbers of clusters, ranging from 4 to 25. The green dashed line at 10 clusters represents what my analysis suggests as the optimal number of clusters based on the gap statistic. Due to the large number of observations, standard errors were negligible, barely noticeable in the plot. Nevertheless, there is a notable elbow in the curve, indicating diminishing returns to increasing the number of clusters beyond the value of 10, implying that the benefit of additional clusters in explaining data variability becomes less substantial.

\begin{figure}[htb!]
    \centering
    \includegraphics[width=\textwidth]{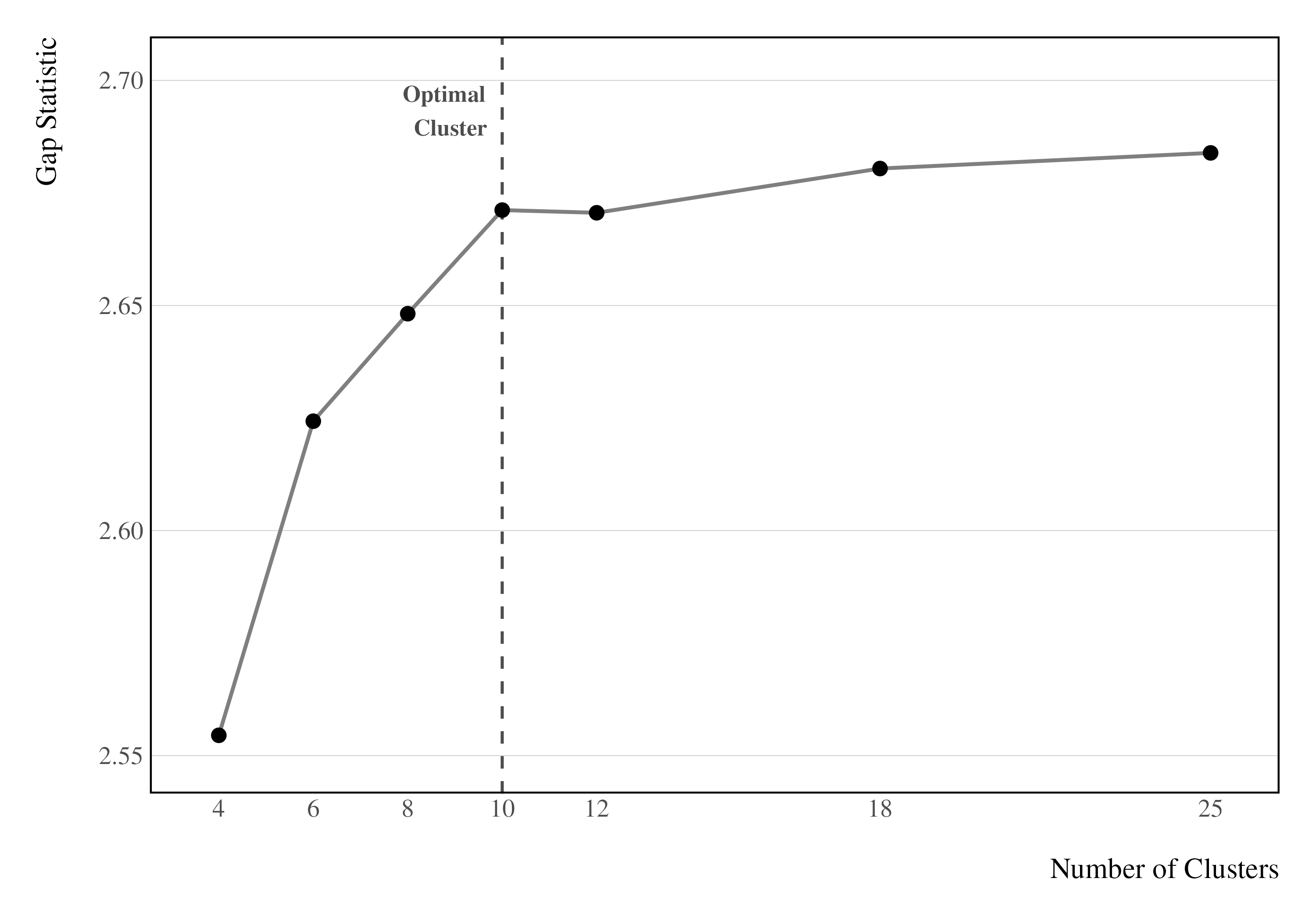}
    \caption{Point Estimate Gap Statistic by Number of Clusters}
    \caption*{\footnotesize \textit{Note:} \textsuperscript{1} Vertical line indicates the number of clusters of choice (10). \textsuperscript{2} Firm classes estimated by a k-means clustering algorithm using as measurement their empirical cumulative distribution function supported by the ventiles of the population, over six biennials (from 2010-12 to 2015-17). \textsuperscript{3} Bootstrapped standard errors were negligible due to the number of observations, with $p < 0.01$.}
    \label{fig:gap_stat}
\end{figure}

\clearpage
\section{Worker Mobility in Firm Clusters}

In this subsection, I explore the exogenous mobility assumption of the Gaussian mixture model. This assumption states that the movement of workers should be related to worker types and firm classes, but not directly on earnings. Therefore, the expected wage on unobservables should be zero for job movers. 

To test this assumption, first I observe job movers within and across clusters. Movements are considered as long as for each gender workers are changing firms from the first period to the second period. I separate these movements into three categories. Upward movements represent when the worker moves from a lower cluster to an upper cluster. Downward is otherwise. Lateral movements are within cluster job changes.

Figure \ref{fig:exomobil} is constructed by first running a regression following Equation \ref{eq:new_akm}. Than I plot the difference in residuals for every transfer, gender, and data sample cell, discriminated by the movement type. Each dot represents a transfer cell observed in the labor market. The size of the dot indicates how common this particular transfer is.

The figure serves two purposes. First, it addresses obvious trends in job changes, which could indicate that unobservable factors not captured in my model are influencing mobility decisions. A lack of symmetry in the figure would suggest that certain labor market transitions are driven by such unobservables. Second, and equally important, I differentiate these movements by gender to identify any discrepancies that may be endogenous to my model but related to gender differences.

The symmetry plot provides robust evidence supporting the exogenous mobility assumption, which is fundamental to the proper identification of my model. It shows that every movement type possess examples of positive and negative difference in residuals, strongly strongly indicating that job transitions in the labor market are primarily governed by stable firm wage policies and worker characteristics, rather than by time-varying, unobserved factors correlated with wages. Last but not least, the analysis also reveals no apparent gender-specific patterns that would undermine the exogenous mobility assumption for male or female worker samples separately.

\begin{figure}[htb!]
    \centering
    \includegraphics[width=0.8\textwidth]{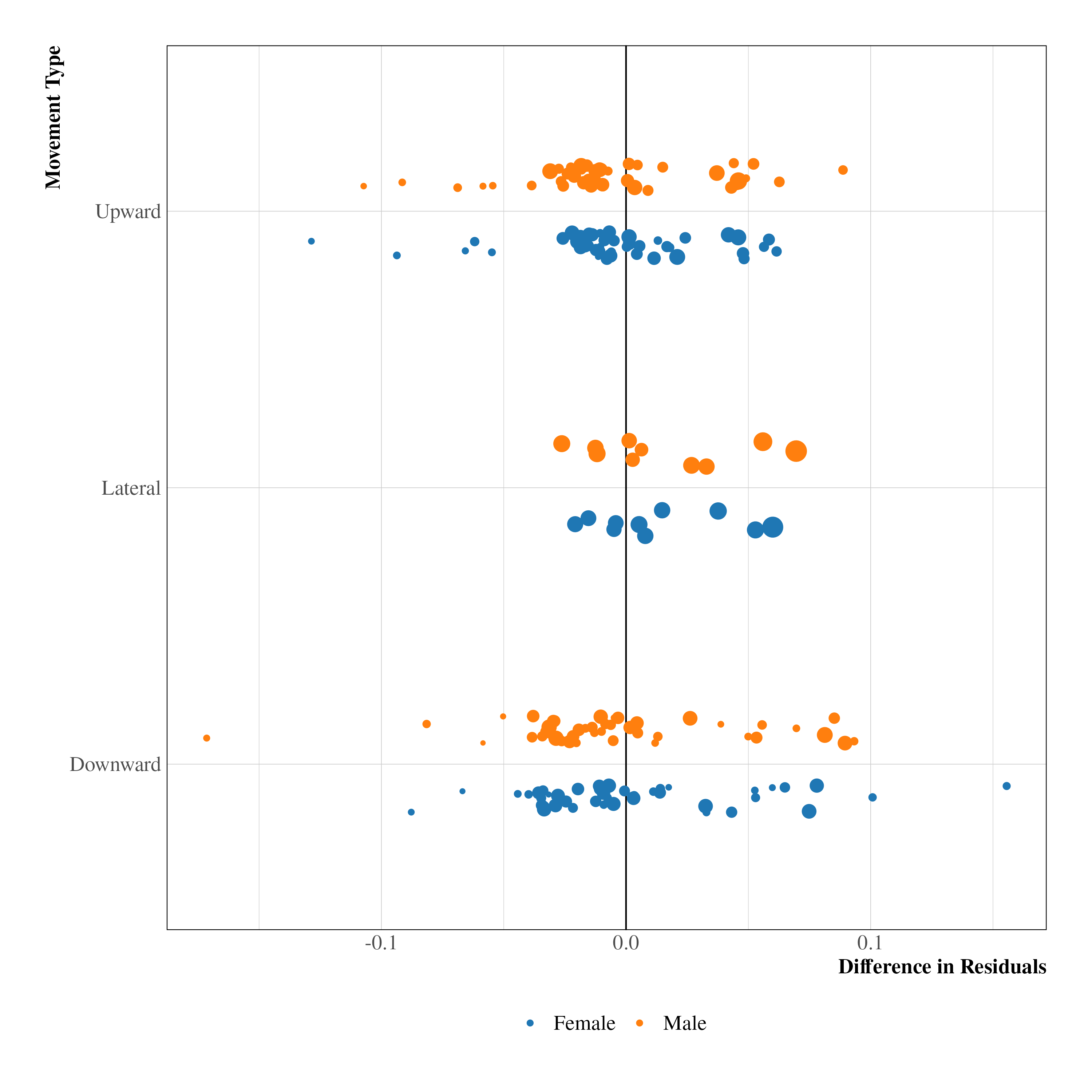}
    \caption{Symmetry plot of job movers' difference in residuals from first to second period.}
    \caption*{\small \textit{Note:} \textsuperscript{1}Dots represent across and within cluster-gender-sample cell movements. \textsuperscript{2}Dot size represents each cell's number of observed movements.}
    \label{fig:exomobil}
\end{figure}

\clearpage
\section{AKM and the Limited Mobility Bias} \label{sec:akm_discussion}

Here I briefly explain the presence of bias in the AKM estimator of \textcite{abowd_high_1999}.

\subsection{The AKM Model}

The AKM is formally written as:

\begin{equation} \label{eq:akm}
    w_{it} = X'_{it}\beta + \alpha_i + \phi_{J(i,t)} + \varepsilon_{it}
\end{equation}
where $w_{it}$ are the log earnings of worker $i$ in time $t$, $X'_{it}\beta$ are exogenous covariates such as age or time period, $\alpha_i$ is the unobserved worker heterogeneity, $J(i,t)$ is an assignment function representing the firm where $i$ works at $t$, meaning $\phi_{J(i,t)}$ represents the unobserved firm heterogeneity, and $\varepsilon_{it}$ is the error term.

Following \textcite{bonhomme_how_2023}, assume $N$ is the number of workers, $J$ the number of firms. For convenience, assume $T = 2$ is the number of time periods. The following assumption must hold:

\begin{equation}
    \mathbb{E}[\varepsilon_{it} | X_{11}, \ldots, X_{NT}, j(1,1), \ldots, j(N, T), \alpha_1, \ldots, \alpha_N, \phi_1, \ldots, \phi_J] = 0
\end{equation}

It is possible, without loss of generality, to rewrite Equation \ref{eq:akm} partialing out $X\beta$ and in vector form. Still following \textcite{bonhomme_how_2023}, I have:

\begin{equation}
    W = A\gamma + \varepsilon
\end{equation}
without loss of generality, assume $W$ is subtracted from $X\beta$ and $A$ represents the column-space of  worker and firm identifiers.

\subsection{Connected Set} \label{sec:dualset} 

In matched employer-employee data, the matrix $AA'$ is typically singular, necessitating an additional data cleaning step to ensure a sample of workers and firms that renders $AA'$ non-singular. This step is crucial for the identification of firm and worker effects in additive wage models. For instance, \textcite{card_bargaining_2016} estimates gender-specific firm wage effects by isolating the largest dual connected set from the main sample.

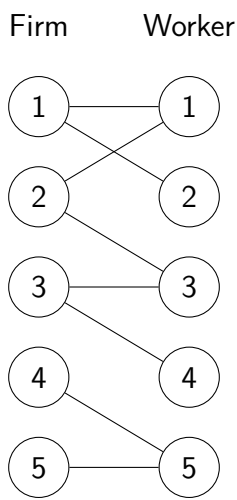
\begin{figure}[htb!]
    \centering
    
    \begin{tikzpicture}
    \begin{scope}[xshift=6cm]
        \foreach \y [count=\i] in {1,...,5}
        {
            \node[draw, circle, minimum size=0.8cm, font=\sffamily] (p\y) at (2, -\i*1.2) {\y};
            \node[draw, circle, minimum size=0.8cm, font=\sffamily] (f\y) at (0, -\i*1.2) {\y};
        }
        
        \draw (f1) -- (p1);
        \draw (f1) -- (p2);
        \draw (f2) -- (p1);
        \draw (f2) -- (p3);
        \draw (f3) -- (p3);
        \draw (f3) -- (p4);
        \draw (f4) -- (p5);
        \draw (f5) -- (p5);
    
        \node[above, font=\sffamily] at (0, -0.4) {Firm};
        \node[above, font=\sffamily] at (2, -0.4) {Worker};
    \end{scope}
    \end{tikzpicture}
    
    \caption{Firm-worker pairs. Firms 1, 2, and 3 are in the largest connected set through workers 1, 2, 3, and 4. Firms 4 and 5 are connected through worker 5 but disjoint from the rest.}
    \label{fig:connected_set_akm}
\end{figure}

This concept of connectivity in the labor market is illustrated in Figure \ref{fig:connected_set_akm}, which provides a simplified representation of worker movements across firms\footnote{For a comprehensive explanation of connected sets and their extraction from data, see \textcite{abowd_computing_2002}.}. It depicts a labor market with five firms and five workers over two time periods. Worker 1 moves from Firm 1 to Firm 2, thereby connecting these two firms. Firm 2 is further connected to Firm 3 through the movement of Worker 3. Firms 4 and 5, while isolated from the first three firms, are connected to each other through Worker 5.

In studies employing matched employer-employee data under additive separability models relying on firm and worker identifiers, researchers typically sample the largest connected set of firms. However, when investigating worker heterogeneity between genders, it is necessary to use the dual connected set, defined as the intersection of the largest connected sets for male and female samples. This approach ensures that firm effects are identified and comparable between both gender groups.

\subsection{Limited Mobility Bias}

The limited mobility bias is a significant concern in the estimation of firm effects, arising from the relative scarcity of job movers in the labor market \parencite{andrews_high_2008, bonhomme_how_2023}. While this bias does not directly appear in the firm effects estimates from Equation \ref{eq:akm}, it manifests in the variance analyses that are commonly employed in the literature to decompose wage inequality.

The sample variances or covariances of interest can be expressed in matrix notation as:
\begin{equation}
\sigma^2 = \gamma' Q \gamma
\end{equation}
where $Q$ is a matrix that depends on the design matrix $A$.

\textcite{andrews_high_2008} demonstrated the existence of this bias by decomposing the estimator $\hat{\sigma}^2$:
\begin{equation}
\mathbb{E}[\hat{\sigma}^2 | A] = \gamma' Q \gamma + \text{trace}(A(A'A)^{-1}Q(A'A)^{-1}A'\mathbb{V}[\varepsilon|A]) = \sigma^2 + \xi
\end{equation}
where $\xi$ represents the bias term.

Directly correcting for this bias is computationally challenging, as it requires inverting a large matrix, often of dimensions in the hundreds of thousands for firms and millions for workers in typical matched employer-employee datasets. \textcite{bonhomme_how_2023} have shown that common approximations used in the literature may be insufficient, particularly when relying on fixed effects derived from identifiers. This insufficiency stems from the fact that these approximations often fail to fully account for the complex network structure of worker mobility across firms.

To address these challenges, \textcite{bonhomme_distributional_2019} proposed a dimension reduction framework. This approach groups firms and workers into a smaller number of classes, thereby increasing the relative probability of observed job changes between groups. While this method effectively mitigates the limited mobility bias, it comes at the cost of imposing additional structure on the estimation model.

\subsection{Firm Size and Connectivity}

Not only does the largest connected set requirement impose a bias due to the rarity of mobility, but it also alters the overall wage distribution of the data. This alteration stems from the fact that the largest connected set tends to include larger firms more frequently than smaller ones. If larger firms differ significantly in their payment schedules and behavior compared to their smaller counterparts, the results derived from such analyses may have limited external validity.

In this section, I provide a formal proof that larger firms are more likely to be included in a connected set of a matched employer-employee dataset. I begin by defining the probability of worker mobility between firms and then demonstrate how this probability scales with firm size.

Without loss of generality, assume $T = 2$. Let $\mathcal{J} = \{1, \ldots, J\}$ be the set of all firms in the economy, and let $N_j$ denote the number of workers in firm $j$. Define $p_{jj'}$ as the probability that a given worker moves from firm $j$ to another firm $j'$. For simplicity, assume that this probability is the same across all workers in the labor market.

\begin{definition}[Connected Set]
A connected set $\mathcal{C} \subseteq \mathcal{F}$ is a subset of firms such that for any two firms $j, j' \in \mathcal{C}$, there exists a sequence of firms $j_1, \ldots, j_C \in \mathcal{C}$ with $j_1 = j$, $j_C = j'$ and for each connection $c \in \{1, \ldots, C-1\}$, there is at least one worker who has been employed in both $i_c$ and $i_{c+1}$.
\end{definition}

\begin{lemma}[Probability of Observed Mobility] \label{lemma1}
The probability of observing at least one worker moving from firm $j$ to firm $j'$ is:
\begin{equation}
P(j \rightarrow j') = 1 - (1 - p_{jj'})^{N_j}
\end{equation}    
\end{lemma}

\begin{proof}
The probability of a single worker not moving from $j$ to $j'$ is $(1 - p_{jj'})$. Assume that for all $N_j$ workers not to move, this must occur independently for each worker. Thus, the probability that no workers are moving is $(1 - p_{jj'})^{N_j}$, and the probability that at least one worker is moving is the complement of this event.
\end{proof}

\begin{theorem}
The probability of a firm being part of the connected set is increasing in firm size.
\end{theorem}

\begin{proof}
For firm $j$ to be part of the connected set, it must have at least one worker moving to or from another firm in the set. The probability of firm $j$ being connected is:

\begin{equation}
P(j \in \mathcal{C}) = 1 - \prod_{j' \neq j}(1 - P(j \rightarrow j')) \cdot \prod_{j'' \neq j} (1 - P(j'' \rightarrow j))
\end{equation}

Substituting the result from Lemma \ref{lemma1}:

\begin{equation}
P(j \in \mathcal{C}) = 1 - \prod_{j' \neq j}(1 - p_{jj'})^{N_j} \cdot \prod_{j'' \neq j}(1 - p_{j''j})^{N_{j''}}
\end{equation}
To show that this probability increases with firm size, take the derivative with respect to $N_j$:

Taking the derivative with respect to $N_j$:
\begin{align}
\frac{\partial P(j \in \mathcal{C})}{\partial N_j} &= -\left(\prod_{j' \neq j}(1 - p_{jj'})^{N_j} \cdot \prod_{j'' \neq j}(1 - p_{j''j})^{N_{j''}}\right) \
&\quad \cdot \sum_{j' \neq j} \log(1 - p_{jj'})
\end{align}

Since $0 < p_{jj'} < 1$, we have $\log(1 - p_{jj'}) < 0$, and thus $\frac{\partial P(j \in \mathcal{C})}{\partial N_j} > 0$.
\end{proof}

\begin{corollary}
As firm size approaches infinity, the probability of being in the connected set approaches 1:

\begin{equation}
    \lim_{N_j \to \infty} P(j \in \mathcal{C}) = 1  
\end{equation}
\end{corollary}

\begin{proof}
As $N_j \to \infty$, $(1 - p_{jj'})^{N_j} \to 0$ since $0 < p_{jj'} < 1$. Therefore, the product $\prod_{j \neq j'} (1 - p_{jj'})^{N_j} \to 0$, and consequently, $P(j \in \mathcal{C}) \to 1$.
\end{proof}

\paragraph{Other firms approaching infinity}
Another consequence of this proof is when the size of any other firm $N_{j_0}$ approaches infinity while $N_j$ remains finite, $P(j \in \mathcal{C})$ also approaches 1, but at a slower rate. This is because:
\begin{equation}
\lim_{N_{j_0} \to \infty} P(j \in \mathcal{C}) = 1 - \prod_{j' \neq j}(1 - p_{jj'})^{N_j} \cdot 0 \cdot \prod_{j'' \neq j, j'' \neq j_0}(1 - p_{jj''})^{N_j''} = 1
\end{equation}

However, this convergence is slower than when $N_j \to \infty$ because only one term in the product approaches zero, rather than all terms involving $N_j$.

\subsubsection{Empirical Evidence}

The theoretical framework is substantiated by empirical evidence presented in Figures \ref{fig:dcs_number_workers} and \ref{fig:dcs_log_wage}, which illustrate the differences between the full sample and the largest dual connected set (LDCS) in terms of firm size and wage distributions.

Figure \ref{fig:dcs_number_workers} reveals a stark contrast in the distribution of workforce size between the full sample and the LDCS. The LDCS exhibits a symmetrical distribution shifted significantly to the right, with a mean firm size of approximately 194 workers, compared to the full sample's mean of 27 workers. This rightward shift is accompanied by increased variability, with the standard deviation in the LDCS being almost four times higher than in the original sample. The median firm size in the LDCS is also notably higher, underscoring the overrepresentation of larger firms in the connected set.

The wage distribution, as depicted in Figure \ref{fig:dcs_log_wage}, further emphasizes the discrepancies between the full sample and the LDCS. The mean log wage in the full sample is 1.87, with a standard deviation of 0.446, while the LDCS shows a substantially higher mean log wage of 2.32 and a larger standard deviation of 0.656. This upward shift in both moments indicates a clear upward bias in wage levels within the connected set. The increased standard deviation in the LDCS also points to greater wage dispersion among the firms included in this subset.

These findings have some implications for the estimation and interpretation of gender wage gaps.  LDCS exhibits a larger gender wage gap ($\delta$ = -0.317) compared to the entire sample ($\delta$ = -0.237), suggesting that studies using the connected set may overestimate the overall wage disparity. This overestimation likely stems from the LDCS that captures wage dynamics primarily in larger, more established firms where gender wage differences might be more pronounced. The exclusion of smaller, potentially lower-paying firms that do not meet the connectivity requirements for AKM-style fixed effects estimation contributes to this bias.

Researchers should exercise caution when generalizing results from the connected set to the broader labor market. The LDCS, while providing the necessary conditions for certain econometric techniques, may not fully represent the wage structures and gender dynamics present in smaller or less connected firms. This limitation is particularly important when studying labor markets with a significant proportion of small enterprises or sectors with limited inter-firm mobility.

\begin{figure}
    \centering
    \includegraphics[width = \textwidth]{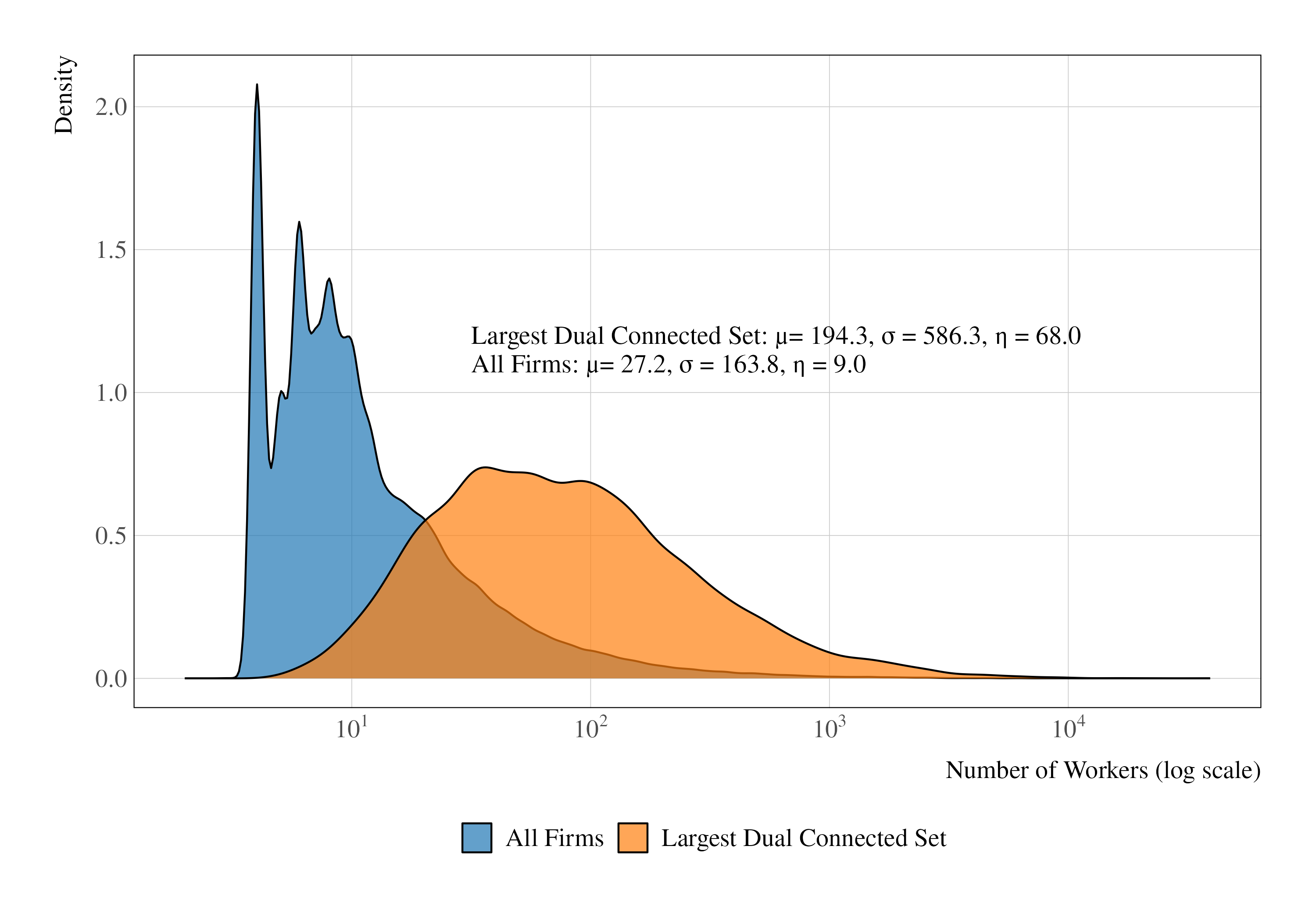}
    \caption{All Firms and The LDCS Number of Workers Distributions}
    \caption*{\footnotesize \textit{Note:} \textsuperscript{1}Distributions generated from the six biennial samples, using the full firm set, and the largest dual connected set of firms. \textsuperscript{2}$\mu$ is the mean of the distribution, $\sigma$ represents the standard deviation, $\eta$ represents the median number of workers per firm.}
    \label{fig:dcs_number_workers}
\end{figure}

\begin{figure}
    \centering
    \includegraphics[width = \textwidth]{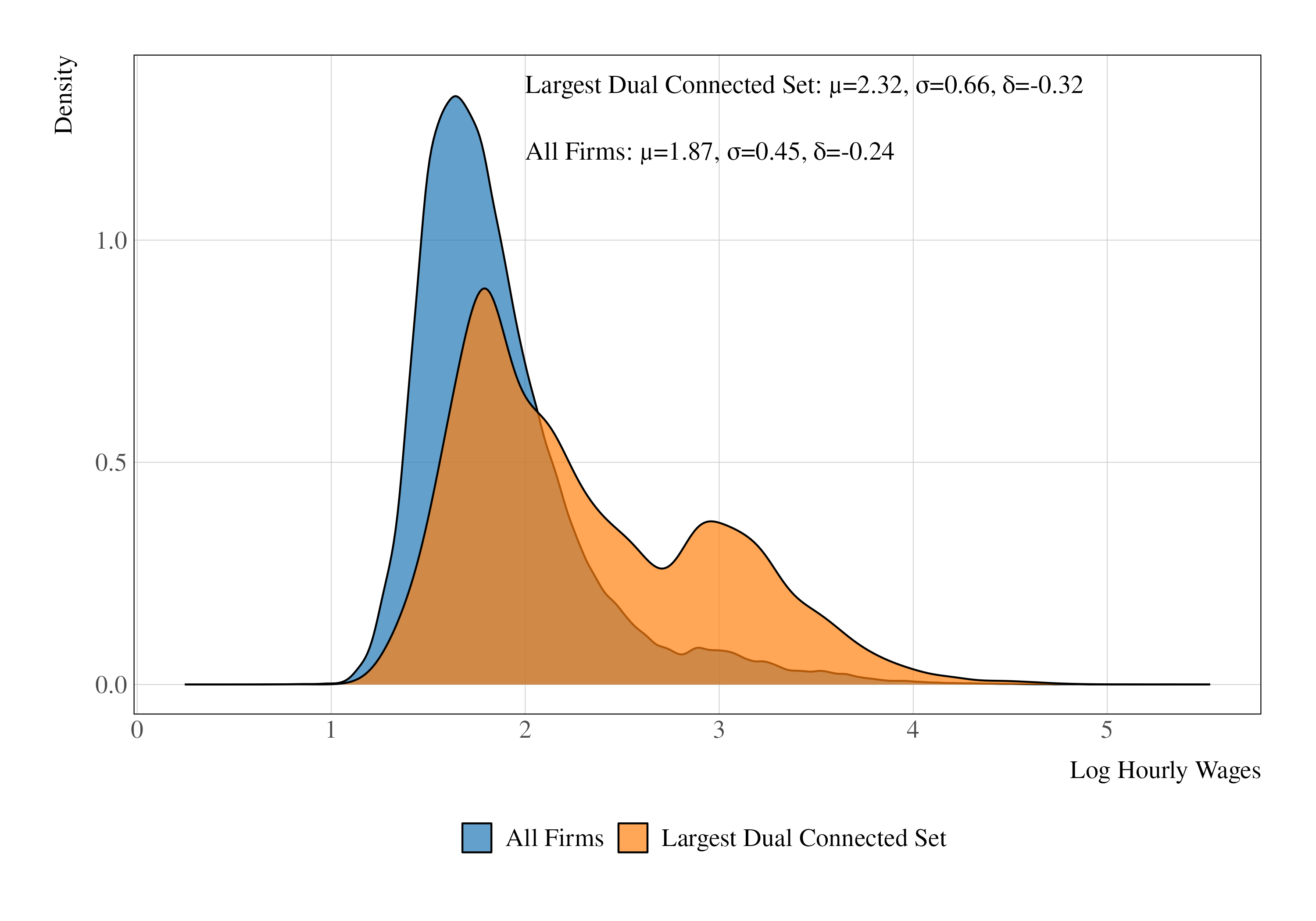}
    \caption{All Firms and The LDCS Log-Weekly Wage Distributions}
    \caption*{\footnotesize \textit{Note:} \textsuperscript{1}Distributions generated from the six biennial samples, using the full firm set, and the largest dual connected set of firms. \textsuperscript{2}$\mu$ is the mean of the distribution, $\sigma$ represents the standard deviation, $\delta$ represents the average female-male wage gap.}
    \label{fig:dcs_log_wage}
\end{figure}

\clearpage
\section{Applying Estimated Clusters in a Linear Framework} \label{sec:CAKM}

Here I provide two exercises to show the mixture model can be used as a plugin estimator of unobserved heterogeneity in a linear regression setting.

First, I employ the estimated firm clusters to estimate firm heterogeneity contribution to the gender wage gap in a classical AKM-KOB from \textcite{card_bargaining_2016}. The novelty is that I keep social identifiers as the worker plugin estimators, however, I leverage the k-means clusters of firms to avoid data trimming.

Second, I provide a variance decomposition analysis, comparing classical AKM decomposition with the clustered AKM provided previously and the BLM decomposition, where I replace worker social identifiers with their respective mixture membership.

\subsection{Estimating Firm Contribution to the Gender Wage Gap}

Here, I employ the estimated firm clusters in an AKM-KOB classical linear framework to examine their contribution to the gender wage gap. The advantage of this approach is that by grouping firms with similar wage structures, I am not required to filter the dual connected set from the data.

The model I employ has strong similarities to the typical AKM framework. A key distinction lies in the treatment of firm heterogeneity. Instead of utilizing the standard firm assignment function $J(i,t)$, which maps each worker-year observation to a specific firm, I introduce cluster assignment function $K(i,t)$. This function maps worker-year observations to firm clusters, thereby reducing the dimensionality of the firm entities.

Formally, the wage equation can be expressed as:
\begin{equation} \label{eq:new_akm}
{w}_{it} = \alpha_i + \psi^g_{K(i,t)} + {X^g}_{it}'\beta + \varepsilon_{it}
\end{equation}
where $w_{it}$ is the log wage of worker $i$ in year $t$, $\alpha_i$ is the worker fixed effect, $\psi_{K(i,t)}$ is the effect of the firm cluster to which worker $i$'s employer belongs in year $t$, $X_{it}$ is a vector of time-varying covariates, and $\varepsilon_{it}$ is the error term. Superscript $g$ indicates that I apply the regression to both the male and female samples.

In particular, the key parameter of interest in this analysis is the difference in firm endowment between male and female workers, derived from the KOB decomposition. This parameter represents the differential distribution of workers across firm clusters by gender, serving as a preliminary measure of gender-specific assortative matching patterns in the labor market, although focusing solely on firm heterogeneity.

Specifically, I adapt the KOB decomposition to the context of firm cluster effects, formalizing it as:

\begin{align}
&\underbrace{E[\psi^f_{K(i,t)} | female] - E[\psi^m_{K(i,t)} | male]}_{\substack{\text{Firm Cluster} \\ \text{Contribution to the} \\ \text{Gender Wage Gap}}} = \\ & \text{Bargaining: }\underbrace{\frac{1}{2}\sum_{x\in{F,M}}(E[\psi^F_{K(i,t)} - \psi^M_{K(i,t)} | g = x])}_{\substack{\text{Unexplained Portion} \\ \text{(Difference in Returns)}}} \\ 
&\text{Sorting: }+ \underbrace{\frac{1}{2}\sum_{x\in{F,M}}(E[\psi^x_{K(i,t)} | g = F] - E[\psi^x_{K(i,t)} | g = M])}_{\substack{\text{Explained Portion} \\ \text{(Difference in Distributions)}}}
\end{align}
where the left-hand side of the equation represents the total contribution of firm cluster effects to the gender wage gap. This contribution is decomposed into two components: the bargaining effect and the sorting effect.

The bargaining effect, following the terminology of \textcite{card_bargaining_2016}, captures the portion of the gap attributable to differences in the estimated firm premia between men and women, holding the distribution of firms constant. This effect is computed as the average of two counterfactuals: one using the observed female distribution of firms and another using the observed male distribution.

The sorting effect, conversely, measures the portion of the gap that arises from differences in the distribution of men and women across firm clusters, assuming gender-neutral firm effects. This is the difference in endowments of the Oaxaca decomposition. This effect is also computed as the average of two counterfactuals: one using the estimated male returns to firm clusters and another using the estimated female returns.

By averaging these counterfactuals for each component, as suggested by \textcite{casarico_what_2024}, I obtain robust estimates that account for potential sensitivity to the choice of reference group. Unless otherwise specified, the reported bargaining and sorting effects refer to these averaged estimates.

\subsubsection{Normalizing Firm Effects}

There is established practice in the AKM literature on gender wage gaps in using the hotel and restaurant industry as a reference\footnote{Examples are \textcite{cruz_effects_2022, casarico_what_2024}}. This sector is often chosen due to its typical low wage premia and high turnover rates, suggesting minimal rents \citep{card_bargaining_2016, coudin_family_2018}.

Although class 2 firms have a slightly higher proportion of hotels and restaurants (1 percent against 0.6 percent in class 1)\footnote{See Figure \ref{fig:hotel_fe} for the estimated premia and the proportion of hotels and restaurants per firm class}, I argue that class 1 is the most appropriate reference for several reasons. First, conditional average wages on firm fluster are fairly linear, with class 1 exhibiting the lowest average wage premium in my pooled regression, aligning with the theoretical expectation that the reference group should represent firms offering minimal rents. When employing the AKM regression with class 1 as the reference, resulting fixed effects preserve the linear behavior, with no class exhibiting negative estimates.

\subsubsection{Oaxaca Decomposition of Firm Cluster Effects} \label{sec:results}

The main findings derived from the estimation of Equation \ref{eq:new_akm} are summarized in Table \ref{tab:main_results}. These results represent weighted average firm cluster effects estimates obtained from the six separate biennial samples.

\begin{table}[htp!]
\centering
\begin{threeparttable}
\caption{Firm Decomposition of the Gender Wage Gap: Overall and by Subgroups}
\label{tab:main_results}
\begin{tabular}{lcccc}
\toprule
& & \multicolumn{3}{c}{Contribution to Gender Wage Gap} \\
\cmidrule(lr){3-5}
Group & Total Gap & Firm       & Sorting            & Bargaining  \\
      &           & Components & Components         & Components \\ 
& (1) & (2) & (3) & (4) \\
\midrule
All & \num{-0.237} & \num{-0.033} & \num{-0.022} & \num{-0.01}\\
 &                 & (\num{0.14}) & (\num{0.09}) & (\num{0.04})\\
\addlinespace
\multicolumn{5}{l}{\textit{By age group:}} \\
Up to age 30 & \num{-0.092} & \num{-0.014} & \num{-0.009} & \num{-0.005}\\
 &  & (\num{0.16}) & (\num{0.1}) & (\num{0.06})\\
Ages 31-50 & \num{-0.305} & \num{-0.041} & \num{-0.029} & \num{-0.012}\\
 &  & (\num{0.14}) & (\num{0.09}) & (\num{0.04})\\
Over age 50 & \num{-0.326} & \num{-0.076} & \num{-0.028} & \num{-0.048}\\
 &  & (\num{0.23}) & (\num{0.09}) & (\num{0.15})\\
\addlinespace
\multicolumn{5}{l}{\textit{By education group:}} \\
No High school & \num{-0.297} & \num{-0.046} & \num{-0.035} & \num{-0.011}\\
 &  & (\num{0.15}) & (\num{0.12}) & (\num{0.04})\\
High school & \num{-0.233} & \num{-0.021} & \num{-0.025} & \num{0.004}\\
 &  & (\num{0.09}) & (\num{0.11}) & (\num{-0.02})\\
College & \num{-0.348} & \num{-0.075} & \num{-0.03} & \num{-0.046}\\
 &  & (\num{0.22}) & (\num{0.08}) & (\num{0.13})\\
\bottomrule
\end{tabular}
\begin{tablenotes}
\small
\item \textit{Notes:} \textsuperscript{1}This table presents the decomposition of the gender wage gap into components attributable to clustered firm-specific factors, using Equation \ref{eq:new_akm}. Column (1) shows the total female-male wage gap in means. Column (2) presents the total contribution of firm-specific factors. Columns (3) and (4) further decompose the firm premium contribution into a sorting (explained) and a bargaining (unexplained) components, respectively. \textsuperscript{2}Numbers in parenthesis represent the fraction of the overall gender wage gap that is attributed to the source described in column heading. \textsuperscript{3}Results are an weighted average of the six biennial samples.
\end{tablenotes}
\end{threeparttable}
\end{table}

First, considering the overall sample, there is a substantial gender wage gap of 23.7 log points. Firm-specific factors, in Column (2), are estimated around 3.1 log points. As explained by \textcite{card_bargaining_2016}, this component can be interpreted as the difference in rent payment relative to firm class 1. The component accounts for 13 percent of this total gap.

This contribution can be further decomposed into the sorting component, the difference in male and female worker distributions considering a gender-neutral relative rent. In this case, the total difference is evaluated at 2.1 log points, corresponding to approximately 9 percent of the total gender wage gap. Likewise, the bargaining channel is the average difference in the estimated premium, assuming both genders possess the same firm share. This channel is measured at 1 log point, about 4 percent of the overall gender wage gap.

Following the literature, the lower rows of Table \ref{tab:main_results} show that the gender wage gap increases dramatically with age, with firms playing a role in this increase, since for individuals older than 50, 23 percent of the 32.6 wage gap is due to estimated firm class effects. Notably, while the sorting component remains relatively stable across age groups (ranging from 0.8 to 2.7 log points), the Bargaining component increases substantially, from 0.7 log points for the youngest group to 4.7 log points for the oldest. This pattern suggests that as workers age, differences in how firms compensate men and women in similar positions become increasingly important in explaining the gender wage gap.

The analysis by education level reveals a pattern that aligns with findings from the U.S. labor markets, given the wage gap is more prominent among college-educated workers (34.8 log points), albeit the smallest among those with a high school education (23.3 log points). Workers without a high school diploma fall in between, with a gap of 29.7 log points. Interestingly, the contribution of firm-specific factors to the gap follows a similar pattern, being the highest for college-educated workers (5.6 log points) and the lowest for high school graduates (2.8 log points).

For workers without a high school education and those with a high school diploma, the sorting component dominates the bargaining component. This suggests that for these groups, the allocation of women across firms plays a more significant role in the gender wage gap than within-firm differences in compensation. However, the picture changes for college-educated workers. In this group, the bargaining component (3.0 log points) marginally exceeds the sorting component (2.6 log points), indicating that within-firm differences in compensation between men and women become more pronounced as individuals accumulate human capital. Individuals with higher levels of human capital tend to be allocated to more specialized occupations, often within larger firms. However, women may be concentrated at lower paying occupations compared to male counterparts, resulting in workers with greater human capital accumulation securing positions that ultimately translates into heterogeneous bargaining effects within large firms.

\subsubsection{Other Cluster Choices and Classical AKM}

Table \ref{tab:cluster_robustness} presents the robustness analysis of the model, allowing different cluster choices in the clustered AKM (C-AKM), which is the empirical specification of my study, and, for comparison, the classical AKM model.  The results span different levels of firm clustering (K = 4, 6, 8, and 10) in my grouped fixed effects approach, as well as the traditional AKM and the baseline clustering under the largest dual connected set.

Under the clustered AKM approach, trimming the data is not required. However, for traditional AKM settings, it is required to extract the largest dual connected set for correct identification. This is reflected by the larger total gap at 31.7 log points, contrasted with the full data 23.7 log points\footnote{I reserve the appendix for a comprehensive analysis of the AKM model and the dual connected set requirement. See Section \ref{sec:akm_discussion}.}.

There is a modest but consistent increase in firm components when moving from K = 4 to K = 6, increasing from 2.4 to 3.1 log points. This trend suggests that finer firm classifications capture additional nuances in firm-specific contributions to the gender wage gap. 

\subsubsection{Bargaining Sensitivity to Cluster Normalization}
The bargaining channel, however, is the most sensitive to variations in cluster choice, becoming the almost sole driving force of my robustness analysis, with the sorting component practically stable across results. However, it seems that bargaining estimates stabilize between K = 8 and K = 10, indicating that beyond a certain point, further granularity in firm classification yields diminishing returns in terms of explanatory power. Interestingly, $K = 10$ represents the optimal cluster choice in the gap statistics evaluation\footnote{See Section \ref{sec:clusterchoiceanalysis} for a more rigorous discussion on cluster choice.} \parencite{tibshirani_estimating_2001}. 

The bargaining channel sensitivity could be attributed to the normalization procedure. Slicing firms in the data based on wage distribution similarities may keep the sorting of male and female workers in the labor market almost intact, given that women are more concentrated in low-paying firms. However, it may potentially underestimate rents coming from firms grouped at the lowest class that are, in fact, different enough to be categorized separately in more granulated settings. This is confirmed by the striking different estimate coming from running a C-AKM on the LDCS sample. The largest dual connect set is overrepresented by larger firms, therefore, it is possible that only larger firms with strong positive rent sharing for men were kept in the sample, severely underestimating the male worker's firm component returns.

Therefore, if the researcher desires to commit to the linearity assumption of AKM leveraging from the benefits of clustering firms, the best practice is to employ several number of cluster choices to determine the most appropriate configuration. Ideally, it should minimize within-cluster variance, ensuring that firms within each group are sufficiently homogeneous in their payment behavior, minimizing the normalization cost. Moreover, it should maintain enough heterogeneity between clusters to capture meaningful differences between clusters and provide better economic intuition.

\begin{table}[htpb!]
\centering
\begin{threeparttable}
\caption{Firm Decomposition: Different Model Specifications}
\label{tab:cluster_robustness}
\begin{tabular}{lcccc}
\toprule
& & \multicolumn{3}{c}{Contribution to Gender Wage Gap} \\
\cmidrule(lr){3-5}
Group & Total Gap & Firm       & Sorting            & Bargaining  \\
      &           & Components & Components         & Components \\ 
& (1) & (2) & (3) & (4) \\
\midrule
K = 4 & -0.237 & -0.024 & -0.019 & -0.005 \\
& & (0.10) & (0.08) & (0.02) \\
\addlinespace
K = 6 & -0.237 & -0.031 & -0.021 & -0.010 \\
& & (0.13) & (0.09) & (0.04) \\
\addlinespace
K = 8 & -0.237 & -0.034 & -0.022 & -0.013 \\
& & (0.15) & (0.09) & (0.06) \\
\addlinespace
K = 10 (Baseline) & -0.237 & -0.033 & -0.022 & -0.010 \\
& & (0.14) & (0.09) & (0.04) \\
\midrule
AKM & -0.317 & -0.046 & -0.026 & -0.020 \\
& & (0.14) & (0.08) & (0.06) \\
K = 10 on LDCS & -0.317 & -0.025 & -0.027 & 0.002\\
 &  & (0.07) & (0.08) & (-0.01)\\
\bottomrule
\end{tabular}
\begin{tablenotes}
\small
\item \textit{Notes:} \textsuperscript{1}K represents the number of firm clusters in the grouped fixed effects models. \textsuperscript{2}AKM refers to the traditional \textcite{abowd_high_1999} specification under the largest dual connected set sample. \textsuperscript{3}Numbers in parenthesis represent the fraction of the overall gender wage gap that is attributed to the source described in column heading. \textsuperscript{4}Column (1) shows the total female-male wage gap in means. Column (2) presents the total contribution of firm-specific factors. Columns (3) and (4) further decompose the firm premium contribution into a sorting (explained) and a bargaining (unexplained) components, respectively. 
\end{tablenotes}
\end{threeparttable}
\end{table}

\clearpage
\section{Additional Tables and Figures}

\begin{figure}[htbp!]
    \centering
    \includegraphics[width=1\linewidth]{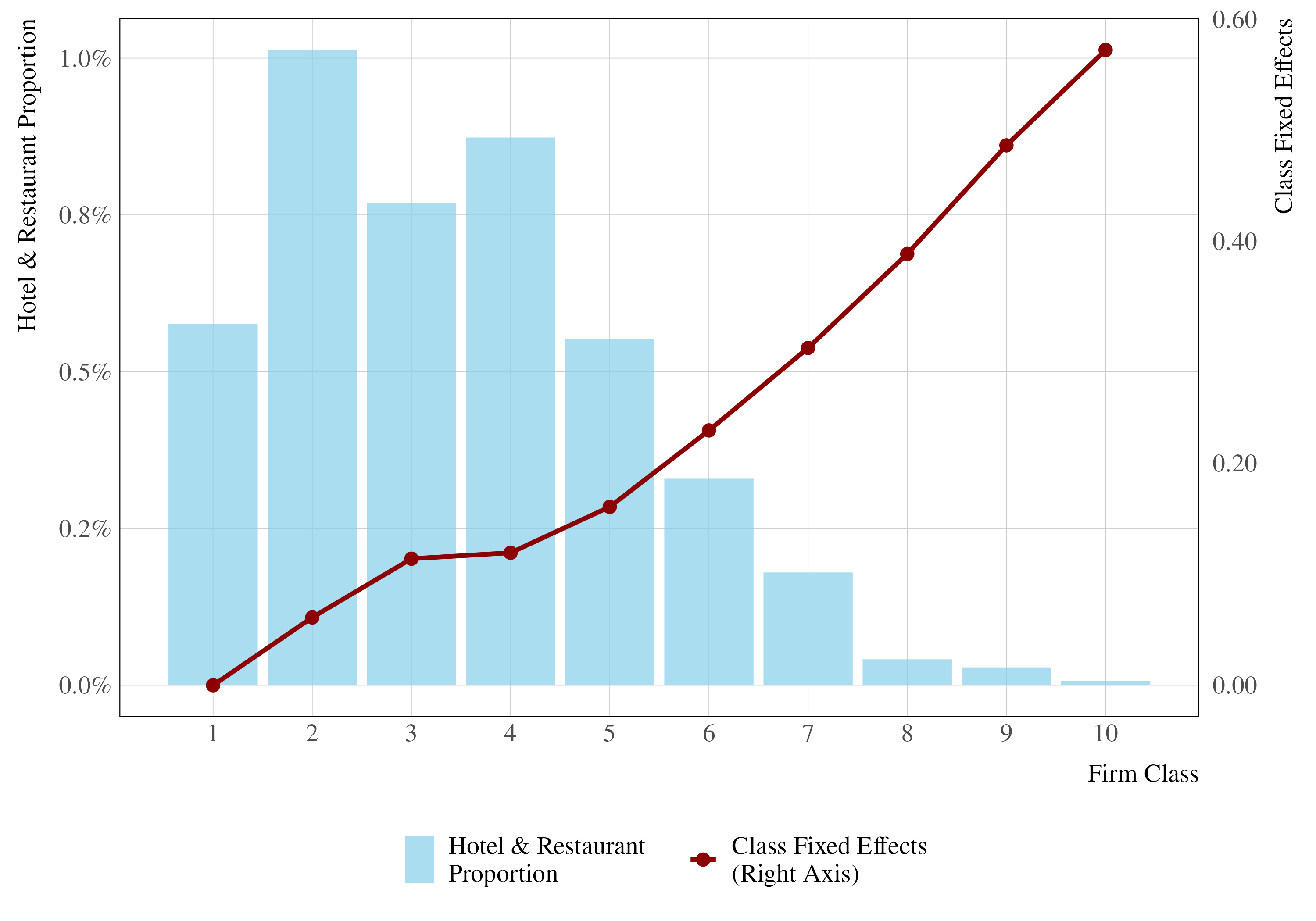}
    \caption{Estimated Effects and Hospitality Industry Proportions Per Firm Class}
    \caption*{\small\textit{Note:} \textsuperscript{1}Firm classes estimated through a kmeans clustering algorithm based on the distribution of logarithmic wages. \textsuperscript{2}Hotels and restaurants extracted from code 55 and 56 of the Brazilian \textit{CNAE} code of economic activities.}
    \label{fig:hotel_fe}
\end{figure}

\begin{table}[htbp!] \label{tab:sumstatldcs}
\centering
\caption{Descriptive Statistics by Gender: Largest Dual Connected Set}
\label{tab:desc_stats_ldcs}
\small
\begin{tabular}{l rr}
\toprule
\textbf{Features} & {\textbf{Female Workers}} & {\textbf{Male Workers}} \\
\midrule
\multicolumn{3}{l}{\textit{Firm Characteristics}} \\
\quad Number of Firms & {\num{24500}} & {\num{24500}} \\
\quad Firms with $\geq$ 10 Workers & {\num{21873}} & {\num{21873}} \\
\quad Firms with $\geq$ 50 Workers & {\num{12835}} & {\num{12835}} \\
\quad Mean Firm Size & 365 & 365 \\
\quad Median Firm Size & 55 & 55 \\
\addlinespace[0.5em]
\multicolumn{3}{l}{\textit{Worker Characteristics}} \\
\quad Education (\%) \\
\qquad Dropout & 19 & 22 \\
\qquad High School Graduates & 46 & 41 \\
\qquad Some College & 35 & 37 \\
\quad Age (\%) \\
\qquad $<$ 30 & 41 & 38 \\
\qquad 31--50 & 51 & 50 \\
\qquad $\geq$ 51 & 7 & 10 \\
\addlinespace[0.5em]
\multicolumn{3}{l}{\textit{Sector of Employment (\%)}} \\
\quad Primary & 1 & 2 \\
\quad Manufacturing & 19 & 28 \\
\quad Construction & 1 & 2 \\
\quad Trade & 15 & 15 \\
\quad Services & 65 & 54 \\
\addlinespace[0.5em]
\multicolumn{3}{l}{\textit{Occupation (\%)}} \\
\quad Scientific and Liberal Arts & 15 & 16 \\
\quad Technicians & 14 & 14 \\
\quad Administrative & 32 & 19 \\
\quad Managers & 5 & 8 \\
\quad Traders & 22 & 19 \\
\quad Rural & 1 & 2 \\
\quad Factory & 18 & 23 \\
\addlinespace[0.5em]
\multicolumn{3}{l}{\textit{Labor Market Outcomes}} \\
\quad Mean Experience (years) & 4.56 & 5.24 \\
\quad Mean log hourly Wage & 2.199 & 2.516 \\
\quad Variance of log hourly Wage & 0.639 & 0.802 \\
\addlinespace[0.5em]
\quad Worker-Year Observations & {\num{4464653}} & {\num{4469690}} \\
\quad Number of Workers & {\num{1831797}} & {\num{1812494}} \\
\quad Gender Fraction (\%) & 50 & 50 \\
\bottomrule
\end{tabular}
\caption*{\small\textit{Note:} \textsuperscript{1} Descriptive statistics calculated from the first year of each biennial sample's largest dual connected set (2010-2015). \textsuperscript{2} Percentages may not sum to 100\% due to rounding. \textsuperscript{3} The number of firms is the same for each gender since every firm in the cleaned sample employs both male and female workers.}
\end{table}

\begin{table}[htbp!]
\centering
\caption{Descriptive Statistics of Lower Firm Classes}
\label{tab:lowerclasses}
\resizebox{\linewidth}{!}{
\begin{tabular}{lrrrrrrrrrr}
\toprule
\multicolumn{1}{c}{ } & \multicolumn{2}{c}{class 1} & \multicolumn{2}{c}{class 2} & \multicolumn{2}{c}{class 3} & \multicolumn{2}{c}{class 4} & \multicolumn{2}{c}{class 5} \\
\cmidrule(l{3pt}r{3pt}){2-3} \cmidrule(l{3pt}r{3pt}){4-5} \cmidrule(l{3pt}r{3pt}){6-7} \cmidrule(l{3pt}r{3pt}){8-9} \cmidrule(l{3pt}r{3pt}){10-11}
\multicolumn{1}{c}{ } & \multicolumn{1}{c}{Female} & \multicolumn{1}{c}{Male} & \multicolumn{1}{c}{Female} & \multicolumn{1}{c}{Male} & \multicolumn{1}{c}{Female} & \multicolumn{1}{c}{Male} & \multicolumn{1}{c}{Female} & \multicolumn{1}{c}{Male} & \multicolumn{1}{c}{Female} & \multicolumn{1}{c}{Male} \\
\cmidrule(l{3pt}r{3pt}){2-2} \cmidrule(l{3pt}r{3pt}){3-3} \cmidrule(l{3pt}r{3pt}){4-4} \cmidrule(l{3pt}r{3pt}){5-5} \cmidrule(l{3pt}r{3pt}){6-6} \cmidrule(l{3pt}r{3pt}){7-7} \cmidrule(l{3pt}r{3pt}){8-8} \cmidrule(l{3pt}r{3pt}){9-9} \cmidrule(l{3pt}r{3pt}){10-10} \cmidrule(l{3pt}r{3pt}){11-11}
& (1) & (2) & (3) & (4) & (5) & (6) & (7) & (8) & (9) & (10) \\
\midrule
Number of Firms & \num{74865} & \num{74865} & \num{101454} & \num{101454} & \num{83715} & \num{83715} & \num{79165} & \num{79165} & \num{80060} & \num{80060}\\
Firms with $\geq$10 Workers & \num{21976} & \num{21976} & \num{35851} & \num{35851} & \num{32126} & \num{32126} & \num{28479} & \num{28479} & \num{37413} & \num{37413}\\
Firms with $\geq$50 Workers & \num{2382} & \num{2382} & \num{5259} & \num{5259} & \num{6176} & \num{6176} & \num{5284} & \num{5284} & \num{8437} & \num{8437}\\
Mean Firm Size & \num{18} & \num{18} & \num{18} & \num{18} & \num{22} & \num{22} & \num{23} & \num{23} & \num{31} & \num{31}\\
Median Firm Size & \num{5} & \num{5} & \num{6} & \num{6} & \num{7} & \num{7} & \num{6} & \num{6} & \num{9} & \num{9}\\
\addlinespace
Dropout & \num{0.54} & \num{0.52} & \num{0.36} & \num{0.44} & \num{0.28} & \num{0.40} & \num{0.31} & \num{0.40} & \num{0.27} & \num{0.39}\\
High School Graduates & \num{0.41} & \num{0.43} & \num{0.57} & \num{0.51} & \num{0.64} & \num{0.54} & \num{0.55} & \num{0.49} & \num{0.60} & \num{0.52}\\
Some College & \num{0.05} & \num{0.05} & \num{0.07} & \num{0.05} & \num{0.08} & \num{0.06} & \num{0.14} & \num{0.11} & \num{0.13} & \num{0.09}\\
\addlinespace
Age ($<$30) & \num{0.33} & \num{0.40} & \num{0.43} & \num{0.42} & \num{0.45} & \num{0.44} & \num{0.41} & \num{0.39} & \num{0.43} & \num{0.40}\\
Age 31-50 & \num{0.52} & \num{0.42} & \num{0.46} & \num{0.42} & \num{0.47} & \num{0.42} & \num{0.48} & \num{0.44} & \num{0.48} & \num{0.46}\\
Age ($\geq$51) & \num{0.13} & \num{0.17} & \num{0.09} & \num{0.14} & \num{0.07} & \num{0.12} & \num{0.09} & \num{0.15} & \num{0.08} & \num{0.13}\\
\addlinespace
Primary Sector & \num{0.05} & \num{0.09} & \num{0.03} & \num{0.05} & \num{0.02} & \num{0.04} & \num{0.03} & \num{0.05} & \num{0.02} & \num{0.03}\\
Manufacturing & \num{0.11} & \num{0.11} & \num{0.16} & \num{0.15} & \num{0.17} & \num{0.18} & \num{0.17} & \num{0.20} & \num{0.27} & \num{0.30}\\
Construction & \num{0.01} & \num{0.01} & \num{0.00} & \num{0.01} & \num{0.01} & \num{0.01} & \num{0.01} & \num{0.01} & \num{0.01} & \num{0.01}\\
Trade & \num{0.15} & \num{0.22} & \num{0.32} & \num{0.35} & \num{0.51} & \num{0.50} & \num{0.25} & \num{0.29} & \num{0.36} & \num{0.36}\\
Services & \num{0.68} & \num{0.57} & \num{0.49} & \num{0.44} & \num{0.29} & \num{0.27} & \num{0.54} & \num{0.44} & \num{0.35} & \num{0.31}\\
\addlinespace
Scientific and Liberal Arts & \num{0.02} & \num{0.02} & \num{0.02} & \num{0.02} & \num{0.02} & \num{0.02} & \num{0.05} & \num{0.04} & \num{0.04} & \num{0.03}\\
Technicians & \num{0.03} & \num{0.05} & \num{0.04} & \num{0.05} & \num{0.04} & \num{0.05} & \num{0.06} & \num{0.08} & \num{0.08} & \num{0.07}\\
Administrative & \num{0.20} & \num{0.14} & \num{0.34} & \num{0.15} & \num{0.39} & \num{0.17} & \num{0.36} & \num{0.19} & \num{0.36} & \num{0.16}\\
Managers & \num{0.02} & \num{0.03} & \num{0.03} & \num{0.05} & \num{0.04} & \num{0.06} & \num{0.04} & \num{0.05} & \num{0.03} & \num{0.05}\\
Traders & \num{0.56} & \num{0.44} & \num{0.39} & \num{0.40} & \num{0.33} & \num{0.36} & \num{0.32} & \num{0.29} & \num{0.27} & \num{0.29}\\
Rural & \num{0.04} & \num{0.10} & \num{0.03} & \num{0.05} & \num{0.02} & \num{0.03} & \num{0.02} & \num{0.05} & \num{0.01} & \num{0.02}\\
Factory & \num{0.13} & \num{0.22} & \num{0.15} & \num{0.28} & \num{0.16} & \num{0.32} & \num{0.15} & \num{0.30} & \num{0.21} & \num{0.38}\\
\addlinespace
Mean experience (years) & \num{2.789} & \num{3.031} & \num{2.818} & \num{3.136} & \num{3.238} & \num{3.521} & \num{3.405} & \num{3.959} & \num{3.635} & \num{4.033}\\
Mean Log-Wage & \num{1.358} & \num{1.466} & \num{1.514} & \num{1.624} & \num{1.660} & \num{1.775} & \num{1.667} & \num{1.860} & \num{1.806} & \num{1.949}\\
Variance of Log-Wage & \num{0.061} & \num{0.101} & \num{0.058} & \num{0.087} & \num{0.074} & \num{0.117} & \num{0.156} & \num{0.206} & \num{0.113} & \num{0.155}\\
Worker-years observations & \num{775686} & \num{540010} & \num{973659} & \num{888531} & \num{953073} & \num{921694} & \num{942303} & \num{881295} & \num{1143801} & \num{1317477}\\
Number of Workers & \num{437700} & \num{311953} & \num{604485} & \num{544422} & \num{607113} & \num{593804} & \num{626618} & \num{583517} & \num{667466} & \num{757214}\\
Fraction of Women & \num{0.59} & \num{0.41} & \num{0.52} & \num{0.48} & \num{0.51} & \num{0.49} & \num{0.52} & \num{0.48} & \num{0.46} & \num{0.54}\\
\bottomrule
\end{tabular}}
\end{table}

\begin{table}[htbp!]
\centering
\caption{Descriptive Statistics of Upper Firm Classes}
\label{tab:upperclasses}
\resizebox{\linewidth}{!}{
\begin{tabular}{lrrrrrrrrrr}
\toprule
\multicolumn{1}{c}{ } & \multicolumn{2}{c}{class 6} & \multicolumn{2}{c}{class 7} & \multicolumn{2}{c}{class 8} & \multicolumn{2}{c}{class 9} & \multicolumn{2}{c}{class 10} \\
\cmidrule(l{3pt}r{3pt}){2-3} \cmidrule(l{3pt}r{3pt}){4-5} \cmidrule(l{3pt}r{3pt}){6-7} \cmidrule(l{3pt}r{3pt}){8-9} \cmidrule(l{3pt}r{3pt}){10-11}
\multicolumn{1}{c}{ } & \multicolumn{1}{c}{Female} & \multicolumn{1}{c}{Male} & \multicolumn{1}{c}{Female} & \multicolumn{1}{c}{Male} & \multicolumn{1}{c}{Female} & \multicolumn{1}{c}{Male} & \multicolumn{1}{c}{Female} & \multicolumn{1}{c}{Male} & \multicolumn{1}{c}{Female} & \multicolumn{1}{c}{Male} \\
\cmidrule(l{3pt}r{3pt}){2-2} \cmidrule(l{3pt}r{3pt}){3-3} \cmidrule(l{3pt}r{3pt}){4-4} \cmidrule(l{3pt}r{3pt}){5-5} \cmidrule(l{3pt}r{3pt}){6-6} \cmidrule(l{3pt}r{3pt}){7-7} \cmidrule(l{3pt}r{3pt}){8-8} \cmidrule(l{3pt}r{3pt}){9-9} \cmidrule(l{3pt}r{3pt}){10-10} \cmidrule(l{3pt}r{3pt}){11-11}
& (1) & (2) & (3) & (4) & (5) & (6) & (7) & (8) & (9) & (10) \\
\midrule
Number of Firms & \num{70196} & \num{70196} & \num{46157} & \num{46157} & \num{25680} & \num{25680} & \num{17077} & \num{17077} & \num{7762} & \num{7762}\\
Firms with $\geq$10 Workers & \num{35405} & \num{35405} & \num{24479} & \num{24479} & \num{14465} & \num{14465} & \num{11627} & \num{11627} & \num{5347} & \num{5347}\\
Firms with $\geq$50 Workers & \num{8915} & \num{8915} & \num{7124} & \num{7124} & \num{4946} & \num{4946} & \num{4855} & \num{4855} & \num{2448} & \num{2448}\\
Mean Firm Size & \num{40} & \num{40} & \num{55} & \num{55} & \num{81} & \num{81} & \num{110} & \num{110} & \num{148} & \num{148}\\
Median Firm Size & \num{10} & \num{10} & \num{10} & \num{10} & \num{12} & \num{12} & \num{20} & \num{20} & \num{20} & \num{20}\\
\addlinespace
Dropout & \num{0.19} & \num{0.31} & \num{0.12} & \num{0.23} & \num{0.07} & \num{0.12} & \num{0.02} & \num{0.05} & \num{0.01} & \num{0.02}\\
High School Graduates & \num{0.57} & \num{0.53} & \num{0.52} & \num{0.51} & \num{0.40} & \num{0.43} & \num{0.15} & \num{0.23} & \num{0.08} & \num{0.11}\\
Some College & \num{0.24} & \num{0.16} & \num{0.36} & \num{0.26} & \num{0.54} & \num{0.44} & \num{0.83} & \num{0.73} & \num{0.91} & \num{0.87}\\
\addlinespace
Age ($<$30) & \num{0.41} & \num{0.37} & \num{0.39} & \num{0.36} & \num{0.39} & \num{0.35} & \num{0.40} & \num{0.33} & \num{0.31} & \num{0.25}\\
Age 31-50 & \num{0.50} & \num{0.49} & \num{0.52} & \num{0.51} & \num{0.53} & \num{0.53} & \num{0.52} & \num{0.54} & \num{0.58} & \num{0.60}\\
Age ($\geq$51) & \num{0.08} & \num{0.12} & \num{0.08} & \num{0.11} & \num{0.07} & \num{0.10} & \num{0.06} & \num{0.11} & \num{0.09} & \num{0.13}\\
\addlinespace
Primary Sector & \num{0.01} & \num{0.01} & \num{0.00} & \num{0.00} & \num{0.00} & \num{0.00} & \num{0.00} & \num{0.00} & \num{0.00} & \num{0.01}\\
Manufacturing & \num{0.29} & \num{0.36} & \num{0.21} & \num{0.34} & \num{0.14} & \num{0.29} & \num{0.12} & \num{0.19} & \num{0.22} & \num{0.27}\\
Construction & \num{0.01} & \num{0.02} & \num{0.01} & \num{0.02} & \num{0.01} & \num{0.02} & \num{0.02} & \num{0.02} & \num{0.01} & \num{0.02}\\
Trade & \num{0.23} & \num{0.25} & \num{0.16} & \num{0.21} & \num{0.12} & \num{0.15} & \num{0.08} & \num{0.08} & \num{0.13} & \num{0.13}\\
Services & \num{0.47} & \num{0.36} & \num{0.61} & \num{0.43} & \num{0.72} & \num{0.53} & \num{0.78} & \num{0.70} & \num{0.63} & \num{0.58}\\
\addlinespace
Scientific and Liberal Arts & \num{0.09} & \num{0.06} & \num{0.14} & \num{0.10} & \num{0.23} & \num{0.17} & \num{0.32} & \num{0.31} & \num{0.36} & \num{0.36}\\
Technicians & \num{0.13} & \num{0.12} & \num{0.19} & \num{0.15} & \num{0.23} & \num{0.20} & \num{0.13} & \num{0.17} & \num{0.14} & \num{0.17}\\
Administrative & \num{0.34} & \num{0.18} & \num{0.34} & \num{0.19} & \num{0.33} & \num{0.21} & \num{0.36} & \num{0.23} & \num{0.29} & \num{0.19}\\
Managers & \num{0.04} & \num{0.05} & \num{0.04} & \num{0.05} & \num{0.06} & \num{0.08} & \num{0.11} & \num{0.14} & \num{0.16} & \num{0.21}\\
Traders & \num{0.21} & \num{0.20} & \num{0.16} & \num{0.16} & \num{0.10} & \num{0.11} & \num{0.06} & \num{0.06} & \num{0.04} & \num{0.04}\\
Rural & \num{0.00} & \num{0.01} & \num{0.00} & \num{0.01} & \num{0.00} & \num{0.00} & \num{0.00} & \num{0.00} & \num{0.00} & \num{0.00}\\
Factory & \num{0.19} & \num{0.39} & \num{0.12} & \num{0.34} & \num{0.06} & \num{0.23} & \num{0.02} & \num{0.09} & \num{0.01} & \num{0.03}\\
\addlinespace
Mean experience (years) & \num{4.264} & \num{4.687} & \num{4.664} & \num{5.089} & \num{5.049} & \num{5.567} & \num{5.367} & \num{6.218} & \num{5.889} & \num{6.418}\\
Mean Log-Wage & \num{1.990} & \num{2.162} & \num{2.258} & \num{2.399} & \num{2.575} & \num{2.729} & \num{2.986} & \num{3.201} & \num{3.497} & \num{3.749}\\
Variance of Log-Wage & \num{0.210} & \num{0.253} & \num{0.302} & \num{0.340} & \num{0.377} & \num{0.463} & \num{0.374} & \num{0.470} & \num{0.416} & \num{0.479}\\
Worker-years observations & \num{1263982} & \num{1561044} & \num{1168436} & \num{1361439} & \num{990776} & \num{1080347} & \num{819258} & \num{1051482} & \num{472259} & \num{680152}\\
Number of Workers & \num{675731} & \num{823446} & \num{587173} & \num{686629} & \num{449740} & \num{499661} & \num{334985} & \num{427467} & \num{179633} & \num{251992}\\
Fraction of Women & \num{0.45} & \num{0.55} & \num{0.46} & \num{0.54} & \num{0.48} & \num{0.52} & \num{0.44} & \num{0.56} & \num{0.41} & \num{0.59}\\
\bottomrule
\end{tabular}}
\end{table}

\begin{table}[htbp!]
\centering
\caption{Descriptive Statistics of Lower Worker Types}
\resizebox{\linewidth}{!}{
\begin{tabular}{lrrrrrrrrrr}
\toprule
\multicolumn{1}{c}{ } & \multicolumn{2}{c}{type 1} & \multicolumn{2}{c}{type 2} & \multicolumn{2}{c}{type 3} & \multicolumn{2}{c}{type 4} & \multicolumn{2}{c}{type 5} \\
\cmidrule(l{3pt}r{3pt}){2-3} \cmidrule(l{3pt}r{3pt}){4-5} \cmidrule(l{3pt}r{3pt}){6-7} \cmidrule(l{3pt}r{3pt}){8-9} \cmidrule(l{3pt}r{3pt}){10-11}
\multicolumn{1}{c}{ } & \multicolumn{1}{c}{Female} & \multicolumn{1}{c}{Male} & \multicolumn{1}{c}{Female} & \multicolumn{1}{c}{Male} & \multicolumn{1}{c}{Female} & \multicolumn{1}{c}{Male} & \multicolumn{1}{c}{Female} & \multicolumn{1}{c}{Male} & \multicolumn{1}{c}{Female} & \multicolumn{1}{c}{Male} \\
\cmidrule(l{3pt}r{3pt}){2-2} \cmidrule(l{3pt}r{3pt}){3-3} \cmidrule(l{3pt}r{3pt}){4-4} \cmidrule(l{3pt}r{3pt}){5-5} \cmidrule(l{3pt}r{3pt}){6-6} \cmidrule(l{3pt}r{3pt}){7-7} \cmidrule(l{3pt}r{3pt}){8-8} \cmidrule(l{3pt}r{3pt}){9-9} \cmidrule(l{3pt}r{3pt}){10-10} \cmidrule(l{3pt}r{3pt}){11-11}
& (1) & (2) & (3) & (4) & (5) & (6) & (7) & (8) & (9) & (10) \\
\midrule
Number of Firms & \num{142765} & \num{115262} & \num{195800} & \num{167929} & \num{244655} & \num{237647} & \num{160499} & \num{173160} & \num{170689} & \num{213850}\\
Firms with $\geq$10 Workers & \num{18551} & \num{18551} & \num{37529} & \num{37529} & \num{73454} & \num{73454} & \num{33732} & \num{33732} & \num{54100} & \num{54100}\\
Firms with $\geq$50 Workers & \num{2262} & \num{2262} & \num{5428} & \num{5428} & \num{12445} & \num{12445} & \num{5354} & \num{5354} & \num{9744} & \num{9744}\\
Mean Firm Size & \num{7} & \num{7} & \num{9} & \num{9} & \num{14} & \num{14} & \num{9} & \num{9} & \num{13} & \num{13}\\
Median Firm Size & \num{2} & \num{2} & \num{3} & \num{3} & \num{4} & \num{4} & \num{3} & \num{3} & \num{3} & \num{3}\\
\addlinespace
Dropout & \num{0.44} & \num{0.46} & \num{0.35} & \num{0.41} & \num{0.30} & \num{0.38} & \num{0.28} & \num{0.35} & \num{0.16} & \num{0.33}\\
High School Graduates & \num{0.51} & \num{0.49} & \num{0.57} & \num{0.52} & \num{0.61} & \num{0.54} & \num{0.57} & \num{0.51} & \num{0.58} & \num{0.53}\\
Some College & \num{0.05} & \num{0.05} & \num{0.08} & \num{0.07} & \num{0.09} & \num{0.07} & \num{0.15} & \num{0.14} & \num{0.26} & \num{0.13}\\
Age ($<$30) & \num{0.40} & \num{0.52} & \num{0.46} & \num{0.55} & \num{0.45} & \num{0.49} & \num{0.41} & \num{0.41} & \num{0.44} & \num{0.39}\\
Age 31-50 & \num{0.47} & \num{0.33} & \num{0.44} & \num{0.33} & \num{0.46} & \num{0.38} & \num{0.48} & \num{0.44} & \num{0.48} & \num{0.48}\\
Age ($\geq$51) & \num{0.11} & \num{0.13} & \num{0.08} & \num{0.11} & \num{0.08} & \num{0.12} & \num{0.09} & \num{0.13} & \num{0.07} & \num{0.12}\\
\addlinespace
Primary Sector & \num{0.04} & \num{0.06} & \num{0.02} & \num{0.04} & \num{0.02} & \num{0.03} & \num{0.02} & \num{0.04} & \num{0.01} & \num{0.03}\\
Manufacturing & \num{0.17} & \num{0.17} & \num{0.20} & \num{0.20} & \num{0.23} & \num{0.24} & \num{0.20} & \num{0.24} & \num{0.21} & \num{0.27}\\
Construction & \num{0.01} & \num{0.01} & \num{0.01} & \num{0.01} & \num{0.01} & \num{0.01} & \num{0.01} & \num{0.01} & \num{0.01} & \num{0.02}\\
Trade & \num{0.21} & \num{0.29} & \num{0.29} & \num{0.34} & \num{0.33} & \num{0.34} & \num{0.26} & \num{0.29} & \num{0.26} & \num{0.28}\\
\addlinespace
Services & \num{0.57} & \num{0.47} & \num{0.47} & \num{0.41} & \num{0.42} & \num{0.38} & \num{0.51} & \num{0.42} & \num{0.52} & \num{0.41}\\
\addlinespace
Scientific and Liberal Arts & \num{0.01} & \num{0.01} & \num{0.02} & \num{0.02} & \num{0.02} & \num{0.02} & \num{0.04} & \num{0.04} & \num{0.07} & \num{0.04}\\
Technicians & \num{0.04} & \num{0.05} & \num{0.04} & \num{0.06} & \num{0.05} & \num{0.06} & \num{0.08} & \num{0.07} & \num{0.14} & \num{0.11}\\
Administrative & \num{0.27} & \num{0.21} & \num{0.34} & \num{0.22} & \num{0.36} & \num{0.19} & \num{0.37} & \num{0.18} & \num{0.42} & \num{0.19}\\
Managers & \num{0.01} & \num{0.01} & \num{0.01} & \num{0.02} & \num{0.02} & \num{0.03} & \num{0.03} & \num{0.05} & \num{0.04} & \num{0.03}\\
Traders & \num{0.47} & \num{0.38} & \num{0.38} & \num{0.36} & \num{0.33} & \num{0.34} & \num{0.31} & \num{0.28} & \num{0.18} & \num{0.23}\\
Rural & \num{0.04} & \num{0.07} & \num{0.02} & \num{0.04} & \num{0.02} & \num{0.04} & \num{0.02} & \num{0.04} & \num{0.00} & \num{0.02}\\
Factory & \num{0.16} & \num{0.26} & \num{0.18} & \num{0.29} & \num{0.20} & \num{0.32} & \num{0.16} & \num{0.34} & \num{0.14} & \num{0.37}\\
\addlinespace
Mean experience (years) & \num{2.552} & \num{2.510} & \num{2.645} & \num{2.530} & \num{3.034} & \num{2.993} & \num{3.517} & \num{3.765} & \num{4.190} & \num{4.230}\\
Mean Log-Wage & \num{1.395} & \num{1.433} & \num{1.519} & \num{1.550} & \num{1.624} & \num{1.692} & \num{1.770} & \num{1.955} & \num{2.016} & \num{2.028}\\
Variance of Log-Wage & \num{0.032} & \num{0.037} & \num{0.052} & \num{0.062} & \num{0.061} & \num{0.174} & \num{0.248} & \num{0.463} & \num{0.057} & \num{0.054}\\
Worker-years observations & \num{742040} & \num{429828} & \num{1233122} & \num{834399} & \num{2216912} & \num{1829454} & \num{951776} & \num{999212} & \num{1378256} & \num{1865438}\\
Number of Workers & \num{634290} & \num{383994} & \num{1006451} & \num{710878} & \num{1517927} & \num{1330372} & \num{808545} & \num{854727} & \num{990969} & \num{1316173}\\
Fraction of Women & \num{0.63} & \num{0.37} & \num{0.60} & \num{0.40} & \num{0.55} & \num{0.45} & \num{0.49} & \num{0.51} & \num{0.42} & \num{0.58}\\
\bottomrule
\end{tabular}}
\end{table}

\begin{table}[htbp!]
\centering
\caption{Descriptive Statistics of Upper Worker Types}
\resizebox{\linewidth}{!}{
\begin{tabular}{lrrrrrrrrrr}
\toprule
\multicolumn{1}{c}{ } & \multicolumn{2}{c}{type 6} & \multicolumn{2}{c}{type 7} & \multicolumn{2}{c}{type 8} & \multicolumn{2}{c}{type 9} & \multicolumn{2}{c}{type 10} \\
\cmidrule(l{3pt}r{3pt}){2-3} \cmidrule(l{3pt}r{3pt}){4-5} \cmidrule(l{3pt}r{3pt}){6-7} \cmidrule(l{3pt}r{3pt}){8-9} \cmidrule(l{3pt}r{3pt}){10-11}
\multicolumn{1}{c}{ } & \multicolumn{1}{c}{Female} & \multicolumn{1}{c}{Male} & \multicolumn{1}{c}{Female} & \multicolumn{1}{c}{Male} & \multicolumn{1}{c}{Female} & \multicolumn{1}{c}{Male} & \multicolumn{1}{c}{Female} & \multicolumn{1}{c}{Male} & \multicolumn{1}{c}{Female} & \multicolumn{1}{c}{Male} \\
\cmidrule(l{3pt}r{3pt}){2-2} \cmidrule(l{3pt}r{3pt}){3-3} \cmidrule(l{3pt}r{3pt}){4-4} \cmidrule(l{3pt}r{3pt}){5-5} \cmidrule(l{3pt}r{3pt}){6-6} \cmidrule(l{3pt}r{3pt}){7-7} \cmidrule(l{3pt}r{3pt}){8-8} \cmidrule(l{3pt}r{3pt}){9-9} \cmidrule(l{3pt}r{3pt}){10-10} \cmidrule(l{3pt}r{3pt}){11-11}
& (1) & (2) & (3) & (4) & (5) & (6) & (7) & (8) & (9) & (10) \\
\midrule
Number of Firms & \num{132151} & \num{167588} & \num{87574} & \num{115373} & \num{51576} & \num{62286} & \num{29157} & \num{39610} & \num{14021} & \num{21175}\\
Firms with $\geq$10 Workers & \num{40967} & \num{40967} & \num{32009} & \num{32009} & \num{15816} & \num{15816} & \num{11785} & \num{11785} & \num{6470} & \num{6470}\\
Firms with $\geq$50 Workers & \num{7489} & \num{7489} & \num{6674} & \num{6674} & \num{3481} & \num{3481} & \num{2998} & \num{2998} & \num{1676} & \num{1676}\\
Mean Firm Size & \num{13} & \num{13} & \num{15} & \num{15} & \num{15} & \num{15} & \num{21} & \num{21} & \num{22} & \num{22}\\
Median Firm Size & \num{3} & \num{3} & \num{3} & \num{3} & \num{2} & \num{2} & \num{3} & \num{3} & \num{3} & \num{3}\\
\addlinespace
Dropout & \num{0.12} & \num{0.26} & \num{0.05} & \num{0.18} & \num{0.03} & \num{0.08} & \num{0.01} & \num{0.04} & \num{0.00} & \num{0.01}\\
High School Graduates & \num{0.43} & \num{0.48} & \num{0.28} & \num{0.40} & \num{0.14} & \num{0.25} & \num{0.07} & \num{0.15} & \num{0.04} & \num{0.07}\\
Some College & \num{0.45} & \num{0.26} & \num{0.67} & \num{0.42} & \num{0.83} & \num{0.67} & \num{0.92} & \num{0.82} & \num{0.96} & \num{0.92}\\
\addlinespace
Age ($<$30) & \num{0.38} & \num{0.32} & \num{0.35} & \num{0.29} & \num{0.28} & \num{0.24} & \num{0.19} & \num{0.17} & \num{0.08} & \num{0.08}\\
Age 31-50 & \num{0.53} & \num{0.54} & \num{0.56} & \num{0.57} & \num{0.62} & \num{0.61} & \num{0.68} & \num{0.66} & \num{0.73} & \num{0.69}\\
\addlinespace
Age ($\geq$51) & \num{0.07} & \num{0.13} & \num{0.07} & \num{0.12} & \num{0.09} & \num{0.13} & \num{0.11} & \num{0.15} & \num{0.16} & \num{0.21}\\
\addlinespace
Primary Sector & \num{0.01} & \num{0.02} & \num{0.00} & \num{0.01} & \num{0.01} & \num{0.01} & \num{0.00} & \num{0.01} & \num{0.00} & \num{0.01}\\
Manufacturing & \num{0.16} & \num{0.28} & \num{0.15} & \num{0.29} & \num{0.16} & \num{0.28} & \num{0.18} & \num{0.29} & \num{0.18} & \num{0.26}\\
Construction & \num{0.01} & \num{0.02} & \num{0.01} & \num{0.02} & \num{0.01} & \num{0.02} & \num{0.02} & \num{0.02} & \num{0.01} & \num{0.02}\\
Trade & \num{0.19} & \num{0.23} & \num{0.14} & \num{0.18} & \num{0.12} & \num{0.14} & \num{0.11} & \num{0.13} & \num{0.12} & \num{0.13}\\
Services & \num{0.64} & \num{0.45} & \num{0.69} & \num{0.50} & \num{0.71} & \num{0.55} & \num{0.69} & \num{0.56} & \num{0.67} & \num{0.58}\\
\addlinespace
Scientific and Liberal Arts & \num{0.17} & \num{0.09} & \num{0.31} & \num{0.18} & \num{0.42} & \num{0.31} & \num{0.44} & \num{0.38} & \num{0.38} & \num{0.36}\\
Technicians & \num{0.22} & \num{0.15} & \num{0.22} & \num{0.20} & \num{0.15} & \num{0.19} & \num{0.11} & \num{0.17} & \num{0.07} & \num{0.09}\\
Administrative & \num{0.35} & \num{0.19} & \num{0.29} & \num{0.19} & \num{0.22} & \num{0.15} & \num{0.19} & \num{0.13} & \num{0.15} & \num{0.10}\\
Managers & \num{0.06} & \num{0.06} & \num{0.08} & \num{0.08} & \num{0.12} & \num{0.13} & \num{0.20} & \num{0.20} & \num{0.38} & \num{0.41}\\
Traders & \num{0.13} & \num{0.16} & \num{0.07} & \num{0.11} & \num{0.08} & \num{0.08} & \num{0.05} & \num{0.05} & \num{0.02} & \num{0.03}\\
Rural & \num{0.00} & \num{0.01} & \num{0.00} & \num{0.00} & \num{0.00} & \num{0.00} & \num{0.00} & \num{0.00} & \num{0.00} & \num{0.00}\\
Factory & \num{0.07} & \num{0.33} & \num{0.03} & \num{0.24} & \num{0.01} & \num{0.12} & \num{0.01} & \num{0.06} & \num{0.00} & \num{0.02}\\
\addlinespace
Mean experience (years) & \num{4.943} & \num{5.413} & \num{5.684} & \num{6.145} & \num{6.201} & \num{6.661} & \num{7.124} & \num{7.294} & \num{8.331} & \num{8.247}\\
Mean Log-Wage & \num{2.347} & \num{2.400} & \num{2.789} & \num{2.765} & \num{3.196} & \num{3.272} & \num{3.624} & \num{3.662} & \num{4.157} & \num{4.226}\\
Variance of Log-Wage & \num{0.200} & \num{0.202} & \num{0.074} & \num{0.080} & \num{0.232} & \num{0.213} & \num{0.080} & \num{0.085} & \num{0.103} & \num{0.131}\\
Worker-years observations & \num{1058198} & \num{1467569} & \num{901644} & \num{1216454} & \num{486706} & \num{666147} & \num{373527} & \num{612301} & \num{161052} & \num{362669}\\
Number of Workers & \num{724998} & \num{996335} & \num{529627} & \num{728982} & \num{315684} & \num{437215} & \num{216089} & \num{351364} & \num{91719} & \num{200886}\\
Fraction of Women & \num{0.42} & \num{0.58} & \num{0.43} & \num{0.57} & \num{0.42} & \num{0.58} & \num{0.38} & \num{0.62} & \num{0.31} & \num{0.69}\\
\bottomrule
\end{tabular}}
\end{table}

\begin{table}[htp!]
\centering
\begin{threeparttable}
\centering
\caption{Wage Levels for Males and Females under Different Scenarios}
\label{tab:base_wages}
\begin{tabular}{l cccccccc}
\hline
& \multicolumn{2}{c}{Baseline} & \multicolumn{2}{c}{Separable} & \multicolumn{2}{c}{Constant} & \multicolumn{2}{c}{Constant Firm} \\
& \multicolumn{2}{c}{} & \multicolumn{2}{c}{Market} & \multicolumn{2}{c}{Returns} & \multicolumn{2}{c}{Allocation} \\
\cmidrule(lr){2-3} \cmidrule(lr){4-5} \cmidrule(lr){6-7} \cmidrule(lr){8-9}
Group  & Female & Male & Female & Male & Female & Male & Female & Male \\
& (1) & (2) & (3) & (4) & (5) & (6) & (7) & (8) \\
\hline
All & 2.10 & 2.33 & 2.11 & 2.30 & 2.16 & 2.25 & 2.19 & 2.22 \\
\hline
\multicolumn{9}{l}{\textit{Education}} \\
No HS & 1.62 & 1.92 & 1.65 & 1.92 & 1.72 & 1.84 & 1.77 & 1.80 \\
HS & 1.90 & 2.14 & 1.92 & 2.13 & 1.97 & 2.07 & 2.01 & 2.04 \\
College & 2.92 & 3.27 & 2.87 & 3.17 & 2.94 & 3.10 & 3.00 & 3.03 \\
\hline
\multicolumn{9}{l}{\textit{Age}} \\
<30 & 1.97 & 2.06 & 1.98 & 2.06 & 2.01 & 2.03 & 2.01 & 2.03 \\
31-50 & 2.18 & 2.49 & 2.18 & 2.44 & 2.25 & 2.38 & 2.30 & 2.33 \\
50> & 2.09 & 2.42 & 2.08 & 2.36 & 2.15 & 2.29 & 2.21 & 2.24 \\
\hline
\multicolumn{9}{l}{\textit{Firm Size}} \\
Firms <10 & 1.66 & 1.78 & 1.65 & 1.77 & 1.70 & 1.72 & 1.70 & 1.73 \\
Firms 10-50 & 1.77 & 1.90 & 1.77 & 1.90 & 1.82 & 1.85 & 1.83 & 1.84 \\
Firms 51> & 1.78 & 1.99 & 1.79 & 1.96 & 1.83 & 1.91 & 1.86 & 1.88 \\
\hline
\multicolumn{9}{l}{\textit{Occupations}} \\
Hotel and Restaurants & 1.62 & 1.74 & 1.62 & 1.74 & 1.71 & 1.67 & 1.67 & 1.70 \\
Engineers \& Economists & 2.91 & 3.30 & 2.88 & 3.23 & 2.98 & 3.13 & 3.03 & 3.08 \\
Managers & 3.03 & 3.35 & 2.90 & 3.12 & 2.96 & 3.05 & 2.99 & 3.02 \\
\bottomrule
\end{tabular}
\begin{tablenotes}
\small
\item \textit{Notes:} \textsuperscript{1}All values represent base wages in log scale. \textsuperscript{2}Baseline is observed wages. \textsuperscript{3}Separable Market assumes no complementarity in worker-firm interactions. \textsuperscript{4}Constant Returns equalizes means and variances of realized worker-firm interactions. \textsuperscript{5}Constant Firm Allocation equalizes firm-specific probabilities.
\end{tablenotes}
\end{threeparttable}
\end{table}

\clearpage



\end{appendices}

\end{document}